\documentclass[11pt]{article}

\usepackage[margin=1in]{geometry}
\usepackage{setspace}
\usepackage[T1]{fontenc}

\usepackage{amsmath}
\usepackage{amssymb}
\usepackage{amsthm}
\usepackage{bm}
\usepackage{dsfont}

\usepackage{newtxtext}
\usepackage{newtxmath}

\usepackage{algorithm}
\usepackage{algpseudocode}
\usepackage{tikz}

\usepackage{natbib}
\bibpunct[, ]{(}{)}{,}{a}{}{,}

\usepackage[hidelinks]{hyperref}

\usepackage{color}
\usepackage[normalem]{ulem}

\newcommand{\calY}{\mathcal{Y}}
\newcommand{\E}{\mathbb{E}}

\newtheorem{theorem}{Theorem}
\newtheorem{proposition}{Proposition}
\newtheorem{lemma}{Lemma}
\newtheorem{corollary}{Corollary}

\theoremstyle{definition}
\newtheorem{assumption}{Assumption}
\newtheorem{example}{Example}

\theoremstyle{remark}
\newtheorem{remark}{Remark}

\let\amsthmproof\proof
\let\endamsthmproof\endproof

\renewcommand{\proof}[1]{\amsthmproof[#1]}
\renewcommand{\endproof}{\endamsthmproof}

\newcommand{\Halmos}{\qedhere}

\usepackage{authblk}

\title{\textbf{Computing Stationary Equilibria in
Measure-Dependent Markov Systems}}

\author[1]{Jing Dong}
\author[2]{Bar Light}
\author[3]{Xin T. Tong}

\affil[1]{Graduate School of Business, Columbia University\\
\texttt{jing.dong@gsb.columbia.edu}}

\affil[2]{Business School and Institute of Operations Research and Analytics,
National University of Singapore\\
\texttt{barlight@nus.edu.sg}}

\affil[3]{Department of Mathematics and Institute of Operations Research and Analytics,
National University of Singapore\\
\texttt{xin.t.tong@nus.edu.sg}}

\date{}

\begin{document}

\maketitle

\begin{abstract}
Many stochastic systems in operations and economics exhibit feedback
between their long-run state distribution and the transition law governing
their dynamics. In this paper, we develop a computational framework for
stationary equilibria in such measure-dependent Markov systems when this
feedback operates through a finite-dimensional aggregate. We show that the
original stationary-equilibrium problem can be reduced to a finite-dimensional
self-consistency equation, separating steady-state analysis of the underlying
Markov system from equilibrium computation. We use properties of the resulting
self-consistency map to guide the choice among fixed-point iteration, relaxed
fixed-point iteration, and minimization of the fixed-point residual. The last
approach requires derivatives of the self-consistency map, which are typically
unavailable in closed form. We therefore develop finite-time infinitesimal
perturbation analysis estimators for these derivatives, with error bounds that
separate Monte Carlo error from finite-time bias. We illustrate the framework
through a strategic $G/G/c$ queue and an opinion-dynamics model, showing how
different structural properties lead naturally to different
equilibrium-computation methods.
\end{abstract}


\bigskip

\section{Introduction}\label{sec:Intro}

Many stochastic systems in operations and economics have dynamics (transition laws) that are shaped by feedback from their own long-run behavior. 
In service systems, customers may decide when or whether to arrive based on the distribution of waiting times. In social-learning models, agents update their beliefs based on population opinions. In models of productivity, wealth, or knowledge
capital, individual states evolve through both idiosyncratic shocks and aggregate market conditions. In these settings, the distribution of the state is not only an output of the stochastic dynamics; it is also an input to the transition law. A stationary analysis must therefore satisfy two requirements simultaneously: the distribution must be invariant under
the transition law it induces, and the transition law used to generate the distribution must be the one implied by that same distribution.

This self-consistency requirement is central for both performance analysis and decision making. In queueing systems, stationary distributions are used to evaluate congestion, service reliability, and capacity decisions. When customers respond to delay information, however, the waiting-time distribution is not only an output of the system; it also influences future arrival behavior, which in turn determines the future waiting-time distribution. This type of stationary-equilibrium formulation appears naturally in queueing games, where customer decisions and long-run waiting-time distributions are jointly determined; see, e.g., \cite{armony2009impact,xu2013supermarket,ata2018equilibrium}. Opinion-dynamics and social-learning models have an analogous structure. The distribution of beliefs is typically the object used to study disagreement, polarization, and comparative statics with respect to social influence or information parameters at equilibrium. At the same time, this distribution also enters the dynamics: population beliefs determine the social signals agents observe, and individual belief updates then reshape the population distribution; 
see \cite{dasaratha2023learning,levy2024stationary}.

More broadly, whenever individual behavior and population-level outcomes interact, the long-run object of interest is often a stationary equilibrium in which individual behavior and the induced population distribution are mutually consistent, rather than the stationary distribution of a fixed Markov chain. This perspective is also common in game-theoretic and mean-field settings, where individual behavior depends on a conjectured long-run population state and, in equilibrium, generates the same state that agents conjectured; see, e.g., \cite{adlakha2013mean,acemoglu2015robust,light2022mean}. We study this self-consistency problem directly through measure-dependent Markov systems. This formulation allows us to capture feedback effects between individual actions and aggregate outcomes without committing to a particular strategic or game-theoretic framework. 
Such stationary equilibria are generally characterized only implicitly, and analytical solutions are often unavailable except in highly stylized models. This makes computation an essential part of the analysis rather than merely an implementation detail: 
an unstable or inaccurate computational procedure can produce misleading performance estimates and comparative statics, and ultimately lead to poor operational or economic decisions.

A natural starting point is to imitate the standard steady-state computation for ordinary Markov chains and iterate the nonlinear Markov operator directly. For an ordinary Markov chain, the transition kernel is fixed, so repeatedly applying the same kernel provides a natural route toward its stationary distribution. In a measure-dependent Markov system, by contrast, the transition kernel itself depends on the distribution being evolved. Each update therefore changes not only the distribution but also the transition law governing the next update, creating an endogenous feedback loop between the state distribution and its dynamics. As a result, convergence properties that are familiar from ordinary Markov chains need not carry over to the nonlinear setting. 

Our starting point is a structural reduction. We focus on the common and practically important setting in which the transition kernel depends on the current distribution only through a finite-dimensional aggregator $H(\mu)\in A\subset\mathbb{R}^d$. For example, in a strategic queue, the arrival law may depend on the mean or variability of the waiting-time; in
an opinion-dynamics model, the updating rule may depend on type-specific mean beliefs. For each fixed aggregate value $a\in A$, freezing the feedback yields an ordinary Markov chain with transition kernel $P_a$. When this frozen chain admits a unique stationary distribution, denoted by $\pi_a$, the original stationary-equilibrium problem reduces to the finite-dimensional self-consistency equation
\[
a=H(\pi_a)=:G(a).
\]
Thus, $G(a)$ records the aggregate generated in steady state when the system is operated under the conjectured aggregate $a$, and an equilibrium is precisely a value $a^\star$ for which the generated aggregate coincides with the conjectured one, i.e., $a^\star=G(a^\star)$.
Once such a fixed point $a^\star$ is found, the corresponding stationary equilibrium of the original measure-dependent system is $\pi_{a^\star}$. This reduction separates two distinct tasks: analyzing the long-run behavior of the frozen Markov chain for a given aggregate and solving the resulting finite-dimensional equilibrium problem over the aggregate parameter.

The reduction also makes clear why no single computational method is appropriate for all models. 
The self-consistency map $G$ can have very different structures across applications, so we view the fixed-point computation as a model-dependent algorithm-selection problem. When $G$ is contractive, the classical Picard iteration, i.e., $a_{k+1}=G(a_k)$, is natural and enjoys geometric convergence. When $G$ is order-preserving on an interval, monotone fixed-point iteration converges even in the absence of a metric contraction. When $G$ is not contractive but satisfies some one-sided contraction condition, relaxed fixed-point iteration can stabilize the computation. When the available regularity is only log-Lipschitz, an adaptive relaxation schedule for relaxed fixed-point iteration ensures convergence despite the lack of global Lipschitz continuity. Finally, when $G$ is differentiable but fixed-point iterations are unstable or ineffective, we minimize the residual objective 
\[\ell(a)=\frac{1}{2}\|a-G(a)\|^2\] 
using gradient-based methods (see Section \ref{sec:fixed_point} for details). 

The structural properties of $G$ therefore guide the choice of algorithm. In some
models, these properties can be verified directly from primitives. For example, monotonicity of the transition kernel can imply monotonicity of $G$, and perturbation bounds for the frozen
stationary distributions can imply Lipschitz or log-Lipschitz regularity. These results help connect model primitives to algorithm selection. The more challenging case arises when the fixed-point iteration is unreliable or when the
structural properties needed for its convergence are difficult to verify. In such settings, a powerful alternative is to solve the equilibrium equation through the residual objective $\ell(a)$. This approach recasts equilibrium computation as an optimization problem. It is less dependent on contraction or monotonicity properties
of $G$, and the residual norm \(\|a-G(a)\|\) provides a natural measure of the violation of the equilibrium condition. The tradeoff is that gradient-based optimization requires estimating the Jacobian \(\nabla G(a)\). Since \(G(a)\) is defined implicitly through the stationary distribution of the frozen Markov chain, this Jacobian is typically unavailable in closed form and must be estimated.

Our main computational contribution is to develop and analyze infinitesimal perturbation analysis (IPA) estimators for
\(\nabla G(a)\). Rather than perturbing the parameter and resimulating the system, IPA
differentiates the simulated sample path directly, evolving a sensitivity process alongside the state process. This avoids finite-difference bias and can recover the standard Monte Carlo error rate under suitable regularity conditions. Because the Jacobian must be estimated repeatedly along the outer optimization path, finite-time error control is crucial. We therefore decompose the IPA error into Monte Carlo error and finite-time bias, establish stability conditions for the sensitivity process, and provide both standard mean-squared error bounds under Lyapunov moment control and a truncated estimator that remains valid under weaker fractional stability.

We demonstrate the breadth of the framework through examples with distinct structural features. These examples illustrate not only the modeling flexibility of measure-dependent Markov systems, but also how model structure guides computation. 

The main example we study is a strategic $G/G/c$ queue in which the interarrival-time distribution depends on the stationary waiting-time distribution through an aggregate congestion signal. In this model, the frozen system is a standard Kiefer--Wolfowitz workload recursion, while the equilibrium condition captures the feedback from congestion to arrivals. We establish the existence of an invariant distribution and derive perturbation bounds for the stationary workload distribution under changes in the aggregate parameter. These bounds imply log-Lipschitz regularity of the self-consistency map under mild conditions. In cases where the additional one-sided stability condition can be verified, relaxed fixed-point iteration provides a provably convergent algorithm. In more complex queueing settings, where such structural conditions may be difficult to check, we develop a finite-time IPA bound for estimating \(\nabla G(a)\), thereby supporting residual-gradient equilibrium computation.

We also study an opinion-dynamics model with multiple agent types. Each agent updates her belief using a combination of her current opinion, a type-dependent weighted average of population mean beliefs, and an idiosyncratic shock. In this setting, the aggregate variable is the vector of within-type mean beliefs. Under natural monotonicity assumptions, e.g., higher own beliefs
or higher social signals lead to higher future beliefs, we show that the stationary response map is order-preserving. This structure allows us to apply monotone fixed-point iteration on an order interval and obtain convergence to an equilibrium.



\subsection{Literature Review}


\textbf{Measure-dependent Markov systems.} Our paper is related to the literature on Markov systems whose transition law depends on the current distribution of the process, also called nonlinear Markov chains, including McKean--Vlasov dynamics, mean-field models, and other distribution-dependent stochastic systems; see \citet{kolokoltsov2010nonlinear}. Much of this literature studies existence, uniqueness, and convergence to invariant distributions. For example, \citet{butkovsky2014ergodic} establishes existence, uniqueness, and ergodicity under contraction-type conditions. As discussed in the introduction, such conditions can fail even in simple settings (see examples Section~\ref{sec:problem-setup}). Relatedly, \citet{neumann2023nonlinear} shows that nonlinear Markov chains can exhibit nonstandard long-run behavior, including failure of convergence to an invariant distribution.

A closely related paper is \citet{light2025invariant}, which studies nonlinear Markov chains in which the current distribution affects the transition law through a one-dimensional aggregator. That paper uses monotonicity arguments to establish existence and uniqueness results and develops an iterative computational method based on the aggregator structure. Our work allows multidimensional aggregators and takes a broader algorithmic perspective. We reduce the stationary-equilibrium problem to the finite-dimensional self-consistency equation $a=G(a)$ and study how different structural properties of $G$ lead to different computational methods. 

\textbf{Strategic queueing and congestion-responsive behavior.}
Our strategic queueing application is related to the literature on queues in which customers respond strategically to congestion or delay information, creating feedback between customer behavior and steady-state system performance \citep{hassin2016rational,hassin2003queue}. For example, \citet{mandelbaum2000model} study rational abandonment in invisible queues, where customers' patience decisions and the waiting-time distribution are jointly determined in equilibrium. \citet{armony2009impact} examine customer responses to delay announcements in many-server queues with abandonment, while \citet{xu2013supermarket} analyze a mean-field queueing game in which customers strategically choose how many queues to sample. \citet{guo2013strategic} study strategic customer behavior in a multiserver system in which decisions may depend on long-run waiting-time information, and \citet{ata2018equilibrium} consider endogenous abandonment when customers respond to the virtual waiting-time distribution. 
Most of this literature characterizes behavioral equilibria within particular queueing models. 

Most of this literature characterizes behavioral equilibria within particular queueing models. Closer to our computational focus, \citet{ravner2024stochastic} develop a simulation-based stochastic-approximation method for computing symmetric Nash equilibria in finite-action queueing games. Our approach is complementary: we model feedback via a measure-dependent transition law and characterize stationarity as a self-consistency condition for the stationary response map. This formulation extends naturally to multidimensional congestion aggregates and applies beyond strategic queueing. In our (G/G/c) application, customer behavior induces an aggregate-dependent arrival law. For the associated frozen Kiefer--Wolfowitz recursion, we establish proper ergodicity and stationary perturbation bounds, enabling gradient-based computation of stationary equilibria.

\textbf{Infinitesimal perturbation analysis and steady-state sensitivity.}
Our gradient-estimation analysis is related to the extensive literature on IPA and steady-state sensitivity analysis. One strand studies the validity of steady-state IPA: whether finite-horizon dynamics can be differentiated pathwise, whether differentiation can be interchanged with expectation, and whether the resulting derivatives identify the derivative of the stationary performance measure; see, e.g., \citet{HeidelbergerCaoZazanisSuri1988,Glasserman1991,Glasserman1992,RheeGlynn2023,Heidergott2026}.
Particularly related to our approach, \citet{flynn2017forward} develops a forward-sensitivity framework for contracting stochastic systems. By augmenting the state process with its pathwise sensitivity, Flynn establishes stationary derivative representations under contraction and regularity conditions on the underlying dynamics. We likewise use a state-sensitivity representation, but study the finite-horizon error relative to its stationary counterpart. Moreover, rather than relying on contraction of the underlying stochastic system, we formulate stability directly in terms of the random affine sensitivity recursion and allow the derivative dynamics to be locally expansive, provided that products of random Jacobians contract in a suitable fractional moment.

A second line of work studies the statistical efficiency and convergence rates of steady-state derivative estimators. In particular, \citet{l1992convergence} analyzes how the bias and variance of IPA and related estimators translate into convergence rates under different allocations of simulation effort across horizon length and replications, while \citet{l1994convergence} compares convergence rates of IPA and finite differences for truncated-horizon steady-state estimation. More recently, \citet{wang2026unbiased} develop unbiased estimators of stationary derivatives based on coupling and Poisson-equation representations.
Our work is complementary to this literature. Rather than taking the decay rate of the transient IPA error as an input to a statistical convergence analysis, we derive that decay from stability properties of the underlying state-sensitivity dynamics. We then translate this stability into finite-time error bounds that separate Monte Carlo error from finite-time bias.


\section{Problem Setup}\label{sec:problem-setup}
Let $(S, \mathcal B(S))$ be a Polish space equipped with its Borel $\sigma$-algebra, and let $\mathcal P(S)$ denote the set of probability measures on $S$. Let $A\subseteq \mathbb R^d$ be a measurable set of aggregate values. For $a\in A$, let $P_a$ be a Markov kernel on $S$; i.e., $P_a(x,\cdot)\in \mathcal P(S)$ for each $x\in S$, and the map
$(x,a)\mapsto P_a(x,B)$
is measurable for every $B\in \mathcal B(S)$.

We consider systems in which the transition law depends on the current distribution only through a finite-dimensional aggregate. Specifically, let
$H:\mathcal P(S)\to A$
be a measurable aggregator. Given a current distribution $\mu\in\mathcal P(S)$, the next-period transition kernel is $P_{H(\mu)}$. The corresponding nonlinear Markov operator $\mathcal T:\mathcal P(S)\to\mathcal P(S)$ is
\[
\mathcal T\mu(B)
=
\int_S P_{H(\mu)}(x,B)\,\mu(dx),
\qquad B\in\mathcal B(S).
\]
Equivalently, using the standard notation $\mu P_a(B):=\int_S P_a(x,B)\,\mu(dx)$, we write
$\mathcal T\mu=\mu P_{H(\mu)}$. The aggregator $H$ captures the finite-dimensional summary of the distribution that affects the transition law. A common specification is an expectation-type aggregator. Let $f:S\to\mathbb R^d$ be a measurable state-level summary function. Then one may define
\[
H(\mu)=\int_S f(x)\,\mu(dx),
\]
whenever the integral is well defined. This formulation includes means, moments, type-specific averages, congestion measures, and other distributional statistics used in the applications below. When moment conditions are needed to ensure that $H$ is finite, we restrict attention to the corresponding subset of $\mathcal P(S)$ without changing the notation.

A probability measure $\mu^\star\in\mathcal P(S)$ is called a stationary distribution or stationary equilibrium of the measure-dependent Markov system if it is a fixed point of the nonlinear Markov operator:
\[\mu^\star=\mathcal T\mu^\star,
\mbox{ or equivalently, }
\mu^\star=\mu^\star P_{H(\mu^\star)}.
\]
Thus, at stationarity, the distribution $\mu^\star$ is invariant under the frozen Markov kernel indexed by the aggregate generated by $\mu^\star$, i.e., $H(\mu^\star)$.

A natural approach to computing a stationary equilibrium is to iterate the nonlinear operator,
\[
\mu_{k+1}=\mathcal T\mu_k=\mu_k P_{H(\mu_k)}.
\]
This procedure resembles the standard iteration used for ordinary Markov chains. The key difference is that, here, the transition kernel is updated endogenously at each step. As a result, the sequence $\{\mu_k\}$ may fail to converge, even when stationary equilibria exist. The following example illustrates this point.

\begin{example} \label{eg1}
Let $S=\{1,2\}$ and $H(\mu)=\mu(\{2\})$. 
For $H(\mu)=a\in[0,1]$, the transition kernel takes the form
\[
P_a=
\begin{pmatrix}
a & 1-a\\
a & 1-a
\end{pmatrix},
\]
with the corresponding stationary distribution $\pi_a=(a,1-a)$.
If $a_k=\mu_k(\{2\})$, then the nonlinear iteration 
$\mu_{k+1}=\mu_k P_{H(\mu_k)}$ satisfies
$a_{k+1}=1-a_k$.
Therefore, unless $a_0=1/2$, the sequence alternates between $a_0$ and 
$1-a_0$ and does not converge, even though the nonlinear chain has the 
stationary distribution $(1/2,1/2)$.
\end{example}

Our approach avoids this difficulty by separating the Markov-chain computation from the equilibrium computation. For each fixed aggregate value $a\in A$, consider the ordinary Markov chain with transition kernel $P_a$. When this frozen chain admits a unique stationary distribution, denote it by $\pi_a$, i.e., $\pi_a=\pi_a P_a$.
The nonlinear stationary problem can then be reduced to a finite-dimensional self-consistency equation, i.e., $G(a):=H(\pi_a)$ for $a\in A$.
A value $a^\star\in A$ is self-consistent if
\[
a^\star=G(a^\star).
\]
In that case, $\pi_{a^\star}$ is a stationary equilibrium of the nonlinear Markov chain, since
\[
\pi_{a^\star}P_{H(\pi_{a^\star})}
=
\pi_{a^\star}P_{a^\star}
=
\pi_{a^\star},
\]
where the first equality uses the self-consistency condition $a^\star=H(\pi_{a^\star})$ and the second uses the stationarity of $\pi_{a^\star}$ for the frozen kernel $P_{a^\star}$.
Conversely, if $\mu^\star$ is a stationary distribution of the nonlinear chain and the frozen kernel $P_{H(\mu^\star)}$ has a unique stationary distribution, then
$\mu^\star=\pi_{H(\mu^\star)}$
and $a^\star:=H(\mu^\star)$ satisfies $a^\star=G(a^\star)$.
Therefore, under the uniqueness of the frozen stationary distributions, computing a stationary equilibrium of the nonlinear Markov chain is equivalent to solving the finite-dimensional fixed-point problem.

Examples~\ref{eg1c} and~\ref{eg2} illustrate two related points. 
In Example~\ref{eg1c}, relaxed iteration on the self-consistency map \(G\) converges whereas the direct nonlinear Markov iteration does not converge. In Example~\ref{eg2}, direct iteration can fail even when \(G\) itself is a contraction.

\begin{example}[Example \ref{eg1} continued.] \label{eg1c}
Recall that $\pi_a=(a,1-a)$. Then 
\[
G(a)=H(\pi_a)=\pi_a(\{2\})=1-a.
\]
Consider a relaxed fixed point iteration with a fixed $\delta\in(0,1)$:
\[
a_{k+1}=(1-\delta)a_k + \delta G(a_k)
=\delta+(1-2\delta)a_{k}.
\]
Then,
\[
a_{k+1}-1/2=(1-2\delta) (a_k-1/2),
\]
which further implies that
\[
|a_{k+1}-1/2|=|1-2\delta|^{k+1}|a_0-1/2|.
\]
i.e., $a_k$ converges to $1/2$.
\end{example}

\begin{example} \label{eg2}
Let $S=\{1,2\}$ and $H(\mu)=\mu(\{2\})$. 
For $H(\mu)=a\in[0,1]$, the transition kernel takes the form
\[
P_a=
\begin{pmatrix}
1-r(a)g(a) & r(a)g(a)\\
r(a)(1-g(a)) & 1-r(a)(1-g(a))
\end{pmatrix},
\quad \mbox{ where } g(a)=\frac{3}{4}-\frac{1}{2}a, \quad r(a)=\frac{3}{2}-2\left(a-\frac{1}{2}\right)^2.
\]
The corresponding stationary distribution is $\pi_a=(1-g(a),g(a))$.
If $a_k=\mu_k(\{2\})$, for $a_0=\frac{1}{2}+\frac{1}{2\sqrt{3}}$ the nonlinear Markov chain iteration $\mu_{k+1}=\mu_k P_{H(\mu_k)}$ alternates between $\frac{1}{2}+\frac{1}{2\sqrt{3}}$ and $\frac{1}{2}-\frac{1}{2\sqrt{3}}$, even though the nonlinear chain has the 
stationary distribution $(1/2,1/2)$. More generally, as long as $a_0\neq 1/2$, the iterates will approach the two-cycle: $\frac{1}{2}+\frac{1}{2\sqrt{3}}, \frac{1}{2}-\frac{1}{2\sqrt{3}}, \frac{1}{2}+\frac{1}{2\sqrt{3}}, \frac{1}{2}-\frac{1}{2\sqrt{3}}, \cdots.$

Meanwhile, $G(a)=H(\pi_a)=g(a)$. Then, fixed-point iteration on $G$ yields
\[
a_{k+1}-\frac{1}{2}=-\frac{1}{2}\left(a_k-\frac{1}{2}\right).
\]
Thus, under fixed-point iteration on $G$, $a_k\rightarrow 1/2$ geometrically.
\end{example}

The examples point to a basic difference between the two computations. The reduced fixed-point problem uses the stationary response of the frozen chain: given $a_k$, it computes the invariant distribution \(\pi_{a_k}\) of \(P_{a_k}\) and then updates the aggregate
to $a_{k+1}=H(\pi_{a_k})=G(a_k)$. Direct iteration instead evolves the distribution one step at a time: $\mu_{k+1}=\mu_kP_{H(\mu_k)}$, with each updated distribution determining the next transition kernel. Consequently, even when every frozen chain converges to its invariant distribution, direct iteration feeds transient changes back into the dynamics before stationarity is reached. As the examples show, this feedback can generate oscillations even when the stationary response map \(G\) is stable. Thus, all computational methods developed below are for the reduced fixed-point problem.

\section{Fixed Point Problem} \label{sec:fixed_point}
We impose the following assumption in our algorithmic development.
\begin{assumption} \label{ass:stationarity}
For every $a\in A$, the transition kernel $P_a$ admits a unique stationary distribution $\pi_a$. Moreover, for every initial distribution $\mu$, $\mu P_a^n \Rightarrow \pi_a$ as $n\rightarrow\infty$.
\end{assumption}
Assumption~\ref{ass:stationarity} is a standard ergodicity condition on the frozen Markov chains. It holds under several classical sets of sufficient conditions. For example, if the state space $S$ is compact and the kernel $P_a$ is Feller, the
Krylov--Bogoliubov theorem guarantees the existence of an invariant
probability measure \citep{hairer2021convergence}.
A global Doeblin minorization condition additionally ensures
uniqueness and uniform geometric convergence in total variation
\citep{meyn1993markov}.
Alternatively, if for each fixed $a$, the kernel $P_a$ defines a Markov chain that is aperiodic, $\psi$-irreducible, and positive Harris recurrent, then it admits a unique stationary distribution $\pi_a$, and its transition probabilities converge to $\pi_a$.
in total
variation from every initial state\citep{meyn1993markov}.

Under Assumption \ref{ass:stationarity}, in this section, we review algorithms to solve $a^\star=G(a^\star)$. Throughout, we assume $A\subset \mathbb{R}^d$ is nonempty and closed. 

The map $G: A\rightarrow A$ is typically available only implicitly: evaluating \(G(a)\) requires computing
or approximating the stationary distribution of the frozen chain at parameter \(a\).
Thus, the overall computation has two levels: an \emph{inner} step that evaluates
\(G(a)\) or the Jacobian matrix $\nabla G(a)$, and an \emph{outer} step that solves the fixed-point equation \(G(a)=a\). We consider three outer iterations:
\begin{align}
&\text{\bf Picard iteration:}~
a_{k+1}=G(a_k), \label{eq:picard-main} \\
&\text{\bf Relaxed fixed-point iteration (RFPI):}~
a_{k+1}=(1-\delta_k)a_k+\delta_k G(a_k), \label{eq:relaxed-main} \\
&\text{\bf Gradient descent on the residual:}~
a_{k+1}=a_k-\eta_k \nabla \ell(a_k),
~
\ell(a):=\frac12\|a-G(a)\|^2.
\label{eq:gd-main}
\end{align}
The outer algorithms we consider are classical, which allows us to leverage existing convergence results to give practical guidance on algorithm selection. We provide a high-level summary here and defer the technical details to Appendix \ref{app:fixed_point_main}. Table~\ref{tab:algorithm-recommendation} serves as a concise decision guide, focusing on the recommended method under each set of structural conditions.

\begin{table}[ht]
\centering
\small
\renewcommand{\arraystretch}{1.12}
\begin{tabular}{p{0.5\linewidth} p{0.4\linewidth}}
\hline
Condition on \(G\) & Recommended algorithm\\ 
\hline
Contraction &
Picard iteration \\

Monotone on an invariant order interval &
Picard iteration \\

Nonexpansive &
RFPI \\

Lipschitz + incremental one-sided condition &
RFPI with constant step size\\

Log-Lipschitz + incremental one-sided condition &
RFPI with decreasing step size  \\

Smooth residual objective \(\ell\) with favorable geometry &
Gradient descent on \(\ell\) \\
\hline
\end{tabular}
\caption{Recommended outer algorithms for solving \(G(a)=a\). Detailed assumptions and precise guarantees are provided in Appendix~\ref{app:fixed_point}.}
\label{tab:algorithm-recommendation}
\end{table}

Among the three, gradient descent \eqref{eq:gd-main} has the least  restrictive requirement for its convergence. This is because gradient descent in general can find a stationary point where $\nabla l(a^*)\approx 0$ \citep{nesterov2013introductory}, which is an approximate fixed point if $I-\nabla G(a^*)$ is not singular. Stronger optimization assumptions like Polyak--\L ojasiewicz condition can further guarantee linear convergence (Proposition~\ref{prop:gd_pl}). On the other hand, finding the gradient $\nabla l(a_k)$ for nonlinear Markov Chain is not a trivial computational task. 
Section~\ref{sec:gradient} develops IPA estimators for \(\nabla G(a)\), which is the key input needed for residual-gradient methods. 

In comparison, both Picard iteration \eqref{eq:picard-main} and RFPI \eqref{eq:relaxed-main} can be implemented as long as we can estimate the invariant measure $\pi_a$ for a fixed parameter $a$. This is a classical task that can be achieved by simulating $P_a$ directly. Note that Picard iteration \eqref{eq:picard-main} is RFPI with 
$\delta_k\equiv 1$, so the latter has better flexibility if we tune the step size $\delta_k$ according to $G$'s regularity. In particular, for any $c\in [0,1)$ and $L>0$, RFPI converges as long as $G$ is $L$-Lipschitz and satisfies the incremental one-sided condition
\[
    \langle G(a)-G(b),a-b\rangle
    \le c\|a-b\|^2,
    \quad \forall a,b\in A,
\]
for some $c<1$.
Propositions~\ref{prop:non-expansive},~\ref{prop:Lip}, and ~\ref{prop:Log-Lip}, provide convergence guarantees for RFPI under different regularity conditions on $G$. 
By contrast, linear convergence of Picard iteration requires the stronger contraction condition
(Proposition~\ref{prop:contraction}). 

Finally, Picard iteration may remain applicable even when RFPI or gradient-based methods are not well defined. This situation can arise when the aggregate space $A$ is not convex, so convex combinations such as \((1-\delta_k)a_k+\delta_k G(a_k)\) are not necessarily feasible, or when $A$ is discrete, so derivatives of $G$ are not meaningful. A typical example is an aggregate space that is a finite or countable subset of the integers. In such settings, order structure can provide an alternative route to convergence. If $A$ is equipped with a partial order and $G$ is order-preserving with respect to this order, then monotone Picard iteration can be used (Proposition \ref{prop:monotone}).

As discussed above, Picard or relaxed fixed-point iterations can be justified when the self-consistency map \(G\) satisfies suitable contraction, monotonicity, or incremental one-sided conditions. These conditions hold in many important models, including some of the applications studied later in the paper. However, they should not be taken for granted. The following simple AR(1) example illustrates how such structural properties may fail, and why alternative computational approaches may be needed.

\begin{example}
\label{ex:ar1-unstable-equilibrium}
Consider the scalar AR(1) model
\[
    X_{t+1}=aX_t+b+\xi_{t+1},
    \quad a\in(-1,1),
\]
where \(\{\xi_t\}\) are i.i.d. with \(\mathbb E[\xi_t]=0\), and let
$H(\mu)=\int x\,\mu(dx)$.
For each fixed \(a\in(-1,1)\), the chain admits a unique stationary distribution with mean
\[
    G(a)=H(\pi_a)=\frac{b}{1-a}.
\]
Take \(b=2/9\). The fixed-point equation \(a=G(a)\) becomes
\[
    a^2-a+\frac{2}{9}=0,
\]
whose two solutions are
$a^*_1=\frac13$ and $a^*_2=\frac23$.
Moreover, note that
$G'(a^*_1)=\frac12<1$ and $G'(a^*_2)=2>1$.
The fixed point \(a^*_1=1/3\) is locally stable for Picard iteration, while
\(a^*_2=2/3\) is unstable.

Residual minimization treats the two equilibria differently from Picard iteration. The
residual objective is
\[
    \ell(a)
    =
    \frac12|a-G(a)|^2
    =
    \frac12
    \left(
        a-\frac{2/9}{1-a}
    \right)^2 .
\]
Both \(a^*_1\) and \(a^*_2\) are global minimizers of \(\ell\), because both solve
\(a=G(a)\). In addition,
\[
    \ell''(a^*_i)
    =
    \bigl(1-G'(a^*_i)\bigr)^2>0,
    \quad i=1,2.
\]
Thus, the Picard-unstable equilibrium \(a^*_2\) is still a locally well-conditioned minimum
of the residual objective. 
\end{example}

\section{Gradient Estimation} \label{sec:gradient}
For gradient descent on the residual $\ell(a)=\frac12\|a-G(a)\|^2$, the gradient takes the form
\[\nabla \ell(a)
=
(I-\nabla G(a))^\top(a-G(a)).\]
Thus, beyond estimating the stationary response $G(a)$ itself, a residual-gradient method requires an estimate of the Jacobian $\nabla G(a)$. This is the main computational difficulty: in the models of interest, $G(a)=H(\pi_a)$ is defined through the invariant distribution of the frozen Markov chain and is rarely available in closed form.

For the pathwise analysis in this section, we focus on a Euclidean state space $\mathcal S\subseteq\mathbb R^m$ and an expectation-type aggregator
\[
    H(\mu)=\int_{\mathcal S} f(x)\,\mu(dx),
    \qquad f: \mathcal S\to\mathbb R^d.
\]

A standard approach to estimating \(\nabla G(a)\) is to use central finite differences. For \(a\in\mathbb R^d\), step size \(h>0\), and coordinate vectors $e_i$, the $(j,i)$ entry is approximated by
\[
    \widehat \nabla^{\mathrm{FD}}_{ji}G(a)
    :=
    \frac{\widehat G_{j,N}(a+he_i)-\widehat G_{j,N}(a-he_i)}{2h},
\]
where each $\widehat G_N(\cdot)$ is a Monte Carlo estimate of the corresponding stationary response.
Suppose, for this comparison, that $G$ is three times continuously differentiable near $a$, that the two function estimates have variance of order $N^{-1}$, and that any finite-horizon bias in estimating $G$ is ignored. Then the central-difference bias is $O(h^2)$, whereas the variance is $O((Nh^2)^{-1})$. Hence
\[
    \operatorname{MSE}\bigl(\widehat \nabla^{\mathrm{FD}}_{ji}G(a)\bigr)
    =
    O\!\left(h^4+\frac{1}{Nh^2}\right).
\]
Balancing the two terms gives $h\asymp N^{-1/6}$ and an MSE of order $N^{-2/3}$, rather than the usual Monte Carlo rate $N^{-1}$. Moreover, a coordinatewise central difference requires two additional steady-state simulations per parameter coordinate.

This calculation is the standard independent-noise benchmark. Carefully designed common random numbers can reduce the variance in particular models, but the amount of cancellation is model dependent and itself relies on a sufficiently regular coupling of nearby systems.  Infinitesimal perturbation analysis (IPA) instead differentiates a single simulated path and therefore avoids the finite-difference bias and the need to resimulate each coordinate perturbation.

Suppose that, for a fixed aggregate value \(a\), the frozen Markov chain
\(\{X_t(a)\}_{t\ge 0}\) admits the sample-path representation
\begin{equation}
\label{eq:random-map-representation}
    X_{t+1}(a)
    =
    \Phi\bigl(X_t(a),a,U_{t+1}\bigr),
\end{equation}
where $(U_t:t\ge1)$ are i.i.d. exogenous random variables.  
We initially write the recursion under ordinary differentiability assumptions. In discrete-event models, global smoothness is usually unnecessary: it is enough that the map be differentiable almost surely along the simulated path, which can often be established by non-atomicity of the primitive random variables.

Let
$Z_t(a):=\nabla_aX_t(a)\in\mathbb R^{m\times d}$
denote the pathwise sensitivity of the state with respect to the aggregate parameter. 
Differentiating \eqref{eq:random-map-representation} gives
\[
    Z_{t+1}(a)
    =
    \nabla_x\Phi(X_t(a),a,U_{t+1})Z_t(a)
    +
    \nabla_a\Phi(X_t(a),a,U_{t+1}).
\]
Define $Y_t(a):=(X_t(a),a,U_{t+1})\in\mathcal{Y}$,
and write
$\Gamma(Y_t(a)):=\nabla_x\Phi(X_t(a),a,U_{t+1})$,
$\xi(Y_t(a)):=\nabla_a\Phi(X_t(a),a,U_{t+1})$. Then the sensitivity process satisfies the random affine recursion
\begin{equation}
\label{eq:sensitivity-recursion}
    Z_{t+1}(a)
    =
    \Gamma(Y_t(a))Z_t(a)+\xi(Y_t(a)).
\end{equation}
Thus, IPA augments the original frozen chain with an auxiliary sensitivity process that evolves along the same sample path.

\begin{example}\label{eg:AR1}
Consider an AR(1)-type frozen dynamic
\[
    X_{t+1}(a)
    =
    \rho(a)X_t(a)+\Psi(a,U_{t+1}),
\]
Here $\rho(a)$ captures persistence, while the innovation law may also depend on $a$ through its common-random-number representation $\Psi(a,U_{t+1})$. The sensitivity recursion is
\[
Z_{t+1}(a)=\rho(a)Z_t(a)+\rho'(a)X_t(a)+\nabla_a\Psi(a,U_{t+1}).
\]
This simple example already shows that stability of the state process does not automatically imply stability of its derivative: the derivative is governed by a second stochastic recursion with its own multiplicative dynamics.
\end{example}

For integers \(m\le n\), define the right-to-left product
$\Gamma_{m:n}:=
\Gamma(Y_{n-1})\cdots \Gamma(Y_m)$, with 
$\Gamma_{m:m}:=I$.
Iterating \eqref{eq:sensitivity-recursion} gives 
\begin{equation}
\label{eq:variation-of-constants}
    Z_t
    =
    \Gamma_{0:t}Z_0
    +
    \sum_{m=0}^{t-1}\Gamma_{m+1:t}\xi(Y_m).
\end{equation}
The random matrix products in \eqref{eq:variation-of-constants} are the central object in the analysis below.


We take the steady-state IPA identity as the starting point for our finite-time analysis. For a fixed $a\in A$, we assume the augmented process $(Y_t,Z_t)$ admits an invariant distribution $\Pi_a$, $G$ is differentiable at $a$, and 
\begin{equation}
\label{eq:steady-state-ipa-identity}
    \nabla G(a)
    =
    \mathbb E_{\Pi_a}\bigl[\nabla f(X)Z\bigr].
\end{equation}
For example, a sufficient route is to justify
\[
    \nabla G(a)
    =
    \nabla_a\lim_{t\to\infty}\mathbb E_\mu[f(X_t(a))] 
    =
    \lim_{t\to\infty}\mathbb E_\mu\bigl[\nabla f(X_t(a))Z_t(a)\bigr]
    =
    \mathbb E_{\Pi_a}\bigl[\nabla f(X)Z\bigr]
\]
by pathwise differentiability and uniform integrability. These interchange conditions are model specific and are verified directly in our applications. Our focus is on quantifying how accurately $\nabla G(a)$ can be estimated from finite-horizon simulation.

Define the terminal-time IPA estimator
\begin{equation}
\label{eq:ipa-estimator}
    \widehat g_{N,T}(a)
    :=
    \frac1N\sum_{i=1}^N
    \nabla f(X_T^{(i)}(a))Z_T^{(i)}(a),
\end{equation}
where \(\{(X_t^{(i)},Z_t^{(i)})\}_{i=1}^N\) are independent copies of the state-sensitivity process. Its error decomposes as
\begin{equation}
\label{eq:decomp}
    \widehat g_{N,t}(a)-\nabla G(a)
    =
    \underbrace{
    \left(
        \widehat g_{N,t}(a)
        -
        \mathbb E_\mu[\nabla f(X_t)Z_t]
    \right)
    }_{\text{Monte Carlo error}}
   +
    \underbrace{
    \left(
        \mathbb E_\mu[\nabla f(X_t)Z_t]
        -
        \mathbb E_{\Pi_a}[\nabla f(X)Z]
    \right)
    }_{\text{finite-time bias}}.
\end{equation}
The first term is controlled by moments of the IPA summand. The second depends on the convergence of the state-sensitivity process to stationarity.

The decomposition in \eqref{eq:decomp} highlights the main technical contribution of this section. Differentiating a finite-horizon recursion is straightforward once pathwise differentiability is available. The difficult step is to control the derivative as the horizon grows. Even when $Y_t$ is geometrically ergodic, the sensitivity may be amplified by occasional values of $\Gamma(Y_t)$ with norm greater than one. We therefore develop a stability theory for the random affine recursion \eqref{eq:sensitivity-recursion}, based on contraction of random Jacobian products on average, and then translate that stability into finite-time error bounds for IPA. The main idea behind our approach is finding an appropriate Lyapunov function based on a Poisson equation. A similar approach has been used for stability analysis of continuous stochastic processes \citep{majda2019simple}. 

\begin{remark}[Time-averaged IPA estimator]
A common implementation uses the post-burn-in average
\[
    \widehat g_{N,b,T}^{\mathrm{avg}}(a)
    :=
    \frac1N\sum_{i=1}^N
    \frac1{T-b}\sum_{t=b+1}^T
    \nabla f\bigl(X_t^{(i)}(a)\bigr)Z_t^{(i)}(a).
\]
This estimator typically uses simulation output more efficiently than a terminal observation. Its analysis additionally requires control of serial dependence along each trajectory, for example through summable covariance or mixing bounds for the augmented process. We use the terminal-time estimator because it cleanly separates cross-replication Monte Carlo error from finite-time bias; analogous time-average results can be obtained under the corresponding mixing assumptions.
\end{remark}


The following example uses an exactly solvable AR(1) model to illustrate the two components of the error decomposition in \eqref{eq:decomp}.

\begin{example}
\label{eg:ar1-finite-time-ipa}
Consider 
\[
    X_{t+1}(a)=\rho X_t(a)+a+\xi_{t+1},
    \qquad |\rho|<1,
\]
with $X_0(a)=0$, $\mathbb E[\xi_t]=0$, and $\mathbb E[\xi_t^2]<\infty$.
Take $f(x)=x^2$.
If $\sigma_\xi^2=\operatorname{Var}(\xi_t)$, then
\[
    G(a)
    = \mathbb{E}_{\pi_a}[X^2]
    =\frac{\sigma_\xi^2}{1-\rho^2}
    +
    \frac{a^2}{(1-\rho)^2},
    \qquad
    G'(a)=\frac{2a}{(1-\rho)^2}.
\]
The pathwise sensitivity satisfies
$Z_{t+1}(a)=\rho Z_t(a)+1$ with $Z_0(a)=0$, so
\[
    Z_T(a)=\frac{1-\rho^T}{1-\rho}.
\]
The terminal-time IPA estimator in \eqref{eq:ipa-estimator} becomes
\[
    \widehat g_{N,T}(a)
    =
    \frac1N\sum_{i=1}^N
    2X_T^{(i)}(a)Z_T(a).
\]
We can compute both terms in \eqref{eq:decomp} explicitly. Since
$\mathbb E[X_T(a)]=a(1-\rho^T)/(1-\rho)$
and $\sup_T\operatorname{Var}(X_T(a))<\infty$, the variance of the Monte Carlo term is $O(N^{-1})$ uniformly in $T$. Meanwhile
\[
    \mathbb E[\widehat g_{N,T}(a)]-G'(a)
    =
    \frac{2a}{(1-\rho)^2}
    \left((1-\rho^T)^2-1\right)
    =
    O(|\rho|^T).
\]
Consequently, there exists $C<\infty$, independent of $N$ and $T$, such that
\[
    \mathbb E\left[
        \left(\widehat J_{N,T}(a)-G'(a)\right)^2
    \right]
    \le
    C\left(\frac1N+|\rho|^{2T}\right).
\]
Thus, taking $T$ proportional to $\log N$ recovers the standard Monte Carlo MSE rate $O(N^{-1})$.
\end{example}
This simple example also illustrates the mechanism developed in the next section: contraction of the state-sensitivity dynamics controls the finite-time bias, while suitable moment stability provides the uniform variance bound needed for the standard Monte Carlo rate.

\subsection{Stability of the Sensitivity Process}
The representation \eqref{eq:variation-of-constants} shows why stability of the sensitivity process is more delicate than stability of the frozen state process. Stability of $Y_t$ controls the environment in which the derivative evolves, but it does not prevent the products $\Gamma_{m:n}$ from occasionally becoming large. The relevant requirement is therefore not uniform contraction at every state, but contraction of these products in a suitable average sense.

Throughout, vector and matrix norms are Euclidean and induced operator norms, respectively; for matrix-valued sensitivities, $\|\cdot\|$ may equivalently be read as the Frobenius norm, with only constants changing. Let $d_{\mathcal Y}$ be a distance on the state space of the driving process, and let $V:\mathcal Y\to[1,\infty)$ be a Lyapunov function. In a countable-state model, a useful choice is the weighted discrete metric
$d_{\mathcal Y}(y,\widetilde y)
    :=
    \mathbf 1\{y\ne\widetilde y\}
    \bigl(V(y)+V(\widetilde y)\bigr)$,
which accommodates coefficient maps that are not Lipschitz in an ordinary Euclidean metric.

\begin{assumption}[Core stability conditions for the sensitivity recursion]
\label{ass:ipa-stability}
There exist $\gamma>1$, $q_0>0$, and $C<\infty$ such that the following conditions hold.
\begin{enumerate}
    \item \textbf{Geometric coupling of the driving process.}
    There is a Markovian coupling $(Y_t,\widetilde Y_t)$ satisfying
    \[
        \mathbb E\left[
            d_{\mathcal Y}(Y_n,\widetilde Y_n)
            \mid
            Y_0=y,\widetilde Y_0=\widetilde y
        \right]
        \le
        C\gamma^{-n}d_{\mathcal Y}(y,\widetilde y)
        \quad \forall y,\widetilde y \in \cal Y, n\ge0.
    \]
    \item \textbf{Regularity of the IPA coefficients.}
    There exist $L_\Gamma,L_\xi<\infty$ such that for all $y,\widetilde y\in\mathcal Y$ 
    \[
        \|\Gamma(y)-\Gamma(\widetilde y)\|
        \le
        L_\Gamma d_{\mathcal Y}(y,\widetilde y),
        \qquad
        \|\xi(y)-\xi(\widetilde y)\|
        \le
        L_\xi d_{\mathcal Y}(y,\widetilde y).
    \]

    \item \textbf{Moment control for the additive sensitivity term.}
    For every initial state $y$,
    \[
        \sup_{t\ge0}
        \mathbb E\left[
            \|\xi(Y_t)\|^{q_0}
            \mid Y_0=y
        \right]
        \le
        C V(y)^{q_0}.
    \]

    \item \textbf{Fractional contraction of Jacobian products.}
    For every initial state $y$ and $n\ge0$,
    \[
        \mathbb E\left[
            \|\Gamma_{0:n}\|^{q_0}
            \mid Y_0=y
        \right]
        \le
        C\gamma^{-q_0n}V(y)^{q_0},
    \]
    and
    $\sup_{t\ge0}
        \mathbb E\left[V(Y_t)\mid Y_0=y\right]
        \le
        C V(y)$.
\end{enumerate}
\end{assumption}

Assumption~\ref{ass:ipa-stability} is formulated as a modular sufficient
condition that separates stability of the driving process from stability of
the differentiated recursion. Its components may be verified directly or
through model-specific mechanisms. In the strategic \(G/G/c\) queue in Section \ref{sec:ggc_queue}, for
example, we exploit the recurrent empty-system reset directly to obtain the
stronger second-moment stability needed for the unmodified IPA estimator.
Condition 1 is a standard geometric-ergodicity requirement expressed through coupling, while Conditions 2 and 3 are ordinary regularity and moment controls for the coefficient of the differentiated recursion. 
The fractional-product bound in
Condition~4 is the principal model-dependent requirement; the accompanying
bound on $V(Y_t)$ is a standard Lyapunov moment condition. The product bound
is automatic under the stronger uniform contraction condition
$\sup_y\|\Gamma(y)\|<1$, but it also allows $\|\Gamma(Y_t)\|>1$ along parts of a sample path, provided that the cumulative product contracts in a fractional moment. Section~\ref{sec:weakcontract} gives two concrete verification mechanisms: recurrent resets and negative mean-log growth. Separate geometric rates in Conditions 1 and 4 can be allowed; using a common $\gamma$ only simplifies notation.

Consider two copies $(Y_t,Z_t)$ and $(\widetilde Y_t,\widetilde Z_t)$ driven by the coupling in Assumption~\ref{ass:ipa-stability}, with $Z_0=\widetilde Z_0=0$.

\begin{theorem}[Fractional-moment contraction of the sensitivity process]
\label{thm:fractional-contraction}
Under Assumption~\ref{ass:ipa-stability}, for every $p\in\left(0,\min\left\{\frac{q_0}{4},\frac14\right\}\right]$
there exists \(C_p<\infty\) such that, for all \(T\ge1\) and all $y,\tilde y \in \calY$,
\[
    \mathbb E\left[
        \|Z_T-\widetilde Z_T\|^p
        \mid
        Y_0=y,\widetilde Y_0=\widetilde y
    \right]
    \le
   C_p T\gamma^{-pT}
    d_{\mathcal Y}(y,\widetilde y)^p
    V(\widetilde y)^p
    \bigl(1+V(y)^{2p}\bigr).
\]
\end{theorem}


A standard way to analyze steady-state IPA would be to assume directly that the augmented state--sensitivity process \((Y_t,Z_t)\) is geometrically ergodic. Theorem~\ref{thm:fractional-contraction}  
instead derives the quantitative
coupling estimate for the sensitivity component from separate conditions on
the driving process and the random Jacobian products. These two ingredients
can be verified through different model-specific arguments, without requiring
a direct ergodicity analysis of the full augmented process.
The theorem implies that on any set on which \(d_{\mathcal Y}\) and \(V\) are bounded,
\[
    \mathbb E\bigl[\|Z_T-\widetilde Z_T\|^p\bigr]
    =
    O\!\left(T\gamma^{-pT}\right).
\]

The use of a fractional moment substantially broadens the scope of the result. Average multiplicative contraction may produce a stable but heavy-tailed sensitivity process for which first or second moments do not exist (see Remark \ref{rem:fraction}). Theorem \ref{thm:fractional-contraction} still provides quantitative control in this regime. This control is sufficient for bounded and clipped IPA functionals; when combined with stronger Lyapunov moment bounds, it also yields the finite-time bias and mean-squared error guarantees for the raw estimator developed next.

\begin{remark}[Why fractional-moment stability is genuinely weaker]\label{rem:fraction}
Consider the scalar recursion
\[
    Z_{t+1}=A_tZ_t+1,
\]
where $(A_t)$ are i.i.d. and
$\mathbb P(A_t=0)=\mathbb P(A_t=10)=1/2$.
Then $\mathbb E[\log A_t]=-\infty$, with the convention $\log0=-\infty$, and
$\mathbb E[A_t^q]=10^q/2<1$
for every $q<\log2/\log10$. The stationary solution
\[
    Z_\infty
    =
    1+A_0+A_0A_{-1}+A_0A_{-1}A_{-2}+\cdots
\]
is finite almost surely and has a finite $q$th moment for each such $q$. It is nevertheless nonintegrable. Indeed, if $\mathbb E[Z_\infty]<\infty$, stationarity and the independence of $A_t$ from $Z_t$ would imply
\[
    \mathbb E[Z_\infty]
    =
    5\mathbb E[Z_\infty]+1,
\]
which is impossible because $Z_\infty\ge0$. Thus, contraction in a small fractional moment does not imply the first- or second-moment bounds needed for the usual MSE analysis of the unmodified IPA estimator.
\end{remark}

\subsection{IPA Gradient Estimation Analysis}
\label{sec:IPAgrad}

We now translate sensitivity stability into finite-time error bounds. There are two regimes. If the augmented process admits sufficiently strong moment control, the raw IPA estimator achieves the standard Monte Carlo rate. If only fractional stability and first-moment integrability are available, clipping yields a robust consistent estimator.

\subsubsection{The Lyapunov-Controlled Case}
Let $h_a(y,z):=\nabla f(x)z$, where $x$ denotes the state component of $y$. The following assumption
adds the moment and coupling controls needed for an $L^2$ analysis of the
IPA summand.

\begin{assumption}[Additional moment control for the IPA estimator]
\label{ass:lyap-untruncated-ipa}
For a fixed \(a\in A\), the following conditions hold.

\begin{enumerate}
    \item \textbf{Lyapunov drift for the augmented process.}
    There exists a measurable function
    \(W:\calY\times\mathbb R^{m\times d}\to[1,\infty)\), constants
    \(\rho\in(0,1)\) and \(b<\infty\), such that
    \begin{equation}\label{eq:Wbound01}
        \mathbb E\left[
            W(Y_1,Z_1)
            \mid
            Y_0=y,Z_0=z
        \right]
        \le
        \rho W(y,z)+b .
    \end{equation}
    The initial distribution \(\mu\) used for simulation satisfies
    $\mathbb E_\mu[W(Y_0,Z_0)]<\infty$.

    \item \textbf{Second-moment control of the IPA summand.}
    There exists \(C_h<\infty\) such that, for all \((y,z)\),
    \begin{equation} \label{eq:W_secondm}
        \|h_a(y,z)\|_F^2
        \le
        C_h\bigl(1+W(y,z)\bigr).
    \end{equation}

    \item \textbf{Fractional coupling of the IPA summand.}
    There exist $s\in(0,1]$, $\beta>1$, and $C<\infty$ such that, under a coupling of a transient copy initialized from $\mu$ and a stationary copy initialized from $\Pi_a$,
    \begin{equation}
    \label{eq:summand-contraction}
        \mathbb E\left[
            \|h_a(Y_t,Z_t)
              -h_a(\widetilde Y_t,\widetilde Z_t)\|_F^s
        \right]
        \le
        Ct\beta^{-t},
        \qquad t\ge1.
    \end{equation}
\end{enumerate}
\end{assumption}

The conditions in Assumption~\ref{ass:lyap-untruncated-ipa} have distinct
roles. Condition 1 provides uniform moment control for the augmented process.
Because \(\Pi_a\) is invariant, the same drift condition also implies $\mathbb E_{\Pi_a}[W(Y,Z)]\le b/(1-\rho)<\infty$.
Condition 2 is stated directly in terms of the IPA summand
\(h_a(y,z)=\nabla f(x)z\), which is the weakest form needed for the theorem.
A convenient sufficient condition is
\[
    \|\nabla f(x)\|_{\mathrm{op}}^4+\|z\|_F^4
    \le
    C_W\bigl(1+W(y,z)\bigr).
\]
When $\nabla f$ is uniformly bounded, it is enough instead to require
$\|z\|_F^2
    \le
C_W\bigl(1+W(y,z)\bigr)$.

Condition~3 is imposed directly on the IPA summand because this is the
quantity that enters the finite-time bias. This formulation preserves
generality in applications where $\nabla f$ is unbounded or fails to be
globally Lipschitz. At the same time, Condition~3 is related to the
preceding sensitivity-stability analysis. In particular, when
Assumption~\ref{ass:ipa-stability} holds, $\nabla f$ is uniformly bounded
and Lipschitz, and $W$ controls $\|z\|_F^2$ as above, the fractional
sensitivity contraction in
Theorem~\ref{thm:fractional-contraction}, together with the Jacobian-product
bound in Assumption~\ref{ass:ipa-stability}, implies
\eqref{eq:summand-contraction} under the corresponding integrability
conditions on the initial coupling; see Appendix~\ref{app:summand-contraction}.

\begin{theorem}
\label{thm:untruncated-ipa-lyap}
Suppose the steady-state IPA identity \eqref{eq:steady-state-ipa-identity} and Assumption~\ref{ass:lyap-untruncated-ipa} hold. Then 
\[
    \mathbb E\left[
        \|\widehat g_{N,t}(a)-\nabla G(a)\|^2
    \right]
    \le
    C\left(\frac1N+t\beta^{-t}\right) \mbox{ for some $C<\infty$.}
\]
\end{theorem}

Theorem \ref{thm:untruncated-ipa-lyap} implies if \(t_N\) is chosen as a sufficiently
large multiple of \(\log N\), then
\[
\mathbb E\left[\left\|\widehat g_{N,t_N}(a)-\nabla G(a)\right\|^2\right]=O(N^{-1}).
\]

\begin{remark}[Rate in replications versus total simulation work]
The $O(N^{-1})$ statement is expressed in the number of independent replications. With $t_N=O(\log N)$, the total number of simulated transitions is $O(N\log N)$, so the corresponding bound as a function of serial simulation work carries a logarithmic overhead. The replication-based formulation is nevertheless natural when trajectories are run in parallel and makes the comparison with ordinary Monte Carlo transparent.
\end{remark}

\subsubsection{The clipped IPA estimator}
If a second-moment Lyapunov bound is unavailable, the raw summand may have large or even infinite variance. The fractional contraction result can still be used after clipping the summand.

For \(M>0\), define the clipped summand
\[
    \mathcal T_M(g)
    :=
    \begin{cases}
        g, & \|g\|\le M,\\[3pt]
        M g/\|g\|, & \|g\|>M.
    \end{cases}
\]
This is the Euclidean projection onto the closed ball of radius $M$, so
$\|\mathcal T_M(g)\|\le M$ and
$\|\mathcal T_M(g)-\mathcal T_M(g')\|\le\|g-g'\|$.
The clipped IPA estimator is
\[
    \widehat g_{N,t}^{\,M}(a)
    :=
    \frac1N\sum_{i=1}^N
    \mathcal T_M(h_a(Y_t^{(i)},Z_t^{(i)})).
\]
The total error decomposes as
\[
\begin{aligned}
    \widehat g_{N,t}^{\,M}(a)-\nabla G(a)
    =
    &\underbrace{
    \left(
    \widehat g_{N,t}^{\,M}(a)
    -
    \mathbb E_\mu[\mathcal T_M(h_a(Y_t,Z_t))]
    \right)
    }_{\text{Monte Carlo error}}  
    +
    \underbrace{
    \left(
    \mathbb E_\mu[\mathcal T_M(h_a(Y_t,Z_t))]
    -
    \mathbb E_{\Pi_a}[\mathcal T_M(h_a(Y,Z))]
    \right)
    }_{\text{finite-time truncated bias}}  \\
    &+
    \underbrace{
    \left(
    \mathbb E_{\Pi_a}[\mathcal T_M(h_a(Y,Z))]
    -
    \mathbb E_{\Pi_a}[h_a(Y,Z)]
    \right)
    }_{\text{stationary truncation bias}} .
\end{aligned}
\]

The Monte Carlo error satisfies the standard bounded-summand estimate:
\[
    \mathbb E
    \left[
        \left\|
        \widehat g_{N,t}^{\,M}(a)
        -
        \mathbb E_\mu[\mathcal T_M(h_a(Y_t,Z_t))]
        \right\|^2
    \right]
    =
    \frac1N
    \mathbb E
    \left[
        \left\|
        \mathcal T_M(h_a(Y_t,Z_t))
        -
        \mathbb E[\mathcal T_M(h_a(Y_t,Z_t))]
        \right\|^2
    \right] 
    \le
    \frac{M^2}{N}.
\]

Assume the fractional summand-coupling condition
\eqref{eq:summand-contraction}.
The finite-time truncated bias is controlled by the fractional contraction
bound. Since \(\mathcal T_M\) is \(1\)-Lipschitz and bounded by \(M\),
$\|\mathcal T_M(g)-\mathcal T_M(g')\|\le (2M)^{1-s}\|g-g'\|^s$.
Therefore, using \eqref{eq:summand-contraction},
\[
\begin{aligned}
    \left\|
    \mathbb E[\mathcal T_M(h_a(Y_t,Z_t))]
    -
    \mathbb E_{\Pi_a}[\mathcal T_M(h_a(Y,Z))]
    \right\|
    &=
    \left\|
    \mathbb E[
        \mathcal T_M(h_a(Y_t,Z_t))
        -
        \mathcal T_M(h_a(\widetilde Y_t,\widetilde Z_t))
    ]
    \right\| \\
    &\le
    \mathbb E
    \left[
        \|\mathcal T_M(h_a(Y_t,Z_t))-\mathcal T_M(h_a(\widetilde Y_t,\widetilde Z_t))\|
    \right] 
    \le
    C M^{1-s}t\beta^{-t}.
\end{aligned}
\]

Finally, the stationary truncation bias satisfies
\[
    B_M
    :=
    \left\|
    \mathbb E_{\Pi_a}[\mathcal T_M(h_a(Y,Z))]
    -
    \mathbb E_{\Pi_a}[h_a(Y,Z)]
    \right\|  
    \le
    \mathbb E_{\Pi_a}
    \left[
        \|h_a(Y,Z)\|\mathbf 1\{\|h_a(Y,Z)\|>M\}
    \right].
\]
Thus \(B_M\to0\) as \(M\to\infty\) whenever
\(\mathbb E_{\Pi_a}\|h_a(Y,Z)\|<\infty\). Putting the three parts together, we have the following theorem.

\begin{theorem}
\label{thm:truncated-ipa}
Suppose \(h_a\in L^1(\Pi_a)\)  and the fractional summand contraction condition
\eqref{eq:summand-contraction} holds. Then there is a constant $C$ so that, for every \(M>0\),
\(N\ge1\), and \(t\ge1\), 
\[
    \mathbb E
    \left[
        \left\|
        \widehat g_{N,t}^{\,M}(a)-\nabla G(a)
        \right\|^2
    \right]
    \le
    C\left(
        \frac{M^2}{N}
        +
        M^{2(1-s)}t^2\beta^{-2t}
        +
        B_M^2
    \right).
\]
\end{theorem}

From Theorem \ref{thm:truncated-ipa}, we have if \(M=M_N\to\infty\) and \(t=t_N\to\infty\) are chosen so that $M_N/\sqrt{N}\to0$, $M_N^{1-s}t_N\beta^{-t_N}\to0$, and $B_{M_N}\to0$
then
\[
    \widehat g_{N,t_N}^{\,M_N}(a)\to \nabla G(a)
    \quad\text{in probability}.
\]

\begin{remark}[An explicit clipping-bias bound]
If, for some $\delta>0$,
$\mathbb E_{\Pi_a}\|h_a(Y,Z)\|^{1+\delta}<\infty$, then
\[
    B_M
    \le
    M^{-\delta}
    \mathbb E_{\Pi_a}\|h_a(Y,Z)\|^{1+\delta}.
\]
This converts Theorem~\ref{thm:truncated-ipa} into an explicit bias-variance tradeoff  and guides the joint choice of $M_N$ and $t_N$.
\end{remark}

\begin{remark}[Hard truncation versus clipping]
One could instead use
\[
    h_a(Y_t,Z_t)\mathbf 1\{\|h_a(Y_t,Z_t)\|\le M\}.
\]
Hard truncation has the same bounded-summand variance and a similar stationary tail bias. Radial clipping is preferable for the finite-time analysis because it is Lipschitz. Hard truncation is discontinuous at the boundary $M$ and generally requires an additional anti-concentration condition near that boundary, or a sequence of cutoffs chosen at continuity points of the relevant distribution.
\end{remark}

\subsection{Weak contraction on average}
\label{sec:weakcontract}
We give two sufficient mechanisms for verifying Condition~4 of
Assumption~\ref{ass:ipa-stability}: recurrent resets and negative
mean-log growth. The first exploits exact resets of the derivative recursion. The second quantifies a negative mean-log growth condition when exact resets are absent.

\paragraph{Reset-induced contraction.}
Suppose there are recurrent states at which an infinitesimal parameter perturbation no longer propagates, so that $\Gamma(Y_t)=0$. Such resets arise naturally in queueing, inventory, and other regenerative models: when the system enters a state whose next transition is locally insensitive to the parameter, the accumulated derivative is erased.

\begin{proposition}[Weak average contraction from recurrent resets]
\label{prop:reset}
Suppose $\|\Gamma(y)\|\le M_\Gamma<\infty$ for all $y\in\mathcal Y$, and define
$\tau_0:=\inf\{t\ge0:\|\Gamma(Y_t)\|=0\}$.
Assume that there exist $C_\tau<\infty$, $\rho\in(0,1)$, and $R:\mathcal Y\to[1,\infty)$ such that
\[
    \mathbb P(\tau_0\ge n\mid Y_0=y)
    \le
    C_\tau R(y)\rho^n,
    \qquad n\ge0.
\]
Then, for every $q>0$ satisfying $\rho M_\Gamma^q<1$, there exists $\gamma>1$ such that
\[
    \mathbb E\left[
        \|\Gamma_{0:n}\|^q
        \mid Y_0=y
    \right]
    \le
    CR(y)\gamma^{-qn}.
\]
In particular, if $R(y)\le CV(y)^q$ and
\[
    \sup_{t\ge0}\mathbb E[V(Y_t)\mid Y_0=y]\le CV(y),
\]
then Condition 4 of Assumption~\ref{ass:ipa-stability} holds with exponent $q$.
\end{proposition}

Because $\rho<1$, a sufficiently small positive $q$ always satisfies $\rho M_\Gamma^q<1$ whenever $M_\Gamma<\infty$. Thus, the substantive part of Proposition~\ref{prop:reset} is the geometric tail bound for the time to the next reset, which can often be obtained from standard drift-and-small-set arguments.

\paragraph{Negative mean-log growth.}
Suppose now that $\|\Gamma(Y_t)\|>0$ almost surely. If $Y_t$ is stationary and ergodic, 
\[
    \frac1n\sum_{k=0}^{n-1}\log\|\Gamma(Y_k)\|
    \rightarrow
    \mathbb E[\log\|\Gamma(Y_\infty)\|]
    \qquad\text{a.s.}
\]
Moreover,
\[
    \|\Gamma_{0:n}\|
    \le
    \prod_{k=0}^{n-1}\|\Gamma(Y_k)\|.
\]
A negative stationary mean of $\log\|\Gamma(Y_t)\|$ is therefore a natural sufficient condition for almost-sure decay of this upper bound. It is stronger than negativity of the top Lyapunov exponent for a general matrix product, but is considerably easier to verify. A block version, applied to $m$-step products, can weaken the condition.

Almost-sure decay does not by itself imply the fractional-moment estimate required in Assumption~\ref{ass:ipa-stability}; rare episodes of large expansion may dominate positive moments. The next assumption adds a quantitative concentration condition. It should be viewed as one convenient sufficient route, not as a baseline assumption for the entire IPA analysis.

\begin{assumption}[A sufficient condition based on negative mean-log growth]
\label{ass:IPAweakaverage}
Let $g(y):=\log\|\Gamma(y)\|$, and let $Y_\infty$ have the stationary law of the driving process.
\begin{enumerate}
    \item \textbf{Negative mean-log growth.}
    For some $\alpha>1$, $\mathbb E[g(Y_\infty)]=-\log\alpha$.

    \item \textbf{Poisson equation and a one-sided local conditional sub-Gaussian bound.}
    There is a measurable solution $\eta$ to
    \begin{equation}
    \label{eq:poisson}
        \eta(y)-\mathbb E[\eta(Y_1)\mid Y_0=y]
        =
        \mathbb E[g(Y_\infty)]-g(y).
    \end{equation}
    Moreover, there exist $\bar q>0$ and $c_\eta<\infty$ such that, for every $r\in(0,\bar q)$,
    \[
        \mathbb E_t\left[
            \exp\left
            \{-r(\eta(Y_{t+1})-\mathbb E_t[\eta(Y_{t+1})])\right\}
        \right]
        \le
        e^{c_\eta r^2}
        \qquad\text{a.s.}
    \]

    \item \textbf{Conditional exponential-moment stability.}
    For all sufficiently small $r>0$, there exists $C_r<\infty$ such that, with the two signs interpreted separately,
    \[
        \sup_{k\ge1}
        \mathbb E_t\left[e^{\pm2r\eta(Y_{t+k})}\right]
        \le
        C_r\left(1+e^{\pm2r\eta(Y_t)}\right)
        \qquad\text{a.s.}
    \]
\end{enumerate}
\end{assumption}

Condition 1 is the familiar average-stability requirement. Conditions 2 and 3 are quantitative strengthenings used to convert it into a finite fractional-moment bound. The one-sided conditional sub-Gaussian condition is automatic when the centered Poisson increment is bounded; Condition 3 follows, for example, from a suitable exponential Lyapunov drift for $\eta(Y_t)$. These conditions are not minimal, but they make the finite-time argument transparent.

\begin{proposition}[Random-product bound under negative mean-log growth]
\label{lem:prod}
Under Assumption~\ref{ass:IPAweakaverage}, for every $\alpha_0\in(1,\alpha)$ there exist $q\in(0,\bar q/2)$ and $C<\infty$ such that, for all $t\ge0$ and $n\ge1$,
\begin{equation}
\label{eq:prod-conditional}
    \mathbb E_t\left[
        \prod_{i=1}^n\|\Gamma(Y_{t+i})\|^q
    \right]
    \le
    C\alpha_0^{-qn}
    \cosh^2\left(\frac{q\eta(Y_t)}{2}\right).
\end{equation}
Consequently, if
\[
    R_q(y)
    :=
    1+
    \|\Gamma(y)\|^q
    \cosh^2\left(\frac{q\eta(y)}{2}\right)
\]
satisfies $R_q(y)\le CV(y)^q$ and
\[
    \sup_{t\ge0}\mathbb E[V(Y_t)\mid Y_0=y]\le CV(y),
\]
then Condition 4 of Assumption~\ref{ass:ipa-stability} holds with exponent $q$ and contraction rate $\gamma=\alpha_0$.
\end{proposition}

\begin{remark}[Verifying the Poisson equation]
Suppose the driving process has invariant law $\pi$, is geometrically contractive under a co-adapted coupling, and $g$ is Lipschitz with the required first moment. Let $\pi g:=\int g\,d\pi$. Then the correctly signed series
\[
    \eta(y)
    =
    \sum_{k=0}^\infty
    \bigl(\pi g-P^kg(y)\bigr)
\]
solves \eqref{eq:poisson}. Under a finite first moment of the coupling distance under $\pi$, geometric contraction also makes $\eta$ Lipschitz. A precise statement is given in the appendix.
\end{remark}

\begin{remark}[Why the exponential conditions are local]
The proof of Proposition~\ref{lem:prod} uses only a sufficiently small positive fractional moment $q$. Accordingly, the conditional moment-generating-function bounds are needed only in a neighborhood of the origin, not for all exponential orders. This distinction allows the result to cover models with substantially lighter regularity than would be required for a global exponential-moment bound.
\end{remark}

\section{Application: Strategic G/G/c Queue} \label{sec:ggc_queue}


We now apply the general framework to a strategic queue in which customers' arrival behavior responds to the stationary waiting-time distribution. 

Consider a first-come, first-served $G/G/c$ queue. Let $T_n$ denote the interarrival time between the $n$th and $(n+1)$st arrivals, and let $S_n$ denote the service requirement of the $n$th customer. The ordered workload vector immediately before the $n$th arrival is
\[
    W_n=(W_n(1),\ldots,W_n(c)),
    \qquad
    W_n(1)\le \cdots\le W_n(c),
\]
with state space $\mathcal S:=\{w\in\mathbb R_+^c:w_1\le\cdots\le w_c\}$.
Let $e_1=(1,0,\ldots,0)$, ${\bf 1}=(1,\ldots,1)$, and let
$r:\mathbb R^c\to\mathbb R^c$ rearrange a vector into nondecreasing order. Define the Kiefer--Wolfowitz map
\[
    F(w,s,t):=
    r\bigl((w+s e_1-t{\bf 1})^+\bigr),
\]
where the positive part is applied coordinatewise. The workload recursion is
\begin{equation}\label{eq:GGC}
    W_{n+1}
    =
    F(W_n,S_n,T_n)
    =
    r\bigl((W_n+S_n e_1-T_n{\bf 1})^+\bigr).
\end{equation}
The first coordinate $W_n(1)$ is the waiting time experienced by the $n$th arrival.

Strategic behavior enters through the interarrival-time distribution. Let
\[
    H:\mathcal P(\mathcal S)\to Z\subseteq\mathbb R^d
\]
be an aggregate congestion statistic. For example, $H(\mu)$ may be the mean waiting
time, $H(\mu)=\int_S w_1\,\mu(dw)$,
or a vector of waiting-time statistics. If $\mu_n=\operatorname{Law}(W_n)$, define $a_n:=H(\mu_n)$. Conditional on $a_n$, the interarrival time $T_n$ has law $T(a_n)$ and is independent of $S_n$ and, conditional on $a_n$, of the past.

As in the general framework, we first freeze the aggregate. For each $a\in A$, let $P_a$ be the transition kernel generated by $W_{n+1}=F(W_n,S_n,T(a))$.
When the frozen workload chain has a stationary distribution, denote it by $\pi_a$. The self-consistency map is $G(a)=H(\pi_a)$.
A stationary equilibrium of the strategic queue is obtained from a fixed point $a^\star=G(a^\star)$; the corresponding stationary workload distribution is $\pi_{a^\star}$.

Throughout this section, $W_1$ denotes the Wasserstein distance induced by the $L_1$ metric on $\mathcal S$. 


\begin{assumption}[Strategic \(G/G/c\) primitives]
\label{ass:ggc}
The following conditions hold.

\begin{enumerate}
\item \textbf{Parameter Lipschitzness of interarrival times.} There exists $L_T<\infty$
such that, 
\[
    W_1\bigl(\mathrm{Law}(T(a)),\mathrm{Law}(T(b))\bigr)
    \le
    L_T\|a-b\|, \quad \forall a,b\in A.
\]

\item \textbf{Light-tail service time distribution.} There exists $\bar\theta>0$ such that
$\mathbb E e^{\theta S}<\infty$, for all $0\le \theta\le \bar\theta$.

\item \textbf{Uniform stability and a reset opportunity.} There exists a nonnegative random
variable $T_\star$ satisfying
$0<\mathbb ET_\star<\infty$, with unbounded support on $[0,\infty)$, such that
$T(a)\succeq_{\rm st} T_\star$, for all $a\in A$.
Moreover, for some $\epsilon>0$,
$\mathbb E S \le (1-\epsilon)c\,\mathbb E T_\star$,
and there exists $s_0<\infty$ such that $\mathbb P\{S\le s_0\}>0$.

\item \textbf{Aggregator growth and continuity.} The aggregator $H$ depends only on the first-coordinate marginal. There exist
$k\in\mathbb N$ and $a_H,b_H\ge0$ such that
\[
    \|H(\mu)\|
    \le a_H+b_H\int_{\mathcal S}w_1^k\,\mu(dw).
\]
In addition, $H$ is continuous along weakly convergent sequences with uniformly bounded $k$th first-coordinate moments: if $\mu_n\Rightarrow\mu$ and 
$\sup_n\int_{\mathcal S} w_1^k\,\mu_n(dw)<\infty$,
then $H(\mu_n)\to H(\mu)$.
\end{enumerate}
\end{assumption}

Assumption~\ref{ass:ggc} separates the ingredients needed below. Part~1 controls how the frozen kernel changes with the aggregate. Parts~2--3 provide a uniform exponential drift and a state-independent opportunity to empty the system. Part~4 is used only for continuity and existence of a self-consistent aggregate.

Parts 1 and 2 are standard and typically easy to verify. For part 3, it is natural to assume that $a\mapsto T(a)$ is increasing in stochastic order. 
In this case, if $a_{0}$ is a lower bound of $A$  we can choose $T_\star:=T(a_0)$.  
As a concrete illustration, suppose the aggregate is scalar and $T(a)\sim \mathrm{Exp}(\lambda(a))$, with $\lambda(a)=1/\max\{1,a\}$.
Then as \(\mathbb E T(a)=\max\{1,a\}\ge1\), writing $T_\star\sim\mathrm{Exp}(1)$, one may couple
$T(a)=\max\{1,a\}T_\star$. Hence $T(a)\succeq_{\rm st}T_\star$.
\[
    W_1\bigl(\operatorname{Law}(T(a)),\operatorname{Law}(T(b))\bigr)
    =\bigl|\max\{1,a\}-\max\{1,b\}\bigr|
    \le |a-b|.
\]
Thus, if $\mathbb ES/c<1$ and the service time has a finite exponential moment, Parts~1--3 hold with $T_\star\sim\mathrm{Exp}(1)$.


The next theorem supplies the queueing input needed by the outer fixed-point algorithms. Its proof uses three properties of the Kiefer--Wolfowitz recursion: nonexpansiveness in the initial workload, Lipschitz dependence on the interarrival time, and a common-noise reset when a sufficiently long interarrival time empties all servers.

\begin{theorem}[Uniform ergodicity and stationary perturbation]
\label{thm:ggc-perturbation}
Suppose Assumption~\ref{ass:ggc}(1)--(3) holds. Then there exist
$\theta_\star>0$, a Lyapunov function
\begin{equation}
\label{eq:lyapunov_ggc}
    V(w):=\frac1c\sum_{i=1}^c e^{\theta_\star w_i},
\end{equation}
constants $C_0<\infty$ and $\rho\in(0,1)$, and, for every $a\in Z$, a unique stationary distribution $\pi_a$, such that
\[
    W_1(\nu P_a^n,\pi_a)
    \le C_0\bigl(1+\nu V\bigr)\rho^n,
    \qquad n\ge0,
\]
for every probability measure $\nu$ with $\nu V<\infty$. All constants are uniform in $a$, and
$\sup_{a\in Z}\pi_a V<\infty$.
Moreover, there exist constants $C<\infty$ and $B<\infty$, independent of $a,b$, such that
\begin{equation}
\label{eq:ggc-log-lip-w1}
    W_1(\pi_a,\pi_b)
    \le
    C\|a-b\|
    \left(1+\log_+\frac{B}{\|a-b\|}\right),
    \qquad a,b\in Z,\ a\ne b.
\end{equation}
Consequently, if $H$ is $L_H$-Lipschitz with respect to $W_1$, then
\[
    \|G(a)-G(b)\|
    \le
    L_HC\|a-b\|
    \left(1+\log_+\frac{B}{\|a-b\|}\right).
\]
\end{theorem}

Theorem~\ref{thm:ggc-perturbation} also yields the uniform moment bound needed for equilibrium existence.

\begin{corollary}[Existence of a stationary equilibrium]
\label{cor:ggc-existence}
Suppose Assumption~\ref{ass:ggc} holds and $Z\subseteq\mathbb R^d$ is nonempty, closed, and convex. Then $G$ has a fixed point $a^\star\in Z$. Consequently, $\pi_{a^\star}$ is a stationary equilibrium of the strategic $G/G/c$ queue.
\end{corollary}

\begin{remark}[Aggregate-dependent service laws]
The preceding results extend to a service law $S(a)$. A convenient uniform formulation assumes
$\sup_{a\in Z}\mathbb E e^{\theta S(a)}<\infty$ for all $ 0\le\theta\le\bar\theta$,
a common stability/reset condition, including
$\inf_{a\in Z}\mathbb P\{S(a)\le s_0\}>0$, and
\[
    W_1\bigl(\operatorname{Law}(S(a)),\operatorname{Law}(S(b))\bigr)
    \le L_S\|a-b\|.
\]
The one-step perturbation bound then contains the additional term $L_S\|a-b\|$; the remainder of the argument is unchanged.
\end{remark}

\subsection{Relaxed Fixed Point Iterations for Strategic $G/G/c$ Queue}
Theorem~\ref{thm:ggc-perturbation} verifies the log-Lipschitz regularity required by the general RFPI result. Convergence additionally requires an incremental one-sided condition. We give two examples in which this condition follows from queue structure.

\begin{example}[Scalar congestion feedback] \label{eg:GGcOneDim}
Suppose $A\subseteq\mathbb R$, and $H$ is increasing with respect to the stochastic order of the waiting-time marginal. Assume also that $T(b)\succeq_{\rm st}T(a)$ whenever $b\geq a$. Thus, a larger congestion signal induces longer interarrival times. Under a common-noise coupling, longer interarrival times produce smaller workload vectors, so the stationary response $G(a)=H(\pi_a)$ is decreasing. This negative feedback need not make $G$ a contraction, but it supplies the required one-sided stability.
\end{example}


\begin{proposition}[Scalar strategic $G/G/c$ queue]
\label{prop:scalar-ggc}
Suppose the conditions of Example~\ref{eg:GGcOneDim} and Assumption~\ref{ass:ggc} hold. Let $K\subseteq A$ be the compact interval constructed in the proof of Corollary~\ref{cor:ggc-existence}, so that $G(K)\subseteq K$. Assume, in addition, that $H$ is Lipschitz with respect to $W_1$. Then $G$ is decreasing on $K$ and has a unique fixed point $a^\star$. For
$\delta_k=\delta_0(k+1)^{-2/3}$,
with $\delta_0>0$ sufficiently small, the relaxed fixed-point iteration converges to $a^{\star}$ from every $a_0\in K$. 
\end{proposition}

In particular, Proposition~\ref{prop:scalar-ggc} covers mean waiting time. Indeed,
$H(\mu)=\int w_1\,\mu(dw)$ is increasing in stochastic order and is $1$-Lipschitz with respect to $W_1$.

\begin{example}[Mean--variance feedback in an $M/M/1$ queue]
\label{eg:2D-MM1}
Consider a strategic $M/M/1$ model in which customers respond to both the mean and variance of the stationary waiting time. For $a=(a_1,a_2)\in\mathbb R_+^2$, define
$\lambda(a):=\bigl(\tau_0-\alpha_1a_1-\alpha_2a_2\bigr)^+$ and $\nu(a):=\nu_0+\gamma a_1$,
where $\tau_0,\nu_0>0$ and $\alpha_1,\alpha_2,\gamma\ge0$. The arrival rate decreases with both summary statistics, whereas the service rate increases with mean waiting time.

For each fixed $a$, let $\pi_a$ be the stationary waiting-time distribution, with the convention $\pi_a=\delta_0$ when $\lambda(a)=0$. Define
$G(a)
:=\left(\mathbb E_{\pi_a}[W],\operatorname{Var}_{\pi_a}(W)\right)$.
Let $\Delta_0:=\nu_0-\tau_0>0$, $\rho_0:=\tau_0/\nu_0$,
$\bar m:=\rho_0/\Delta_0$,
$\bar v:=\rho_0(2-\rho_0)/\Delta_0^2$, and set
$K:=[0,\bar m]\times[0,\bar v]$.
Similar to above, for $a\in K$, write
$\rho(a):=\lambda(a)/\nu(a)$ and $\Delta(a):=\nu(a)-\lambda(a)$.
Then we have
\begin{equation}
\label{eq:mm1-response}
    G(a)
    =\left(
        \frac{\rho(a)}{\Delta(a)},
        \frac{\rho(a)(2-\rho(a))}{\Delta(a)^2}
    \right),
\end{equation}
with $G(a)=(0,0)$ when $\lambda(a)=0$.
\end{example}

\begin{proposition}[Strategic $M/M/1$ mean--variance example]
\label{prop:mm1-mean-var}
Fix $\tau_0>0$ and $\alpha_1,\alpha_2,\gamma\ge0$. There exists $\bar\Delta<\infty$ such that, whenever $\Delta_0>\bar\Delta$, the map $G$ in \eqref{eq:mm1-response} maps $K$ into itself, is Lipschitz on $K$, and satisfies
\[
    \langle G(x)-G(y),x-y\rangle
    \le c_{\rm os}\|x-y\|^2,
    \qquad x,y\in K,
\]
for some one-sided Lipschitz constant $c_{\rm os}<1$. One explicit sufficient condition is
\begin{equation}
\label{eq:mm1-simple-sufficient}
    \frac{\alpha_2}{2\Delta_0^2}
    +\frac{\alpha_1+\gamma}{\Delta_0^3}
    <1.
\end{equation}
Consequently, if $L$ is any Lipschitz constant of $G$ on $K$, then every fixed damping parameter $\delta\in(0,1]$ satisfying
$(1-\delta)^2+2c_{\rm os}\delta(1-\delta)+L^2\delta^2<1$
yields linear convergence of the relaxed iteration to the unique fixed point $a^\star\in K$.
\end{proposition}


Figure~\ref{fig:osl_map} illustrates the one-sided condition above numerically. The blue region is the parameter region in which the analytical one-sided condition is certified.  The red region indicates only that this sufficient condition is not verified; it does not imply that the iteration diverges. The grey region corresponds to an unstable baseline queue, $\nu_0\le\tau_0$.


\begin{figure}[htbp]
    \centering
    \includegraphics[width=0.6\linewidth, height=7cm, keepaspectratio]{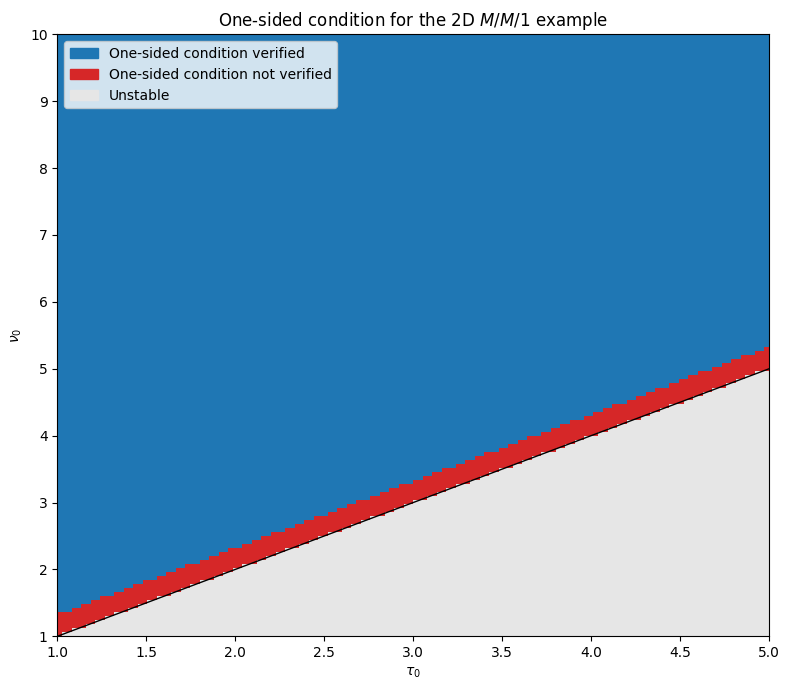}
    \caption{
    The blue region represents parameter combinations for which the one-sided condition is verified numerically. The red region indicates where this numerical check does not verify the condition (but it does not imply divergence of the relaxed fixed-point iteration). The grey area represents an unstable baseline queue ($\nu_0 \le \tau_0$). We use $\alpha_1=0.05$, $\alpha_2=0.05$, and $\gamma=0.03$.}
    \label{fig:osl_map}
\end{figure}

\subsection{IPA for the Strategic $G/G/c$ Queue}
\label{subsec:ggc-ipa}

The relaxed fixed-point results above are useful when log-Lipschitz regularity and an incremental one-sided condition can be verified analytically. In more complex queueing models, however, these conditions may be difficult to establish. A natural alternative is to minimize
\[
    \ell(a)=\frac12\|a-G(a)\|^2
\]
by a gradient-based method. This requires an estimator of $\nabla G(a)$. The result below shows that the queue's regenerative empty-state structure controls both the pathwise sensitivity and the finite-horizon bias of an IPA estimator.

Fix a compact convex set $K\subseteq A$ and an open set $K^+\subseteq\mathbb R^d$ containing $K$. Suppose the frozen primitives are defined on $K^+$, the light-tail and stability/reset conditions in Assumption~\ref{ass:ggc}(2)--(3) hold uniformly for $a\in K^+$, and the interarrival time has a common-random-number representation $T(a)=\tau(a,U)$ for $a\in K^+$, where $U$ is exogenous. For each $a\in K^+$, define
\[
    W_{n+1}(a)
    =
    r\bigl((W_n(a)+S_ne_1-\tau(a,U_n){\bf 1})^+\bigr).
\]
Let $J_n(a):=\nabla_a W_n(a)\in\mathbb R^{c\times d}$.

\begin{assumption}[IPA regularity for the \(G/G/c\) queue]
\label{ass:GGc-ipa}
The following conditions hold.

\begin{enumerate}
\item For every $a\in K^+$, the law of $\tau(a,U)$ has a density with respect to Lebesgue measure, and the service-time distribution is non-atomic.

\item For almost every $U$, the map $a\mapsto\tau(a,U)$ is continuously differentiable on $K^+$. There exist $q>2$ and a random variable $B_T(U)$ such that
\[
    \sup_{a\in K^+}\|\nabla_a\tau(a,U)\|
    \le B_T(U),
    \quad
    \mathbb E[B_T(U)^q]<\infty.
\]

\item The aggregator has the expectation form
$H(\mu)=\int_{\mathcal S} h(w)\,\mu(dw)$,
where \(h: \mathcal S\to\mathbb R^d\) has a continuously differentiable extension to an open neighborhood of  \(\mathcal S\). For the Lyapunov function $V$ in \eqref{eq:lyapunov_ggc}, there exist $C_h,L_h<\infty$ such that
$\|h(w)\|\le C_h V(w)$ and $\|\nabla h(w)\|_F\le L_h$ for any 
$w\in \mathcal S$.
\end{enumerate}
\end{assumption}

Under Assumption~\ref{ass:GGc-ipa}, the finite-horizon map $a\mapsto W_n(a)$ is differentiable almost surely. Let $R_n(a)$ be the masked permutation matrix that records which coordinates survive the positive-part operation and how the positive coordinates are reordered. Then $\|R_n(a)\|\le1$ and
\begin{equation}
\label{eq:ggc-sensitivity-recursion}
    J_{n+1}(a)
    =R_n(a)\left(J_n(a)-{\bf 1}\nabla_a\tau(a,U_n)^\top\right).
\end{equation}
When the next workload vector is zero, $R_n(a)=0$, and hence the entire sensitivity matrix is reset to zero. This exact reset is the key difference between the queueing application and the general fractional-moment setting of Section~\ref{sec:gradient}: together with the moment condition $q>2$, it yields second-moment stability of the unmodified IPA summand.

For $W_0=0$ and $J_0=0$, define
\begin{equation}
\label{eq:ggc-ipa-estimator}
    \widehat{\nabla G}_{N,n}(a)
    :=\frac1N\sum_{i=1}^N
    \nabla h\bigl(W_n^{(i)}(a)\bigr)J_n^{(i)}(a),
\end{equation}
where the $N$ workload--sensitivity trajectories are independent.

\begin{theorem}[Finite-time IPA for the strategic $G/G/c$ queue]
\label{thm:GGc_IPA_valid}
Suppose the light-tail and stability/reset conditions in Assumption~\ref{ass:ggc}(2)--(3) hold uniformly for $a\in K^+$, and suppose Assumption~\ref{ass:GGc-ipa} holds. Then $G$ is continuously differentiable on $K$. For each $a\in K$, the augmented process $(W_n(a),J_n(a))$ admits an invariant distribution $\Pi_a$, whose workload marginal is $\pi_a$, and
\begin{equation}
\label{eq:ggc-steady-state-ipa}
    \nabla G(a)=\mathbb E_{\Pi_a}[\nabla h(W)J].
\end{equation}
Moreover, there exist constants $C<\infty$ and $\beta\in(0,1)$, independent of $a\in K$, $N$, and $n$, such that
\begin{equation}
\label{eq:ggc-ipa-mse}
    \mathbb E\left[
        \left\|\widehat{\nabla G}_{N,n}(a)-\nabla G(a)\right\|_F^2
    \right]
    \le
    \frac{CL_h^2}{N}
    +CL_h^2\beta^{2n}.
\end{equation}
Consequently, choosing $n$ proportional to $\log N$ yields
\[
    \mathbb E\left[
        \left\|\widehat{\nabla G}_{N,n}(a)-\nabla G(a)\right\|_F^2
    \right]
    =O(N^{-1})
\]
uniformly over $a\in K$.
\end{theorem}



\begin{remark}[Bounded-gradient statistics]
The bounded-gradient condition covers mean waiting time, $h(w)=w_1$, and other Lipschitz performance statistics. Statistics with unbounded derivative, such as
$h(w)=w_1^2$, require a polynomial-growth extension of
Theorem~\ref{thm:GGc_IPA_valid}. The exponential workload Lyapunov bound makes a polynomial-growth extension feasible, but that extension requires corresponding moment assumptions on both $\nabla h(W)$ and the sensitivity process.
\end{remark}

\section{Application: Opinion Dynamics}
\label{sec:opinion-dynamics}

We next consider an opinion-dynamics model with heterogeneous agent types. Agents update their beliefs based on their own current beliefs and on the average beliefs of other agents.

There are \(m\) types of agents, denoted by \(\Theta=\{\tau_1,\ldots,\tau_m\}\). An agent's state is \(s=(y,\tau)\), where \(y\in\mathbb R\) is her current belief and
\(\tau\in\Theta\) is her type. Thus, the state space is $S := Y\times \Theta$. Types are fixed over time. Let \(\lambda_\tau>0\) denote the population share of type
\(\tau\), with \(\sum_{\tau\in\Theta}\lambda_\tau=1\). Since types do not change, the
type composition remains fixed, and only the belief distribution within each type evolves.

For a population distribution \(\mu\) on \(S\), define the type-level mean-belief vector (aggregator)
$H(\mu)
    =
    \bigl(H(\mu)_\tau\bigr)_{\tau\in\Theta}
    \in \mathbb R^m$
as the average belief for each agent type:
\begin{equation}
\label{eq:opinion-H}
    H(\mu)_\tau
    :=
    \frac{1}{\lambda_\tau}
    \int_S y\,\mathbf 1\{\theta=\tau\}\,\mu(dy,d\theta),
    \quad \tau\in\Theta.
\end{equation}
Let \(W\in\mathbb R_+^{m\times m}\) be a row-stochastic matrix. The entry
\(W_{\tau,\tau'}\) measures how much type-\(\tau\) agents weight the average belief of
type-\(\tau'\) agents. Given an aggregate belief vector \(a\in\mathbb R^m\), the social
signal observed by a type-\(\tau\) agent is
\begin{equation}
\label{eq:opinion-social-signal}
    S_\tau(a)
    :=
    (Wa)_\tau
    =
    \sum_{\tau'\in\Theta} W_{\tau,\tau'}a_{\tau'} .
\end{equation}
Since \(W\) has nonnegative entries, a larger aggregate belief vector generates a larger
social signal for every type.

The belief dynamics are as follows. Suppose the current population distribution is
\(\mu_n\), and let
$a_n:=H(\mu_n)$
be the vector of current type-level mean beliefs. A type-\(\tau\) agent with current belief
\(Y_n\) updates according to
\begin{equation}
\label{eq:opinion-nonlinear-dynamics}
    Y_{n+1}
    =
    w_\tau\bigl(Y_n,S_\tau(a_n),\varepsilon_{n+1}\bigr),
\end{equation}
where \(w_\tau\) is the type-\(\tau\) updating function and \(\varepsilon_{n+1}\) is an
idiosyncratic shock with distribution \(F_\tau\). The type remains fixed. Therefore the
population law evolves according to
$\mu_{n+1}=\mu_n P_{H(\mu_n)}$, i.e., it is a nonlinear Markov chain.
More explicitly, for a fixed \(a\in\mathbb R^m\), the frozen kernel \(P_a\) is defined by
\begin{equation}
\label{eq:opinion-Pa}
P_a\bigl((y,\tau),A\times B\bigr)
=
\mathbf 1\{\tau\in B\}
\int
\mathbf 1\left\{
w_\tau\bigl(y,S_\tau(a),\epsilon\bigr)\in A
\right\}
F_\tau(d\epsilon),
\end{equation}
for \(A\in\mathcal B(\mathbb R)\) and \(B\subseteq\Theta\). Thus, for each frozen aggregate
\(a\), the system is an ordinary Markov chain. We denote its stationary distribution by
\(\pi_a\), when it exists. 

\begin{example}[Noisy mean-field DeGroot model]
\label{ex:opinion-degroot}
A useful special case is a noisy mean-field version of the DeGroot model. Suppose
\[
    w_\tau(y,s,\epsilon)
    =
    \alpha_\tau y+\beta_\tau s+\gamma_\tau\epsilon,
\]
where \(0\le \alpha_\tau<1\), \(\beta_\tau\ge 0\), and
\(\mathbb E|\epsilon_\tau|<\infty\). Then a type-\(\tau\) agent updates by taking a weighted
combination of her current belief $y$, the social signal \(s\), and an idiosyncratic shock $\epsilon$.

For a fixed aggregate vector \(a\), the frozen belief process for type \(\tau\) is the affine
autoregressive recursion
\[
    Y_{n+1}
    =
    \alpha_\tau Y_n+\beta_\tau (Wa)_\tau+\gamma_\tau\epsilon_{n+1}.
\]
Since \(0\le \alpha_\tau<1\), this recursion has a unique stationary distribution under the
moment condition above. Its stationary mean is
\[
    G_\tau(a)
    =
    \frac{\beta_\tau}{1-\alpha_\tau}(Wa)_\tau
    +
    \frac{\gamma_\tau}{1-\alpha_\tau}\mathbb E[\epsilon_\tau].
\]
Thus, \(G\) is an affine, coordinatewise nondecreasing map. If, in addition,
\[
    \rho\left(
    \operatorname{diag}\left\{\frac{\beta_\tau}{1-\alpha_\tau}:\tau\in\Theta\right\}
    W
    \right)<1,
\]
then the fixed point is unique and the standard fixed-point iteration
\(a^{n+1}=G(a^n)\) converges globally.
\end{example}

The DeGroot example is linear, but the same monotonicity idea applies much more broadly.
If an agent's next-period belief is increasing in both her current belief and the social
signal she receives from the population, then a larger aggregate belief vector produces a
stochastically larger stationary belief distribution. Consequently, the map
\(G(a)=H(\pi_a)\) is increasing. This is exactly the structure needed for monotone
fixed-point iteration in  Proposition \ref{prop:monotone}. 
The following theorem makes this statement precise.

\begin{theorem}[Monotone computation for opinion dynamics]
\label{thm:opinion-dynamics}
Let \(Z=\mathbb R^m\), endowed with the coordinate-wise order. Let
\(\underline a,\overline a\in Z\) satisfy \(\underline a\le \overline a\), and define
\[
    K:=[\underline a,\overline a]
    :=
    \{a\in Z:\underline a\le a\le \overline a\}.
\]
Suppose $\underline a\le G(\underline a)$ and $G(\overline a)\le \overline a$.
Assume the following conditions hold.

\begin{enumerate}
\item[\textup{(i)}]
For every \(a\in K\), the frozen kernel \(P_a\) admits a unique stationary distribution
\(\pi_a\) among probability measures with type marginal \(\lambda\), and the frozen chain
converges weakly to \(\pi_a\) from any initial distribution with type marginal \(\lambda\).

\item[\textup{(ii)}]
The frozen kernels $P_a$ are weakly continuous on \(K\). In addition, the stationary distributions are
uniformly integrable in the belief coordinate:
\[
    \lim_{R\to\infty}
    \sup_{a\in K}
    \int_S |y|\,\mathbf 1\{|y|>R\}\,\pi_a(dy,d\theta)
    =0.
\]

\item[\textup{(iii)}]
For every type \(\tau\in\Theta\), the updating function is nondecreasing in the agent's own
belief and in the social signal. That is, for all \(k,\epsilon\in\mathbb R\),
$y\mapsto w_\tau(y,k,\epsilon)$
is nondecreasing, and for all \(y,\epsilon\in\mathbb R\),
$k\mapsto w_\tau(y,k,\epsilon)$
is nondecreasing.
\end{enumerate}

Then \(G\) is continuous and coordinatewise nondecreasing on \(K\), and \(G(K)\subseteq K\).
Consequently, the fixed-point iteration
with $a^0=\underline a$
converges monotonically to a fixed point \(a^*\in K\).
\end{theorem}

\section{Conclusion}
In this paper, we develop a computational framework for stationary equilibria in Markov systems whose transition laws depend on their own state distributions.
When this dependence operates through a finite-dimensional aggregate, freezing the aggregate at $a$ yields an ordinary Markov chain with stationary distribution $\pi_a$, reducing the equilibrium problem to the finite-dimensional self-consistency equation $a=G(a):=H(\pi_a)$. This reduction separates the steady-state analysis of the frozen system from the equilibrium computation and highlights how Markov-chain stability and the geometry of $G$ jointly determine an appropriate computational method.

When the fixed-point structure of $G$ is difficult to exploit directly, we instead formulate equilibrium computation as the residual-minimization problem 
$\min_a \frac12|a-G(a)|^2$.
To support this approach, we develop finite-time IPA estimators for $\nabla G(a)$, with error bounds that separate Monte Carlo error from finite-time bias and characterize the stability required of the associated sensitivity process.

We illustrate the framework through strategic $G/G/c$ queues and opinion dynamics. These applications demonstrate both the breadth of measure-dependent Markov models and the fact that different structural properties of the stationary response map naturally call for different equilibrium-computation methods.

\section*{Acknowledgments}

The research of  X.T.T. is supported by  Ministry of Education, Singapore, under the Academic Research Fund Tier 1 A-8002956-00-00.


\bibliographystyle{abbrvnat} 
\bibliography{nonlinear} 

\newpage

\appendix

\section{Algorithmic Summary for the Fixed Point Problem} \label{app:fixed_point_main}

Assume that \(G\) admits a fixed point \(a^\star\), except in parts \textup{(1)} and
\textup{(2)} below, where existence follows from the stated assumptions. 


\begin{enumerate}
\item \textbf{Contraction.}
If \(G\) is a contraction with modulus \(L\in[0,1)\), then \(a^\star\) is unique and Picard
iteration satisfies
\[
\|a_k-a^\star\|\le L^k\|a_0-a^\star\|.
\]
The relaxed iteration with constant damping $0<\delta\le 1$
also converges linearly, with factor \(1-\delta+\delta L\).
See Propositions ~\ref{prop:contraction} and ~\ref{prop:Lip}.

\item \textbf{Monotone order interval.}
Suppose \(K=[\underline a,\overline a]\subseteq A\subseteq \mathbb R^d\), \(G:K\to K\) is
continuous and coordinatewise nondecreasing, and
$\underline a\le G(\underline a)$, $G(\overline a)\le \overline a$.
Then Picard iteration started from \(\underline a\) converges monotonically to a fixed point
in \(K\). The same holds for the relaxed iteration with constant damping \(0<\delta\le 1\).
See Proposition~\ref{prop:monotone}.

\item \textbf{Nonexpansive.}
If \(G\) is nonexpansive, then the relaxed iteration with constant damping \(0<\delta<1\) satisfies
\[
\min_{0\le i\le k-1}\|a_i-G(a_i)\|^2
\le
\frac{1}{\delta(1-\delta)k}\|a_0-a^\star\|^2,
\]
i.e., the fixed-point residual decays at the ergodic rate \(O(1/k)\).
See  Proposition~\ref{prop:non-expansive}.

\item \textbf{Lipschitz with an incremental one-sided condition.}
Suppose \(G\) is \(L\)-Lipschitz on \(A\), and
\begin{equation}\label{eq:one_sided}
    \langle G(a)-G(b),a-b\rangle
    \le c\|a-b\|^2,
    \quad a,b\in Z, \mbox{ for some \(c<1\).}
\end{equation}
Then the relaxed iteration with constant damping
\(0<\delta\le 1\) converges linearly whenever
$(1-\delta)^2+2c\delta(1-\delta)+\delta^2L^2<1$.
See Proposition~\ref{prop:Lip}. When \(G\) is continuously
differentiable, the above incremental condition is implied by the Jacobian bound
\begin{equation}\label{eq:Jaco_bound}
    \frac{\nabla G(a)+\nabla G(a)^\top}{2}\preceq cI.
\end{equation}

\item \textbf{Log-Lipschitz with an incremental one-sided condition.}
Suppose 
\begin{equation}\label{eq:log_lip}
    \|G(a)-G(b)\|
    \le
    L\|a-b\|
    \left(1+\log_+\frac{B}{\|a-b\|}\right),
    \quad a\ne b.
\end{equation}
If \(G\) also satisfies the incremental one-sided condition \eqref{eq:one_sided},
then the relaxed iteration with decreasing damping
$\delta_k=\delta_0(k+1)^{-2/3}$,
for \(\delta_0\) sufficiently small, achieves a stretched-exponential bound 
\[
    \|a_k-a^\star\|^2
    \le
    C_0\exp\{-q\delta_0 k^{1/3}\}
\]
for a suitable constant \(q>0\).
See Proposition~\ref{prop:Log-Lip}.



\item \textbf{Gradient descent under PL geometry.} Suppose \(Z=\mathbb R^d\), \(G\in C^1(\mathbb R^d;\mathbb R^d)\). If \(\nabla \ell\) is \(L\)-Lipschitz continuous on \(\mathbb R^d\), and 
\(\ell\) satisfies the Polyak--\L ojasiewicz inequality
\begin{equation}\label{eq:PL_ineq}
\frac12\|\nabla \ell(a)\|^2\ge \mu \ell(a),
\quad a\in\mathbb R^d,
\mbox{ for some \(\mu>0\).} 
\end{equation}
Then the gradient-descent iteration with $0<\eta\leq 1/L$ yields
\[
\ell(a_k)\le (1-\mu\eta)^k \ell(a_0),
\quad k\ge 0.
\]
See Proposition~\ref{prop:gd_pl}.
Note that if $G$ is \(L\)-Lipschitz and satisfies \eqref{eq:Jaco_bound} with $c<1$, then \(\ell\) satisfies the
Polyak--\L ojasiewicz inequality with constant $\mu=(1-c)^2$.
\end{enumerate}

\subsection{Primitive conditions for regularity of the self-consistency map}
\label{sec:primitive-conditions}

The algorithmic summary above is stated in terms of the self-consistency map $G$.
We now provide simple sufficient conditions on the frozen kernels
\(\{P_a:a\in Z\}\) under which \(G\) is monotone or Lipschitz. These results
provide the link between the abstract convergence criteria for the outer algorithms and
the model primitives that appear in applications.

Recall that for each \(a\in Z\), \(P_a\) denote the one-step kernel of the frozen Markov
chain. Let \(P_a^n\) denote its \(n\)-step kernel. For a bounded measurable function
\(g:\mathcal S\to\mathbb R\), write
\[
(P_ag)(x):=\int_{\mathcal S} g(y)\,P_a(x,dy).
\]
For concreteness, we state the next results for aggregators of the form
\[
H(\mu)=\int_{\mathcal S} f(x)\,\mu(dx),
\]
although the same arguments extend to more general Lipschitz or order-preserving
aggregators. 

We first establish sufficient conditions for monotonicity. Suppose \(S\) and \(Z\subseteq \mathbb R^d\)
are partially ordered, where \(Z\) carries the coordinatewise order. We write
\(\mu\le_{st}\nu\) if
$\int_{\mathcal S} g(x)\,\mu(dx)\le \int_{\mathcal S} g(x)\,\nu(dx)$
for every bounded increasing measurable function \(g: S\to\mathbb R\).

\begin{assumption}[Increasing kernels]
\label{ass:increase}
For every bounded increasing measurable function \(g: S\to\mathbb R\),
the map \(x\mapsto (P_ag)(x)\) is increasing for each \(a\in Z\). In addition,
whenever \(a\le b\),
$(P_ag)(x)\le (P_bg)(x)$,for $x\in S$.
\end{assumption}

\begin{assumption}[Decreasing kernels]
\label{ass:decrease}
For every bounded increasing measurable function \(g: S\to\mathbb R\),
the map \(x\mapsto (P_ag)(x)\) is increasing for each \(a\in Z\). In addition,
whenever \(a\le b\),
$(P_ag)(x)\ge (P_bg)(x)$, for $x\in S$.
\end{assumption}

\begin{theorem}
\label{th:Monotone}
For
$H(\mu)=\int_{\mathcal S} f(x)\,\mu(dx)$,
suppose \(f:\mathcal S\to \mathbb R^d\) has increasing components.
\begin{enumerate}
\item Under Assumption~\ref{ass:increase}, \(G\) is increasing on \(Z\).
\item Under Assumption~\ref{ass:decrease}, \(G\) is decreasing on \(Z\).
\end{enumerate}
\end{theorem}

We next establish some sufficient conditions for Lipschitz or log-Lipschitz continuity.

\begin{assumption}
\label{ass:continuous}
There exist constants \(\rho\in[0,1)\) and \(L_P\in(0,\infty)\), and a measurable function
\(V:\mathcal S\to [1,\infty)\), such that for all \(a,b\in Z\) and all
probability measures \(\mu,\nu\in\mathcal P(\mathcal S)\),
\[
W_d(\mu P_a,\nu P_a)\le \rho\, W_d(\mu,\nu)
\quad\mbox{ and }\quad
W_d(\mu P_a,\mu P_b)
\le
L_P\|a-b\|\int_{ S}V(x)\,\mu(dx).
\]
In addition,
$M:=\sup_{a\in Z}\int_{ S}V(x)\,\pi_a(dx)<\infty$.
\end{assumption}

\begin{theorem}
\label{th:Lipschitz}
Under Assumption~\ref{ass:continuous}, suppose
\(f:\mathcal S\to \mathbb R^d\) is \(L_f\)-Lipschitz. Then
\[
\|G(a)-G(b)\|
\le
\frac{L_f L_P M}{1-\rho}\,\|a-b\|,
\quad a,b\in Z.
\]
\end{theorem}


The preceding result obtains Lipschitz continuity of \(G\) when the frozen kernels are uniformly contractive. The next result gives a log-Lipschitz
conclusion when the frozen kernels are only nonexpansive but converge geometrically to stationarity.

\begin{assumption}\label{ass:continuous2}
Let \(d\) be a bounded lower semicontinuous metric on \(S\). 
\begin{enumerate}
\item Uniform geometric convergence: there exist \(K_0<\infty\) and \(\rho\in(0,1)\)
such that
\[
W_d(\pi_b P_a^n,\pi_a)\le K_0\rho^n,
\quad a,b\in Z,\ n\ge 0.
\]
\item One-step parameter Lipschitzness: there exists \(L_0\in(0,\infty)\) such that
\[
W_d(\delta_xP_a,\delta_xP_b)\le L_0\|a-b\|,
\quad x\in S, ~ a,b\in Z.
\]
\item Nonexpansiveness of the frozen kernel:
\[
W_d(\mu P_a,\nu P_a)\le W_d(\mu,\nu),
\quad \mu,\nu\in\mathcal P(S).
\]
\end{enumerate}
\end{assumption}

\begin{theorem} \label{th:log-Lipschitz}
Under Assumption \ref{ass:continuous2}, with \(\kappa=-\log\rho\), for any $a,b\in Z$, $a\neq b$,
\[
W_d(\pi_a,\pi_b)
\le
L_0\left(1+\frac1\kappa\right)
\|a-b\|
\left(1+\log_+\frac{K_0\kappa/L_0}{\|a-b\|}\right).
\]
If, in addition, \(H\) is \(L_H\)-Lipschitz with respect to \(W_d\), i.e.,
$\|H(\mu)-H(\nu)\|
    \le
    L_H W_d(\mu,\nu)$,
for $\mu,\nu\in\mathcal P(S)$,
then the self-consistency map \(G(a)=H(\pi_a)\) satisfies
\[
    \|G(a)-G(b)\|
    \le
    L_H L_0\left(1+\frac1\kappa\right)\|a-b\|
    \left(
        1+\log_+\frac{K_0\kappa/L_0}{\|a-b\|}
    \right),
\]
i.e., \(G\) is log-Lipschitz on \(Z\).
\end{theorem}

\begin{remark}[Bounded versus unbounded metrics]
\label{rem:bunb}
The boundedness of \(d\) in Assumption \ref{ass:continuous2} is imposed mainly for technical convenience. It ensures that
\(W_d\) is finite for all probability measures and avoids imposing moment conditions on
the frozen stationary distributions. The proof of Theorem~\ref{th:log-Lipschitz}
uses only the three inequalities in Assumption~\ref{ass:continuous2} and therefore extends
without change to an unbounded lower semicontinuous metric \(d\), provided all measures
under consideration belong to the finite-\(d\)-moment class
\[
    \mathcal P_d(S)
    :=
    \left\{
        \mu\in\mathcal P(S):
        \int_S d(x,x_0)\,\mu(dx)<\infty
    \right\}
\]
for some, and hence every, \(x_0\in S\), and provided Assumption~\ref{ass:continuous2}
holds on this class.
\end{remark}

\subsection{Convergence Analyses for Fixed Point Algorithms} \label{app:fixed_point}

\begin{proposition}[Contraction]\label{prop:contraction}
Suppose $G$ is a contraction, i.e., $\|G(a)-G(\tilde a)\|\leq L\|a-\tilde a\|$ for some $L\in[0,1)$ and $a^\star$ is a fixed point of $G$. For the fixed-point iteration,
\[
\|a_k-a^{\star}\|\leq L^k\|a_0-a^{\star}\|.
\]
\end{proposition}

\proof{Proof.}
For the fixed-point iteration, since
\[
\|a_k-a^{\star}\| = \|G(a_{k-1})-G(a^{\star})\| \leq L\|a_{k-1}-a^{\star}\|
\]
we have $\|a_k-a^{\star}\|\leq L^k\|a_0-a^{\star}\|$.
\Halmos
\endproof

\begin{proposition}[Monotone convergence on an order interval] \label{prop:monotone}
Suppose $A\subseteq \mathbb R^d$ is endowed with the coordinatewise order, and let
$G:A \to A$ be continuous and coordinatewise nondecreasing. In addition, suppose we start with an $a_0$ such that $a_0\leq G(a_0)$, and there exists $b\geq a_0$ such that $G(b)\leq b$ and the order interval
$[a_0,b]:=\{a\in\mathbb R^d:\ a_0\le a\le b\}$ is contained in $A$. For the fixed-point iteration and the relaxed fixed-point iteration with $\delta_k\equiv\delta$ and $0<\delta\leq 1$, we have that there exists $a^{\star}$ with $G(a^{\star})=a^{\star}$ such that
$\lim_{k\rightarrow\infty} a_k = a^{\star}$.
\end{proposition}
\proof{Proof.}
Fix $0<\delta\le1$, and define
$T_\delta(a):=(1-\delta)a+\delta G(a)$, $a\in[a_0,b]$. Note that when $\delta=1$, we have $T_\delta(a)=G(a)$.
Since $G$ is coordinatewise nondecreasing, and addition and multiplication by
nonnegative scalars preserve the coordinatewise order, $T_\delta$ is also
coordinatewise nondecreasing.

Next, for any $a\in[a_0,b]$, monotonicity of $G$ gives
\[
a_0\le G(a_0)\le G(a)\le G(b)\le b,
\]
so $G(a)\in[a_0,b]$. Because $T_\delta(a)$ is a
convex combination of $a$ and $G(a)$, both of which belong to $[a_0,b]$, $T_\delta(a)\in[a_0,b]$ as well. 

We now show by induction that
$a_0\le a_1\le a_2\le \cdots \le b$.
Assume $a_{k-1}\le a_k\le b$. Since $T_\delta$ is nondecreasing,
\[
a_k=T_\delta(a_{k-1})\le T_\delta(a_k)=a_{k+1}.
\]
Also, as noted above, $a_k\in[a_0,b]$ implies $a_{k+1}=T_\delta(a_k)\in[a_0,b]$, so
$a_{k+1}\le b$. Thus $\{a_k\}_{k\ge0}$ is coordinatewise nondecreasing and bounded
above by $b$.

Because $\mathbb R^d$ is finite-dimensional, each coordinate sequence
$\{(a_k)_i\}_{k\ge0}$ is increasing and bounded above, hence convergent. Therefore
$a_k\to a^\star$ for some $a^\star\in[a_0,b]$.

Since $T_\delta$ is continuous,
\[
a^\star=\lim_{k\to\infty}a_{k+1}
=\lim_{k\to\infty}T_\delta(a_k)
=T_\delta\!\left(\lim_{k\to\infty}a_k\right)
=T_\delta(a^\star).
\]
Thus $a^\star$ is a fixed point of $T_\delta$. Because $\delta>0$, $T_\delta(a^\star)=a^\star$ implies $G(a^\star)=a^\star$.
\Halmos
\endproof

\begin{proposition}[Nonexpansive] \label{prop:non-expansive}
Suppose $A\subseteq \mathbb R^d$ is convex, $G: A\to A$ is
nonexpansive, that is, $\|G(a)-G(\tilde a)\|\leq \|a-\tilde a\|$, and $a^\star$ is a fixed point of $G$. For the relaxed fixed-point iteration with $\delta_k\equiv\delta$ and $0<\delta<1$,
\[
\min_{0\leq k\leq N-1} \|a_k - G(a_k)\|^2 \leq \frac{1}{\delta(1-\delta)N}\|a_0-a^{\star}\|^2.
\]
\end{proposition}
\proof{Proof.}
Since $a^{\star}$ is a fixed point of $G$, for the relaxed fixed-point iteration, we have
\[\begin{split}
\|a_{k+1}-a^{\star}\|^2 &= \|(1-\delta)(a_k-a^{\star}) + \delta(G(a_k)-a^{\star})\|^2\\
&=(1-\delta)\|a_{k}-a^{\star}\|^2 + \delta\|G(a_k)-G(a^{\star})\|^2 - \delta(1-\delta)\|a_k - G(a_k)\|^2\\
&\leq \|a_k-a^{\star}\|^2-\delta(1-\delta)\|a_k - G(a_k)\|^2
\end{split}\]
Sum over $k$, we have
\[
\sum_{k=0}^{N-1}\delta(1-\delta)\|a_k - G(a_k)\|^2\leq \|a_0-a^{\star}\|^2.
\]
Then,
\[
\min_{0\leq k\leq N-1} \|a_k - G(a_k)\|^2 \leq \frac{1}{N}\sum_{k=0}^{N-1}\|a_k - G(a_k)\|^2 \leq \frac{1}{\delta(1-\delta)N}\|a_0-a^{\star}\|^2.
\]
\Halmos
\endproof

We note from Proposition \ref{prop:non-expansive} that setting $\delta=1/2$ maximizes the decaying rate. 

For the following proposition, we focus on $L>1$. Otherwise, we have a contraction.

\begin{proposition}[Lipschitz] \label{prop:Lip}
Let \(A\subseteq \mathbb R^d\) be nonempty, closed, and convex, and let
\(G:A\to A\). Suppose that \(G\) admits a fixed point \(a^\star\in A\).
Assume that \(G\) is \(L\)-Lipschitz on \(A\), i.e., $\|G(a)-G(\tilde a)\|\leq L\|a-\tilde a\|$,
and satisfies the incremental one-sided condition \eqref{eq:one_sided}.
For the relaxed fixed-point iteration with constant damping \(0<\delta\le 1\), if
$q_{\delta}:=(1-\delta)^2+2c\delta(1-\delta)+\delta^2L^2<1$,
then
\[
    \|a_k-a^\star\|^2
    \le q_\delta^k\|a_0-a^\star\|^2,
    \quad k\ge 0.
\]
In particular, \(a_k\to a^\star\). Moreover, the fixed point is unique.
\end{proposition}

\proof{Proof.}
Since $G(a^{\star})=a^{\star}$,
\[
    a_{k+1}-a^\star
    =
    (1-\delta)(a_k-a^\star)
    +
    \delta\bigl(G(a_k)-G(a^\star)\bigr).
\]
Expanding the squared norm gives
\[
\begin{aligned}
    \|a_{k+1}-a^\star\|^2
    &=
    (1-\delta)^2\|a_k-a^\star\|^2
    +\delta^2\|G(a_k)-G(a^\star)\|^2  \\
    &\quad
    +2\delta(1-\delta)
    \langle a_k-a^\star, G(a_k)-G(a^\star)\rangle .
\end{aligned}
\]
By the Lipschitz condition and the one-sided condition,
\[
    \|G(a_k)-G(a^\star)\|^2
    \le L^2\|a_k-a^\star\|^2
\quad\mbox{ and }\quad
    \langle a_k-a^\star, G(a_k)-G(a^\star)\rangle
    \le c\|a_k-a^\star\|^2.
\]
Therefore, for $q_\delta=(1-\delta)^2+2c\delta(1-\delta)+\delta^2L^2$,
\[
    \|a_{k+1}-a^\star\|^2
    \le
    q_\delta \|a_k-a^\star\|^2.
\]
Iterating this inequality yields
$\|a_k-a^\star\|^2\le q_\delta^k\|a_0-a^\star\|^2$.

It remains to show uniqueness. Suppose \(a^\star\) and \(\bar a\) are both fixed points.
Then, by \eqref{eq:one_sided}
\[
    \|a^\star-\bar a\|^2
    =
    \langle G(a^\star)-G(\bar a),a^\star-\bar a\rangle
    \le
    c\|a^\star-\bar a\|^2.
\]
Since \(c<1\), this implies \(\|a^\star-\bar a\|=0\). 
\Halmos
\endproof





\begin{proposition}[Log-Lipschitz]
\label{prop:Log-Lip}
Let \(A\subseteq \mathbb R^d\) be nonempty, closed, and convex, and let
\(G:A\to A\). Suppose that \(G\) admits a fixed point \(a^\star\in A\).
Assume that \(G\) satisfies the incremental one-sided condition \eqref{eq:one_sided}.
Assume moreover that there exist constants \(L<\infty\)
and \(B\ge 1\) such that the log-Lipschitz condition \eqref{eq:log_lip} holds.
Fix any $q\in \bigl(0,\min\{1,1-c\}\bigr)$.
Then there exists \(\bar\delta>0\) such that, for every
\(0<\delta_0\le \bar\delta\), the relaxed fixed-point iteration with $\delta_k=\delta_0(k+1)^{-2/3}$
satisfies
\[
    \|a_k-a^\star\|^2
    \le
    (\|a_0-a^\star\|^2\vee B^2\vee 1)\exp\{-q\delta_0 k^{1/3}\},
    \quad k\ge 0.
\]
In particular, \(a_k\to a^\star\). Moreover, the fixed point is unique.
\end{proposition}

\proof{Proof.}

Let $r_k:=\|a_k-a^\star\|^2$ and $h_B(r):=1+\log_+(B/\sqrt r)$ for $r>0$.
Since \(G(a^\star)=a^\star\),
\[
    a_{k+1}-a^\star
    =
    (1-\delta_k)(a_k-a^\star)
    +
    \delta_k\bigl(G(a_k)-G(a^\star)\bigr).
\]
Expanding the squared norm gives
\[
    r_{k+1}
    =
    (1-\delta_k)^2 r_k
    +\delta_k^2\|G(a_k)-G(a^\star)\|^2 
    +2\delta_k(1-\delta_k)
    \langle a_k-a^\star,G(a_k)-G(a^\star)\rangle .
\]
By the one-sided condition \eqref{eq:one_sided},
\[
    \langle a_k-a^\star,G(a_k)-G(a^\star)\rangle
    \le c r_k.
\]
By the log-Lipschitz condition \eqref{eq:log_lip},
\[
    \|G(a_k)-G(a^\star)\|
    \le
    L\sqrt{r_k}\,h_B(r_k).
\]
Therefore,
\[
\begin{aligned}
    r_{k+1}
    &\le
    \left[
        (1-\delta_k)^2
        +2c\delta_k(1-\delta_k)
        +L^2\delta_k^2 h_B(r_k)^2
    \right]r_k  \\
    &=
    \left[
        1-2(1-c)\delta_k
        +(1-2c)\delta_k^2
        +L^2\delta_k^2h_B(r_k)^2
    \right]r_k.
\end{aligned}
\]
Let $A:=L^2+(1-2c)_+$.
Since \(h_B(r)\ge 1\),
\[
    r_{k+1}
    \le
    \left[
        1-2(1-c)\delta_k
        +A\delta_k^2 h_B(r_k)^2
    \right]r_k
    =:F_k(r_k).
\]
In the $r\to 0$ limit, we can also define $F_k(r)=0$ as the continuous limit. 
We next note that, for \(\delta_0\) sufficiently small, \(F_k\) is increasing on
\((0,\infty)\) for every \(k\). Indeed, if \(r>B^2\), then \(h_B(r)=1\), and
\[
    F_k'(r)=1-2(1-c)\delta_k+A\delta_k^2>0
\]
for all \(k\) when \(\delta_0\) is small enough. If \(0<r\le B^2\), then
\[
    h_B(r)=1+\log B-\frac12\log r,
    \quad
    h_B'(r)=-\frac{1}{2r}.
\]
Thus,
\[
\begin{aligned}
    F_k'(r)
    &=
    1-2(1-c)\delta_k
    +A\delta_k^2
    \left[
        h_B(r)^2+2rh_B(r)h_B'(r)
    \right]  \\
    &=
    1-2(1-c)\delta_k
    +A\delta_k^2
    \left[
        h_B(r)^2-h_B(r)
    \right]
    \ge
    1-2(1-c)\delta_k.
\end{aligned}
\]
Hence \(F_k\) is increasing if \(\delta_0<1/[2(1-c)]\).

Now define
$S_k:=\sum_{j=0}^{k-1}\delta_j$, $C_0:=r_0\vee B^2\vee 1$, $x_k:=C_0e^{-qS_k}$.
Then \(r_0\le x_0\). Also, since \(C_0\ge B^2\),
\[
    \log_+\frac{B}{\sqrt{x_k}}
    =
    \log_+\left(\frac{B}{\sqrt{C_0}}e^{qS_k/2}\right)
    \le
    \frac{q}{2}S_k.
\]
Because
\[
    S_k
    =
    \delta_0\sum_{j=0}^{k-1}(j+1)^{-2/3}
    \le
    3\delta_0 k^{1/3},
\]
we have
\[
    h_B(x_k)
    \le
    1+\frac{3}{2}q\delta_0 k^{1/3}.
\]
Therefore,
\[
    \delta_k h_B(x_k)^2
    \le
    \delta_0(k+1)^{-2/3}
    \left(
        1+\frac32 q\delta_0 k^{1/3}
    \right)^2 
    \le
    \delta_0
    \left(
        1+3q\delta_0+\frac94 q^2\delta_0^2
    \right).
\]
Choose \(\bar\delta>0\) small enough so that, for every
\(0<\delta_0\le \bar\delta\),
\[
    A\delta_0
    \left(
        1+3q\delta_0+\frac94 q^2\delta_0^2
    \right)
    \le
    2(1-c)-q,
\]
and also \(\delta_0<1/[2(1-c)]\).

We prove by induction that \(r_k\le x_k\) for all \(k\). The base case holds by the
definition of \(C_0\). If \(r_k\le x_k\), then, using the monotonicity of \(F_k\),
\[
\begin{aligned}
    r_{k+1}
    \le
    F_k(r_k)
    &\le
    F_k(x_k)  \\
    &\le
    \left[
        1-2(1-c)\delta_k
        +A\delta_k^2h_B(x_k)^2
    \right]x_k  \\
    &\le
    (1-q\delta_k)x_k
    \le
    e^{-q\delta_k}x_k
    =
    x_{k+1}.
\end{aligned}
\]
Thus \(r_k\le x_k\) for all \(k\). Finally,
\[
    S_k
    =
    \delta_0\sum_{j=0}^{k-1}(j+1)^{-2/3}
    \ge
    \delta_0 k\cdot k^{-2/3}
    =
    \delta_0 k^{1/3}.
\]
Therefore,
\[
    \|a_k-a^\star\|^2
    =
    r_k
    \le
    C_0e^{-qS_k}
    \le
    C_0\exp\{-q\delta_0 k^{1/3}\}.
\]

For uniqueness, if \(a^\star\) and \(\bar a\) are both fixed points, then
\[
    \|a^\star-\bar a\|^2
    =
    \langle G(a^\star)-G(\bar a),a^\star-\bar a\rangle
    \le
    c\|a^\star-\bar a\|^2,
\]
and \(c<1\) implies \(a^\star=\bar a\).
\endproof

Unlike the relaxed fixed-point iteration, whose convergence analysis is phrased directly
in terms of \(G\), gradient descent is more naturally analyzed through the residual objective
$\ell(a):=\frac12\|a-G(a)\|^2$.
Assume throughout that \(G\in C^1(\mathbb R^d;\mathbb R^d)\), so that
\[
\nabla \ell(a)=(I-\nabla G(a))^\top(a-G(a)).
\]
The next proposition gives a general linear-convergence criterion; the Jacobian condition
used below is then a simple sufficient condition for that criterion.

\begin{proposition}[Gradient descent under a PL condition] 
\label{prop:gd_pl}
Suppose \(\nabla \ell\) is \(L_0\)-Lipschitz continuous on \(\mathbb R^d\), and assume that
\(\ell\) satisfies the Polyak--\L ojasiewicz inequality \eqref{eq:PL_ineq}
Then the gradient-descent iteration with $0<\eta\leq 1/L_0$ yields
\[
\ell(a_k)\le (1-\mu\eta)^k \ell(a_0),
\qquad k\ge 0.
\]
In particular, \(\ell(a_k)\to 0\), and every accumulation point of \((a_k:k\geq 0)\) is a fixed
point of \(G\).
\end{proposition}

\proof{Proof.}
Since \(\nabla \ell\) is \(L_0\)-Lipschitz continuous, the descent lemma gives
\[
\ell(a_{k+1})
\le
\ell(a_k)-\eta\|\nabla \ell(a_k)\|^2+\frac{L_0}{2}\eta^2\|\nabla \ell(a_k)\|^2.
\]
If \(0<\eta\le 1/L_0\), then
\[
\ell(a_{k+1})
\le
\ell(a_k)-\frac{\eta}{2}\|\nabla \ell(a_k)\|^2.
\]
Using the Polyak--\L ojasiewicz inequality \eqref{eq:PL_ineq},
$\frac12\|\nabla \ell(a_k)\|^2\ge \mu \ell(a_k)$.
Then we obtain
\[
\ell(a_{k+1})
\le
\ell(a_k)-\eta\mu \ell(a_k)
=
(1-\mu\eta)\ell(a_k).
\]
Iterating yields
\[
\ell(a_k)\le (1-\mu\eta)^k\ell(a_0).
\]
Since \(\ell(a_k)\to 0\) and \(\ell\) is continuous, any accumulation point \(\bar a\) of
\(\{a_k\}\) satisfies \(\ell(\bar a)=0\), hence \(\bar a=G(\bar a)\).
\Halmos
\endproof

\paragraph{A sufficient Jacobian condition.}
Suppose \(\nabla \ell\) is \(L_0\)-Lipschitz continuous on \(\mathbb R^d\), and assume
\[
\frac{\nabla G(a)+\nabla G(a)^\top}{2}\preceq cI,
\qquad a\in\mathbb R^d,
\]
for some \(c<1\). Then \(\ell\) satisfies the Polyak--\L ojasiewicz inequality with constant
$\mu=(1-c)^2$.
\[
\frac{(I-\nabla G(a))+(I-\nabla G(a))^\top}{2}
=
I-\frac{\nabla G(a)+\nabla G(a)^\top}{2}
\succeq (1-c)I.
\]
Hence, for every \(v\in\mathbb R^d\),
\[
(1-c)\|v\|^2
\le
v^\top (I-\nabla G(a))v
\le
\|v\|\,\|(I-\nabla G(a))^\top v\|,
\]
so
$\|(I-\nabla G(a))^\top v\|\ge (1-c)\|v\|$.
Taking \(v=a-G(a)\) and using \(\nabla \ell(a)=(I-\nabla G(a))^\top (a-G(a))\), we obtain
\[
\frac12\|\nabla \ell(a)\|^2
=
\frac12\|(I-\nabla G(a))^\top (a-G(a))\|^2
\ge
\frac12(1-c)^2\|a-G(a)\|^2
=
(1-c)^2\ell(a).
\]
In particular, $\nabla \ell(a)=0$ implies $\ell(a)=0$, which further implies $a=G(a)$.
Hence the residual objective has no spurious critical points. Therefore,
Proposition~\ref{prop:gd_pl} applies with \(\mu=(1-c)^2\). 

\begin{remark}
[What remains true without a PL condition.]
If \(\nabla \ell\) is \(L_0\)-Lipschitz continuous on \(\mathbb R^d\) and
\(0<\eta\le 1/L_0\), then the same descent estimate still gives
\[
\ell(a_0)-\ell(a_k)
\ge
\frac{\eta}{2}\sum_{i=0}^{k-1}\|\nabla \ell(a_i)\|^2.
\]
Consequently,
\[
\min_{0\le i\le k-1}\|\nabla \ell(a_i)\|^2
\le
\frac{1}{k}\sum_{i=0}^{k-1}\|\nabla \ell(a_i)\|^2
\le
\frac{2\ell(a_0)}{\eta k}.
\]
Thus, even without a PL condition, gradient descent yields an \(O(1/k)\) stationarity
bound for the objective \(\ell\). However, this does not by itself imply convergence to a
fixed point of \(G\), because \(I-\nabla G(a)\) may be singular, so a small value of
\(\|\nabla \ell(a)\|\) need not force \(\|a-G(a)\|\) to be small.
\end{remark}

\subsection{Proof of Theorems \ref{th:Monotone} -- \ref{th:log-Lipschitz}}

\proof{Proof of Theorem \ref{th:Monotone}.}
We prove part~1. Fix \(a,b\in Z\) with \(a\le b\), and let
$\nu_n^a:=\nu_0P_a^n$ and $\nu_n^{b}:=\nu_0P_{b}^n$,
for an arbitrary initial distribution \(\nu_0\). We claim that
$\nu_n^a\le_{st}\nu_n^{b}$, for
$n\ge 0$.
The claim is immediate for \(n=0\). Suppose it holds for some \(n\ge 0\), and let
\(g:\mathcal S\to\mathbb R\) be bounded and increasing. Then
$\nu_{n+1}^a(g)=\nu_n^a(P_ag)$.
By Assumption~\ref{ass:increase}, the function \(P_ag\) is increasing in \(x\). Hence the
induction hypothesis gives
\[
\nu_n^a(P_ag)\le \nu_n^{b}(P_ag).
\]
Using again Assumption~\ref{ass:increase},
\[
\nu_n^{b}(P_ag)\le \nu_n^{b}(P_{b}g)=\nu_{n+1}^{b}(g).
\]
Therefore \(\nu_{n+1}^a\le_{st}\nu_{n+1}^{b}\), and the claim follows by induction.

Passing to the limit \(n\to\infty\) yields
$\pi_a\le_{st}\pi_{b}$.
Since each component of \(f\) is increasing, we obtain
\[
G(a)=\int_{\mathcal S} f(x)\,\pi_a(dx)
\le
\int_{\mathcal S} f(x)\,\pi_{b}(dx)
=G(b),
\]
where the inequality is understood coordinatewise. 

Part~2 is identical, except that Assumption~\ref{ass:decrease} reverses the second
inequality in the induction step, so one obtains
$\nu_n^a\ge_{st}\nu_n^{b}$ for $n\geq 0$,
and hence
$G(a)\ge G(b)$.
\Halmos
\endproof

\proof{Proof of Theorem \ref{th:Lipschitz}.}
Using the stationarity of \(\pi_a\) and \(\pi_{b}\),
$W_d(\pi_a,\pi_{\tilde a})
=
W_d(\pi_aP_a,\pi_{b}P_{b})$.
Thus,
\begin{align*}
W_d(\pi_a,\pi_{b})
&\le
W_d(\pi_aP_a,\pi_{b}P_a)
+
W_d(\pi_{b}P_a,\pi_{b}P_{b}) \\
&\le
\rho\,W_d(\pi_a,\pi_{b})
+
L_P\|a-b\|\int_{\mathcal S}V(x)\,\pi_{b}(dx) 
\le
\rho\,W_d(\pi_a,\pi_b) + L_PM\|a-b\|.
\end{align*}
Rearranging yields
\[
W_d(\pi_a,\pi_{b})
\le
\frac{L_P M}{1-\rho}\,\|a-b\|.
\]

Next, let \(\Gamma\) be any coupling of \(\pi_a\) and \(\pi_{b}\). Then
\begin{align*}
\|G(a)-G(b)\|
&=
\left\|\int_{\mathcal S} f(x)\,\pi_a(dx)-\int_{\mathcal S} f(y)\,\pi_{b}(dy)\right\| \\
&=
\left\|\int_{\mathcal S\times\mathcal S} \bigl(f(x)-f(y)\bigr)\,\Gamma(dx,dy)\right\| \\
&\le
\int_{\mathcal S\times\mathcal S}\|f(x)-f(y)\|\,\Gamma(dx,dy) \\
&\le
L_f\int_{\mathcal S\times\mathcal S} d(x,y)\,\Gamma(dx,dy).
\end{align*}
Taking the infimum over all couplings \(\Gamma\) gives
\[
\|G(a)-G(b)\|
\le
L_f W_d(\pi_a,\pi_{b})
\le
\frac{L_f L_PM}{1-\rho}\,\|a-b\|.
\]
\Halmos
\endproof

\proof{Proof of Theorem \ref{th:log-Lipschitz}}
Fix \(a,b\in Z\), $a\neq b$.
We first extend the one-step parameter Lipschitz condition from point masses to arbitrary
initial distributions. For any \(\mu\in\mathcal P(S)\),
\begin{equation}\label{eq:intial_Lip}
    W_d(\mu P_a,\mu P_b)
    \le
    \int_S W_d(\delta_xP_a,\delta_xP_b)\,\mu(dx)
    \le
    L_0\|a-b\|.
\end{equation}

Now fix an integer \(n\ge0\). Since \(\pi_b\) is invariant for \(P_b\),
$\pi_b=\pi_bP_b^n$.
Hence, by the triangle inequality,
\[
\begin{aligned}
    W_d(\pi_a,\pi_b)
    &\le
    W_d(\pi_a,\pi_bP_a^n)
    +
    W_d(\pi_bP_a^n,\pi_bP_b^n).
\end{aligned}
\]
The first term is controlled by uniform geometric convergence:
\[
W_d(\pi_a,\pi_bP_a^n)\le K_0e^{-\kappa n} .
\]
We next bound the second term. We first note that
\[
    W_d(\pi_bP_a^n,\pi_bP_b^n)
    \le
    \sum_{j=0}^{n-1} W_d\bigl(\pi_bP_b^jP_aP_a^{\,n-j-1},
    \pi_bP_b^jP_bP_a^{\,n-j-1}\bigr).
\]
Using the nonexpansiveness of \(P_a\), iterated \(n-j-1\) times, we have 
\[
W_d\bigl(\pi_bP_b^jP_aP_a^{\,n-j-1},
             \pi_bP_b^jP_bP_a^{\,n-j-1}\bigr)  
    \le
    W_d(\pi_bP_b^jP_a,\pi_bP_b^jP_b)  \leq L_0\|a-b\|,
\]
where the last inequality follows from \eqref{eq:intial_Lip}.
Thus,
\[
W_d(\pi_bP_a^n,\pi_bP_b^n)\le n L_0\|a-b\|.
\]
Combining the two bounds gives, for every integer \(n\ge0\),
\begin{equation}\label{eq:loglip_main}
    W_d(\pi_a,\pi_b)
    \le
    nL_{0}\|a-b\|+K_0e^{-\kappa n}.
\end{equation}

It remains to optimize over \(n\). Let
$R:=\frac{K_0\kappa/L_0}{\|a-b\|}$.
If \(R\le1\), then taking \(n=0\) in (1) gives
\[
    W_d(\pi_a,\pi_b)\le K_0\le \frac{L_0\|a-b\|}{\kappa}
    \le
    L_0\|a-b\|\left(1+\frac1\kappa\right).
\]
Since \(\log_+R=0\) in this case, this implies the desired bound.

If \(R>1\), choose
$n=\left\lceil \frac{1}{\kappa}\log R\right\rceil$.
Then
\[
    K_0e^{-\kappa n}
    \le
    K_0e^{-\log R}
    =
    \frac{K_0}{R}
    =
    \frac{L_0\|a-b\|}{\kappa}.
\]
Substituting this choice of \(n\) into \eqref{eq:loglip_main}, we obtain
\[
W_d(\pi_a,\pi_b)
    \le
    L_0\|a-b\|\left(1+\frac{1}{\kappa}\log R\right)
    +
    \frac{L_0\|a-b\|}{\kappa} 
    =
    L_0\|a-b\|\left(
        1+\frac{1+\log R}{\kappa}
    \right).
\]
Since \(\log R>0\),
\[
    1+\frac{1+\log R}{\kappa}
    \le
    \left(1+\frac1\kappa\right)(1+\log R).
\]
Therefore,
\[
    W_d(\pi_a,\pi_b)
    \le
    L_0\|a-b\|\left(1+\frac1\kappa\right)(1+\log R).
\]
\endproof

\section{Gradient Estimation Appendix} \label{app:gradient_estimation}
Let $\psi(y):=-\mathbb E[\eta(Y_1)\mid Y_0=y]$.

\subsection{Auxiliary Results}

The following lemma is needed for the proof of Theorem  \ref{thm:fractional-contraction}.

\begin{lemma}\label{lem:Poisson-Lip}
Suppose Assumption \ref{ass:ipa-stability} (1) holds. If $b: Y\to\mathbb R$ is Lipschitz and $\mathbb E[b(Y_\infty)]=0$, then the solution $\phi$ to $\phi(y)-\mathbb E[\phi(Y_1)\mid Y_0=y]=b(y)$ is Lipschitz in $d_{\calY}$. Moreover $\tilde\phi(y):=-\mathbb E[\phi(Y_1)\mid Y_0=y]$ is Lipschitz as well.
\end{lemma}
\proof{Proof.}
Write $\phi(y)=\sum_{k\ge 0}\mathbb E[b(Y_k)\mid Y_0=y]$. Under the coupling in Assumption \ref{ass:ipa-stability} (1),
\[
|\phi(y)-\phi(\tilde y)|
\leq \sum_{k\ge 0} L_b\,\mathbb E[d_{\calY}(Y_k,\tilde Y_k)]
\leq L_b\,C_\gamma\!\sum_{k\ge 0}\gamma^{-k} d_{\calY}(y,\tilde y)=\frac{L_b\gamma C_\gamma}{\gamma-1}d_{\calY}(y,\tilde y).
\]
The same argument applies to $\tilde\phi(y)=-\sum_{k\ge 1}\mathbb E[b(Y_k)\mid Y_0=y]$.
\Halmos
\endproof
\proof{Proof of Proposition \ref{prop:reset}.}
The product $\Gamma_{0:n}$ contains a zero factor whenever $\tau_0<n$. Hence
$\|\Gamma_{0:n}\|^q=\|\Gamma_{0:n}\|^q\mathbf 1\{\tau_0\ge n\}$.
Then
\[
    \mathbb E\left[
        \|\Gamma_{0:n}\|^q
        \mid Y_0=y
    \right]
    \le
    M_\Gamma^{qn}
    \mathbb P(\tau_0\ge n\mid Y_0=y)
    \le
    C_\tau R(y)(\rho M_\Gamma^q)^n.
\]
Choose $\gamma>1$ such that $\gamma^{-q}=\rho M_\Gamma^q$. \Halmos
\endproof

\proof{Proof of Proposition  \ref{lem:prod}}
With $g(y)=\log\|\Gamma(y)\|$ and $\psi(y)=-\E[\eta(Y_1)\mid Y_0=y]$, we have $\ g(Y_{t+1})+\psi(Y_{t+1})=-\log\alpha-\eta(Y_{t+1})$. Thus, for any $s\ge t$,
\begin{align*}
\E_s\exp\!\big(2q(g(Y_{s+1})+\psi(Y_{s+1}))\big)
&=\exp(-2q\log\alpha)\, \E_s \exp\!\big(-2q\,\eta(Y_{s+1})\big)\\
&\le \exp(-2q\log\alpha)\,\exp\!\Big(-2q\,\E_s\eta(Y_{s+1})+4c_\eta q^2\Big)\\
&= \exp\!\Big(-2q\log\alpha + 4c_\eta q^2\Big)\,\exp\!\big(2q\,\psi(Y_s)\big),
\end{align*}
where we used Assumption \ref{ass:IPAweakaverage} Condition 2 with $2q<q_0$.
Iterating this bound yields
\begin{align}
&\E_t \exp\left(2q \sum_{i=1}^{n}g(Y_{t+i}) + 2q\psi(Y_{t+n})\right)\nonumber\\
=&\E_t\left[
\exp\big(2q \sum_{i=1}^{n-1}g(Y_{t+i})\big)
\E_{t+n-1}\exp\big(2q(g(Y_{t+n})+\psi(Y_{t+n}))\big)
\right]\nonumber\\
\leq& \exp(-2q\log\alpha + 4c_\eta q^2)
\E_t \exp\big(2q \sum_{i=1}^{n-1}g(Y_{t+i}) + 2q\psi(Y_{t+n-1})\big)\nonumber\\
\leq& \exp\!\Big((- 2q\log\alpha + 4c_\eta q^2 )\,n\Big)\,\exp\!\big(2q\psi(Y_t)\big). \label{eq:iterative}
\end{align}
Next, using
\[
\prod_{i=1}^n \|\Gamma(Y_{t+i})\|^q
= \exp\!\left(q \sum_{i=1}^{n}g(Y_{t+i})\right)
= \exp\!\left(q \sum_{i=1}^{n}g(Y_{t+i})+q\psi(Y_{t+n})\right)\,e^{-q\psi(Y_{t+n})}.
\]
Cauchy–Schwarz gives
\begin{align*}
\E_t\!\left[\prod_{i=1}^n \|\Gamma(Y_{t+i})\|^q\right]
&= \E_t\!\left[\exp\left(q \sum_{i=1}^{n}g(Y_{t+i})+q\psi(Y_{t+n})\right)\,\exp(-q\psi(Y_{t+n}))\right]\\
&\le \E_t\!\left[\exp\left(2q \sum_{i=1}^{n}g(Y_{t+i})+2q\psi(Y_{t+n})\right)\right]^{\!1/2}\,
     \E_t\!\left[\exp(-2q\psi(Y_{t+n}))\right]^{\!1/2}\,.
\end{align*}
Combining this with the bound in \eqref{eq:iterative}, we have
\[
\E_t\!\left[\prod_{i=1}^n \|\Gamma(Y_{t+i})\|^q\right]
\;\le\; \exp\!\Big((2c_\eta q^2 - q\log\alpha)\,n\Big)\,
\exp\big(q\,\psi(Y_t)\big)\,
\E_t[\exp(-2q\psi(Y_{t+n}))]^{1/2}.
\]
Then we note that by Jensen's inequality and Condition 3
\[
\exp(q\psi(Y_t))=\exp(-q\E[\eta(Y_{t+1})|Y_t])\leq 
\E[\exp(-q\eta(Y_{t+1}))|Y_t]\leq 
C[\exp(-q\eta(Y_{t}))+1].
\]
\[
\begin{split}
\E_t[\exp(-2q\psi(Y_{t+n}))]^{1/2}
&=
\E_t[\exp(2q\E[\eta(Y_{t+n+1})|Y_{t+n}])]^{1/2}\\
&\leq
\E_t[\exp(2q\eta(Y_{t+n+1}))|Y_t]\leq 
C[\exp(q\eta(Y_{t}))+1].
\end{split}
\]

Finally, choosing $\alpha_0\in(1,\alpha)$ and take $q>0$ small enough so that $\log \alpha_0 \leq \log \alpha - 2c_\eta q$, i.e., $q\leq \log(\alpha/\alpha_0)/(2c_\eta)$, we have
\[
\exp\!\Big((2c_\eta q^2 - q\log\alpha)\,n\Big)\leq \alpha_0^{-q n}.
\]
This yields the conditional product bound \eqref{eq:prod-conditional}, i.e.,
\[
\E_t\!\left[\prod_{i=1}^{n}\|\Gamma(Y_{t+i})\|^{q}\right]
\;\le\; C\alpha_0^{-q n}(\exp(q\eta(Y_t))+1)(\exp(-q\eta(Y_t))+1)=
4C a_0^{-qn}\cosh(q\eta(Y_t)/2)^2,
\]
which is finite and uniformly bounded in $t$ by assumption. 
\Halmos\endproof

\subsection{Proof of Theorem \ref{thm:fractional-contraction}}
Throughout the proof, let $\mathbb E_0[\cdot]:=
    \mathbb E[\cdot\mid
        Y_0=y_0,\widetilde Y_0=\widetilde y_0].$
We repeatedly use three consequences of Assumption~\ref{ass:ipa-stability}.
First, by the weak average contraction condition, the Markov property, and
Jensen's inequality, for any \(r\in(0,\min\{q_0,1\}]\),
\[
    \mathbb E_0\|\Gamma_{a:b}\|^r
    \le
    C\gamma^{-r(b-a)}V(y_0)^r,
    \qquad
    \mathbb E_0\|\widetilde\Gamma_{a:b}\|^r
    \le
    C\gamma^{-r(b-a)}V(\widetilde y_0)^r .
\]
Indeed, conditioning on \(Y_a\) gives
\[
    \mathbb E_0\|\Gamma_{a:b}\|^r
    \le
    C\gamma^{-r(b-a)}
    \mathbb E_0[V(Y_a)^r],
\]
and since \(r\le1\),
\[
    \mathbb E_0[V(Y_a)^r]
    \le
    \left(\mathbb E_0[V(Y_a)]\right)^r
    \le
    C V(y_0)^r .
\]
The bound for \(\widetilde\Gamma_{a:b}\) is identical.
Second, let $q=2p$, since \(2q\le 1\), the coupling contraction implies
\[
    \mathbb E_0\bigl[
        d_{\calY}(Y_t,\widetilde Y_t)^{2q}
    \bigr]
    \le
    \left(
        \mathbb E_0[
            d_{\calY}(Y_t,\widetilde Y_t)
        ]
    \right)^{2q}
    \le
    C\gamma^{-2qt}
    d_{\calY}(y_0,\widetilde y_0)^{2q}.
\]
Third, by the transient moment condition for \(\xi\),
$\sup_{t\ge0}
    \mathbb E_0\|\xi(Y_t)\|^{2q}
    \le
    C V(y_0)^{2q}$.

From \eqref{eq:variation-of-constants},
\[
Z_T-\tilde Z_T=\sum_{m=0}^{T-1}\Big((\Gamma_{m+1:T}-\tilde \Gamma_{m+1:T})\xi(Y_m)+\tilde \Gamma_{m+1:T}\big(\xi(Y_m)-\xi(\tilde Y_m)\big)\Big).
\]
Since $q/2<1$, the map \(x\mapsto x^{q/2}\) is subadditive on
\(\mathbb R_+\). Hence
\[
\mathbb E_0[\|Z_T-\tilde Z_T\|^{q/2}]
\le \sum_{m=0}^{T-1} \left(\mathbb E_0\left\|(\Gamma_{m+1:T}-\tilde \Gamma_{m+1:T})\xi(Y_m)\right\|^{q/2}
+\E_0\left\|\tilde \Gamma_{m+1:T}\big(\xi(Y_m)-\xi(\tilde Y_m)\big)\right\|^{q/2}\right).
\]
We next bound the two terms separately.

First, note that by Assumption \ref{ass:ipa-stability}, Conditions 2 and 4, and Cauchy-Schwarz,
\begin{align}
\notag
\mathbb E_0[\|\tilde \Gamma_{m+1:T}(\xi(Y_m)-\xi(\tilde Y_m))\|^{q/2}]
&\le L_\xi^{q/2} \Big(\mathbb E_0[\|\widetilde\Gamma_{m+1:T}\|^q]\Big)^{1/2}\Big(\mathbb E_0 [d_{\calY}(Y_m,\tilde Y_m)^{q}]\Big)^{1/2} \\
\notag
&\leq C
\gamma^{-q(T-m-1)/2}
V(\widetilde y_0)^{q/2}
\gamma^{-qm/2}
d_{\calY}(y_0,\widetilde y_0)^{q/2}\\
\label{eq:bound01}
&\leq C\gamma^{-qT/2} V(\tilde{y}_0)^{q/2} d_{\calY}(y_0,\tilde y_0)^{q/2},
\end{align}
where the harmless factor \(\gamma^{q/2}\) is absorbed into \(C\).


Second, for  \(m\le T-2\), the telescoping identity (valid for right-to-left products) gives
\[
\Gamma_{m+1:T}-\tilde \Gamma_{m+1:T}
=\sum_{k=m+1}^{T-1}\tilde \Gamma_{k+1:T}\,\big(\Gamma(Y_k)-\Gamma(\tilde Y_k)\big)\,\Gamma_{m+1:k}.
\]
For \(m=T-1\), this sum is empty and the corresponding term is zero.  Therefore,
using subadditivity and the Lipschitz continuity of \(\Gamma\),
\[
\left\|
    \bigl(\Gamma_{m+1:T}-\widetilde\Gamma_{m+1:T}\bigr)\xi(Y_m)
\right\|^{q/2}  
\le
L_\Gamma^{q/2}
\sum_{k=m+1}^{T-1}
\|\widetilde\Gamma_{k+1:T}\|^{q/2}
\|\Gamma_{m+1:k}\|^{q/2}
d_{\calY}(Y_k,\widetilde Y_k)^{q/2}
\|\xi(Y_m)\|^{q/2}.
\]
Fix \(m\) and \(k\). By Cauchy--Schwarz,
\[
\begin{aligned}
&\mathbb E_0\left[
\|\widetilde\Gamma_{k+1:T}\|^{q/2}
\|\Gamma_{m+1:k}\|^{q/2}
d_{\calY}(Y_k,\widetilde Y_k)^{q/2}
\|\xi(Y_m)\|^{q/2}
\right]  \\
\le&
\left(
    \mathbb E_0[
        \|\widetilde\Gamma_{k+1:T}\|^q
        \|\Gamma_{m+1:k}\|^q
    ]
\right)^{1/2}
\left(
    \mathbb E_0[
        d_{\calY}(Y_k,\widetilde Y_k)^q
        \|\xi(Y_m)\|^q
    ]
\right)^{1/2}.
\end{aligned}
\]
Applying Cauchy--Schwarz once more to each factor,
\[
\begin{aligned}
\left(
    \mathbb E_0[
        \|\widetilde\Gamma_{k+1:T}\|^q
        \|\Gamma_{m+1:k}\|^q
    ]
\right)^{1/2}
&\le
\left(
    \mathbb E_0\|\widetilde\Gamma_{k+1:T}\|^{2q}
\right)^{1/4}
\left(
    \mathbb E_0\|\Gamma_{m+1:k}\|^{2q}
\right)^{1/4}  \\
&\le
C
\gamma^{-q(T-k-1)/2}
\gamma^{-q(k-m-1)/2}
V(\widetilde y_0)^{q/2}
V(y_0)^{q/2},
\end{aligned}
\]
and
\[
\begin{aligned}
\left(
    \mathbb E_0[
        d_{\calY}(Y_k,\widetilde Y_k)^q
        \|\xi(Y_m)\|^q
    ]
\right)^{1/2}
&\le
\left(
    \mathbb E_0[
        d_{\calY}(Y_k,\widetilde Y_k)^{2q}
    ]
\right)^{1/4}
\left(
    \mathbb E_0[
        \|\xi(Y_m)\|^{2q}
    ]
\right)^{1/4}  \\
&\le
C
\gamma^{-qk/2}
d_{\calY}(y_0,\widetilde y_0)^{q/2}
V(y_0)^{q/2}.
\end{aligned}
\]
Combining the last two displays, then
\[\begin{split}
    &\mathbb E_0\left[
\|\widetilde\Gamma_{k+1:T}\|^{q/2}
\|\Gamma_{m+1:k}\|^{q/2}
d_{\calY}(Y_k,\widetilde Y_k)^{q/2}
\|\xi(Y_m)\|^{q/2}
\right]  \\
\leq & C
    \gamma^{-q(T-k-1)/2}
    \gamma^{-q(k-m-1)/2}
    \gamma^{-qk/2}
    d_{\calY}(y_0,\widetilde y_0)^{q/2}
    V(\widetilde y_0)^{q/2}
    V(y_0)^q .
\end{split}\]
Summing over \(k=m+1,\ldots,T-1\), we obtain
\begin{align}
\notag
\mathbb E_0
\left[
    \left\|
    \bigl(\Gamma_{m+1:T}-\widetilde\Gamma_{m+1:T}\bigr)\xi(Y_m)
    \right\|^{q/2}
\right]  
\le&
C
d_{\calY}(y_0,\widetilde y_0)^{q/2}
V(\widetilde y_0)^{q/2}
V(y_0)^q
\gamma^{-q(T-m-2)/2}
\sum_{k=m+1}^{T-1}\gamma^{-qk/2}  \\
\notag
\le&
C
d_{\calY}(y_0,\widetilde y_0)^{q/2}
V(\widetilde y_0)^{q/2}
V(y_0)^q
\gamma^{-q(T-m-2)/2}\gamma^{-q(m+1)/2} \\
\label{eq:bound02}
\le&
C
\gamma^{-qT/2}
d_{\calY}(y_0,\widetilde y_0)^{q/2}
V(\widetilde y_0)^{q/2}
V(y_0)^q .
\end{align}
The final inequality again absorbs constants depending only on \(q\) and
\(\gamma\).

Combining the bounds in \eqref{eq:bound01} and \eqref{eq:bound02},
and summing $m=0,\dots,T-1$ yields
\[
\mathbb E \|Z_T-\tilde Z_T\|^{q/2}\leq C \gamma^{-q T/2}Td_{\calY}(y_0,\tilde y_0)^{q/2}V(\tilde{y}_0)^{q/2}(1+V(y_0)^q).
\]
\Halmos

\subsection{From sensitivity contraction to summand contraction}\label{app:summand-contraction}
Fix
$s\in\left(0,\min\left\{q_0/4,1/4\right\}\right)$.
Suppose that
$\|\nabla f(x)\|_{\mathrm{op}}\le M_f$ and
$\|\nabla f(x)-\nabla f(\widetilde x)\|_{\mathrm{op}}
    \le
    L_f d_{\mathcal Y}(y,\widetilde y)$,
where \(x\) and \(\widetilde x\) are the state components of \(y\) and
\(\widetilde y\), respectively. Then
\[
\begin{aligned}
    \|h_a(y,z)-h_a(\widetilde y,\widetilde z)\|_F
    &\le
    M_f\|z-\widetilde z\|_F
    +
    L_f d_{\mathcal Y}(y,\widetilde y)\|\widetilde z\|_F .
\end{aligned}
\]

To compare a transient sensitivity \(Z_t\) with a stationary sensitivity
\(\widetilde Z_t\), let \(\overline Z_t\) be driven by the stationary
\(\widetilde Y\)-path but initialized from \(\overline Z_0=0\). Since
\(s\le1\),
\[
    \mathbb E\|Z_t-\widetilde Z_t\|_F^s
    \le
    \mathbb E\|Z_t-\overline Z_t\|_F^s
    +
    \mathbb E\|\widetilde\Gamma_{0:t}\widetilde Z_0\|_F^s.
\]
The first term is bounded by
Theorem~\ref{thm:fractional-contraction}. The second is bounded by the
Jacobian-product estimate in
Assumption~\ref{ass:ipa-stability}, together with a finite stationary
\(2s\)-moment of \(Z\). Moreover,
\[
\begin{aligned}
    \mathbb E\left[
        d_{\mathcal Y}(Y_t,\widetilde Y_t)^s
        \|\widetilde Z_t\|_F^s
    \right]
    &\le
    \left(
        \mathbb E d_{\mathcal Y}(Y_t,\widetilde Y_t)^{2s}
    \right)^{1/2}
    \left(
        \mathbb E\|\widetilde Z_t\|_F^{2s}
    \right)^{1/2}.
\end{aligned}
\]
The geometric coupling of the driving process controls the first factor,
while stationarity controls the second. Consequently, under suitable
integrability of the initial-state factors,
\[
    \mathbb E\left[
        \|h_a(Y_t,Z_t)
          -h_a(\widetilde Y_t,\widetilde Z_t)\|_F^s
    \right]
    \le
    Ct\gamma^{-st}.
\]
Thus, Condition~3 of
Assumption~\ref{ass:lyap-untruncated-ipa} holds with
\(\beta=\gamma^s\). In particular, the required stationary sensitivity
moment follows from Condition~1 whenever
\(\|z\|_F^2\le C_W(1+W(y,z))\).

\subsection{Proof of Theorem \ref{thm:untruncated-ipa-lyap}.}
Recall the decomposition \eqref{eq:decomp}. We first bound the Monte Carlo variance. Iterating
the drift inequality \eqref{eq:Wbound01} gives
\[
    \sup_{t\ge0}
    \mathbb E_\mu[W(Y_t,Z_t)]
    \le
    \mathbb E_\mu[W(Y_0,Z_0)]
    +
    \frac{b}{1-\rho}
    <\infty.
\]
For the stationary copy, invariance and the Lyapunov drift condition
imply
\[
    \mathbb E_{\Pi_a}[W(Y,Z)]
    \le
    \frac{b}{1-\rho}.
\]
Here the inequality is justified by applying the drift condition to
the truncations \(W\wedge M\) and then letting \(M\to\infty\). Hence \eqref{eq:W_secondm} implies
\[
    \sup_{t\ge0}
    \mathbb E\|h_a(Y_t,Z_t)\|^2<\infty
    \quad\mbox{ and }\quad
    \sup_{t\ge0}
    \mathbb E\|h_a(\widetilde Y_t,\widetilde Z_t)\|^2
    <\infty.
\]
Since the \(N\) simulation replications
are independent,
\[
    \mathbb E\left[
        \left\|
            \widehat g_{N,t}(a)-\mathbb E[h_a(Y_t,Z_t)]
        \right\|^2
    \right]
    =
    \frac1N
    \mathbb E\left[
        \left\|
            h_a(Y_t,Z_t)-\mathbb E[h_a(Y_t,Z_t)]
        \right\|^2
    \right]  
    \le \frac{C}{N}.
\]

Second, we control the finite-time bias. By the stationarity of the comparison copy,
\[
    \left\|
        \mathbb E[h_a(Y_t,Z_t)]-\mathbb E_{\Pi_a}[h_a(Y,Z)]
    \right\|
    =
    \left\|
        \mathbb E[h_a(Y_t,Z_t)]-\mathbb E[h_a(\widetilde Y_t,\widetilde Z_t)]
    \right\|  
    \le
    \mathbb E\|h_a(Y_t,Z_t)-h_a(\widetilde Y_t,\widetilde Z_t)\|.
\]
Since \(s<1\),
\[
    \|h_a(Y_t,Z_t)-h_a(\widetilde Y_t,\widetilde Z_t)\|
    \le
    \|h_a(Y_t,Z_t)-h_a(\widetilde Y_t,\widetilde Z_t)\|^{s/2}
    \bigl(\|h_a(Y_t,Z_t)\|+\|h_a(\widetilde Y_t,\widetilde Z_t)\|\bigr)^{1-s/2}.
\]
By Cauchy--Schwarz,
\[
\begin{aligned}
    \mathbb E\|h_a(Y_t,Z_t)-h_a(\widetilde Y_t,\widetilde Z_t)\|
    &\le
    \left(
        \mathbb E\|h_a(Y_t,Z_t)-h_a(\widetilde Y_t,\widetilde Z_t)\|^{s}
    \right)^{1/2}
    \left(
        \mathbb E
        \bigl[
            (\|h_a(Y_t,Z_t)\|+\|h_a(\widetilde Y_t,\widetilde Z_t)\|)^{2-s}
        \bigr]
    \right)^{1/2}.
\end{aligned}
\]
The second factor is uniformly bounded by the uniform second-moment bound,
because \(2-s<2\). The first factor is controlled by the third condition in Assumption~\ref{ass:lyap-untruncated-ipa}. Therefore,
\[
    \mathbb E\|h_a(Y_t,Z_t)-h_a(\widetilde Y_t,\widetilde Z_t)\|
    \le
    C \sqrt t\beta^{-t/2}.
\]
Consequently,
\[
    \left\|
        \mathbb E[h_a(Y_t,Z_t)]-\mathbb E_{\Pi_a}[h_a(Y,Z)]
    \right\|^2
    \le
    C t \beta^{-t}.
\]

Combining the variance and squared bias bounds gives
\[
    \mathbb E\left[
        \left\|
            \widehat g_{N,t}(a)-\nabla G(a)
        \right\|^2
    \right]
    \le
    C\left(
        \frac1N
        +
        t\beta^{-t}
    \right).
\]
\Halmos
\endproof

\section{G/G/c Queue Appendix} \label{app:GGC}
\subsection{Uniform ergodicity and stationary perturbation}

We begin with two elementary consequences of the Kiefer--Wolfowitz recursion.

\begin{lemma}[One-step nonexpansiveness and parameter sensitivity]
\label{lem:ggc-one-step}
For every $a,b\in Z$, $w,x,y\in\mathcal S$, and probability measures $\mu,\nu$ with finite first moments,
\begin{align}
    W_1(\delta_wP_a,\delta_wP_b)
    &\le cL_T\|a-b\|,                                             \label{eq:ggc-one-step-parameter}\\
    W_1(\mu P_a,\nu P_a)
    &\le W_1(\mu,\nu),                                            \label{eq:ggc-one-step-state}
\end{align}
where $c$ is the number of servers and $L_T$ is defined in Assumption \ref{ass:ggc}.
\end{lemma}

\proof{Proof.}
Since both the coordinatewise positive-part and sorting maps are nonexpansive in $L_1$,
\[
    \|F(w,s,t)-F(w,s,t')\|_1
    \le c|t-t'|, \quad
    \|F(w,s,t)-F(w',s,t)\|_1
    \le \|w-w'\|_1.
\]
For \eqref{eq:ggc-one-step-parameter}, optimally couple $T(a)$ and $T(b)$ in $W_1$ and use the same service time in both systems. For \eqref{eq:ggc-one-step-state}, start with an arbitrary coupling of $X\sim\mu$ and $Y\sim\nu$ and use common service and interarrival times. Taking expectations and then the infimum over initial couplings proves the result.\Halmos
\endproof

The next lemma supplies a uniform exponential drift.

\begin{lemma}[Uniform exponential drift]
\label{lem:ggc-drift}
Under Assumption~\ref{ass:ggc}(2)--(3), there exist $\theta_\star>0$ and $\eta\in(0,1)$ such that, with
$V(w)=\frac1c\sum_{i=1}^ce^{\theta_\star w_i}$,
we have
\begin{equation}
\label{eq:ggc-drift}
    P_aV(w)\le 1+\eta V(w),
    \qquad a\in A,\ w\in\mathcal S.
\end{equation}
\end{lemma}

\proof{Proof.}
For $\theta\in[0,\bar\theta]$, the inequality
$e^{\theta(x-t)^+}\le 1+e^{-\theta t}e^{\theta x}$
and symmetry of $V_\theta(w):=c^{-1}\sum_i e^{\theta w_i}$ under permutations give
\[
    P_aV_\theta(w)
    \le
    1+\mathbb E[e^{-\theta T(a)}]
    \left(1+\frac{\mathbb E[e^{\theta S}]-1}{c}\right)V_\theta(w)
    \le
    1+\varphi(\theta)V_\theta(w),
\]
where
\[
    \varphi(\theta)
    :=\mathbb E[e^{-\theta T_\star}]
      \left(1+\frac{\mathbb E[e^{\theta S}]-1}{c}\right).
\]
Here we used $T(a)\succeq_{\rm st}T_\star$. Because $\varphi(0)=1$ and
\[
    \varphi'(0)
    =-\mathbb ET_\star+\frac{\mathbb ES}{c}
    \le-\epsilon\mathbb ET_\star<0,
\]
there exists $\theta_\star>0$ sufficiently small that $\eta:=\varphi(\theta_\star)<1$. \Halmos
\endproof

We next verify a strict contraction on a Lyapunov sublevel set. 

\begin{lemma}[Reset contraction on the coupling set]
\label{lem:ggc-reset-coupling}
Let $\theta_\star$ and $\eta$ be as in Lemma~\ref{lem:ggc-drift}. There exist $d>2/(1-\eta)$ and $p_\star>0$, independent of $a$, such that, for
$\mathcal C:=\{(x,y):V(x)+V(y)<d\}$,
we have
\begin{align}
    W_1(P_a(x,\cdot),P_a(y,\cdot))
    &\le(1-p_\star)\|x-y\|_1,
    &&(x,y)\in\mathcal C,                                      \label{eq:ggc-local-contract}\\
    W_1(P_a(x,\cdot),P_a(y,\cdot))
    &\le\|x-y\|_1,
    &&(x,y)\notin\mathcal C.                                   \label{eq:ggc-global-nonexpand}
\end{align}
\end{lemma}

\proof{Proof.}
Choose $a_0$ large enough that
\begin{equation}
\label{eq:a0-choice}
    a_0>\frac c{\theta_\star}\log\frac{2c}{1-\eta}.
\end{equation}
This enlargement is harmless because $T_\star$ has unbounded support. Let
$C_0:=\{w\in\mathcal S:\|w\|_1\le a_0\}$.
By Assumption~\ref{ass:ggc}(3), 
\[
    p_\star
    :=\mathbb P\{S\le s_0\}\mathbb P\{T_\star\ge a_0+s_0\}>0.
\]
For any $x\in C_0$, the event
$E:=\{S\le s_0,\ T(a)\ge a_0+s_0\}$
forces $F(x,S,T(a))=0$. Therefore, under the common-noise coupling and for $x,y\in C_0$,
\[
    \mathbb E\|F(x,S,T(a))-F(y,S,T(a))\|_1
    \le(1-p_\star)\|x-y\|_1.
\]

Choose
$d\in\left(\frac{2}{1-\eta},\frac1c e^{\theta_\star a_0/c}\right]$.
The interval is nonempty by \eqref{eq:a0-choice}. If $V(w)<d$, then
$e^{\theta_\star w_i}\le cV(w)<cd$ for every $i$, and hence
\[
    \|w\|_1
    \le\frac c{\theta_\star}\log(cd)
    \le a_0.
\]
Thus $\mathcal C\subseteq C_0\times C_0$, and the reset coupling proves \eqref{eq:ggc-local-contract}. Equation \eqref{eq:ggc-global-nonexpand} is Lemma~\ref{lem:ggc-one-step}. \Halmos
\endproof

\proof{Proof of Theorem~\ref{thm:ggc-perturbation}.}
We apply the generalized drift-and-contraction criterion of \citep{qin2022geometric} with the metric
$\psi(x,y):=\|x-y\|_1$.
The metric is controlled by the Lyapunov function because, for $u\ge0$, $u\le e^{\theta_\star u}/\theta_\star$, and therefore
\begin{equation}
\label{eq:metric-dominated-by-V}
    \|x-y\|_1
    \le\|x\|_1+\|y\|_1
    \le\frac c{\theta_\star}\bigl(V(x)+V(y)\bigr).
\end{equation}
Lemma~\ref{lem:ggc-drift} gives the required drift inequality. Lemma~\ref{lem:ggc-reset-coupling} gives strict contraction on the prescribed Lyapunov sublevel set and global nonexpansiveness outside it. Because the outside factor is one, the remaining compatibility condition in the generalized contraction theorem is automatic. All constants are uniform in $a$.

Corollary 2.1 together with Propositions 2.9--2.10 of \citep{qin2022geometric} therefore gives, for every $a\in A$, a unique invariant distribution $\pi_a$ and constants $C_0<\infty$ and $\rho\in(0,1)$, independent of $a$, such that
\begin{equation}
\label{eq:ggc-pointwise-W1}
    W_1(\delta_wP_a^n,\pi_a)
    \le C_0(1+V(w))\rho^n,
    \qquad w\in\mathcal S.
\end{equation}

We next establish the stationary $V$-moment. Starting from $W_0=0$, iteration of \eqref{eq:ggc-drift} gives
\[
    \sup_{a\in A}\sup_{n\ge0}\mathbb E_0[V(W_n(a))]
    \le V(0)+\frac1{1-\eta}<\infty.
\]
Equation \eqref{eq:ggc-pointwise-W1} implies weak convergence of $\delta_0P_a^n$ to $\pi_a$. Since $V$ is nonnegative and lower semicontinuous, the Portmanteau theorem yields
\[
    \pi_aV
    \le\liminf_{n\to\infty}\mathbb E_0[V(W_n(a))]
    \le V(0)+\frac1{1-\eta}.
\]
Hence $\sup_a\pi_aV<\infty$. Applying the same Wasserstein contraction bound to an initial law $\nu$ and using $\nu V<\infty$ now gives
\begin{equation}
\label{eq:ggc-uniform-W1}
    W_1(\nu P_a^n,\pi_a)
    \le C_0(1+\nu V)\rho^n.
\end{equation}

It remains to prove the stationary perturbation bound. By Lemma~\ref{lem:ggc-one-step},
\[
W_1(\delta_wP_a,\delta_wP_b)\le L_0\|a-b\|,
    \quad L_0:=cL_T,
    \quad
    W_1(\mu P_a,\nu P_a)\le W_1(\mu,\nu).
\]
Moreover, \eqref{eq:ggc-uniform-W1} and the uniform stationary $V$-moment imply that, for some $K_0<\infty$,
\[
    W_1(\pi_bP_a^n,\pi_a)\le K_0\rho^n,
    \qquad a,b\in A.
\]
Thus the three conditions of Theorem~\ref{th:log-Lipschitz} hold with the unbounded $L_1$ metric. Applying that theorem along with Remark \ref{rem:bunb} gives constants $C,B<\infty$ such that
\[
    W_1(\pi_a,\pi_b)
    \le C\|a-b\|
    \left(1+\log_+\frac{B}{\|a-b\|}\right).
\]
The stated bound for $G=H\circ\pi$ follows immediately when $H$ is Lipschitz in $W_1$. \Halmos
\endproof

\proof{Proof of Corollary~\ref{cor:ggc-existence}.}
Since
\[
    w_1^k
    \le \frac{k!}{\theta_\star^k}e^{\theta_\star w_1}
    \le \frac{k!c}{\theta_\star^k}V(w),
\]
we have $\sup_{a\in Z}\int w_1^k\,\pi_a(dw)<\infty$.
Assumption~\ref{ass:ggc}(4) therefore gives a constant $R_0<\infty$ such that
$\sup_{a\in Z}\|G(a)\|\le R_0$.
Choose $a_0\in A$ and define
$K:=Z\cap\{a\in\mathbb R^d:\|a\|\le\max\{R_0,\|a_0\|\}\}$.
Because $A$ is nonempty, closed, and convex, $K$ is nonempty, compact, and convex. Since $G(a)\in A$ and $\|G(a)\|\le R_0$, we have $G(K)\subseteq K$.

To show continuity, let $a_n\to a$. Theorem~\ref{thm:ggc-perturbation} gives $W_1(\pi_{a_n},\pi_a)\to0$, hence $\pi_{a_n}\Rightarrow\pi_a$. The uniform exponential moment gives uniform integrability of every fixed polynomial moment, so Assumption~\ref{ass:ggc}(4) yields
$G(a_n)=H(\pi_{a_n})\to H(\pi_a)=G(a)$. Brouwer's fixed-point theorem now gives $a^\star\in K$ with $G(a^\star)=a^\star$.\Halmos
\endproof

\subsection{Scalar Strategic G/G/c Queue}
\proof{Proof of Proposition \ref{prop:scalar-ggc}.}
Fix $a\le b$. By the stochastic-order assumption, one may couple the interarrival times so that
$T(b)\ge T(a)$ almost surely. Use the same service requirements in both frozen systems and initialize both at zero. The Kiefer--Wolfowitz map is coordinatewise increasing in the workload and decreasing in the interarrival time. Induction therefore gives
$W_n(b)\le W_n(a)$ coordinatewise for every $n$.
Passing to the stationary limit gives stochastic dominance of the waiting-time marginals:
$\pi_b^{(1)}\preceq_{\rm st}\pi_a^{(1)}$. Since $H$ is increasing in this order,
\[
    G(b)=H(\pi_b)\le H(\pi_a)=G(a).
\]
Thus $G$ is decreasing. Consequently,
\[
    (G(a)-G(b))(a-b)\le0,
\]
so the incremental one-sided condition holds with constant zero. Theorem~\ref{thm:ggc-perturbation} and the $W_1$-Lipschitz property of $H$ give the required log-Lipschitz bound. Proposition~\ref{prop:Log-Lip} therefore yields convergence of RFPI with $\delta_k=\delta_0(k+1)^{-2/3}$ for sufficiently small $\delta_0$.

Finally, a decreasing scalar map has at most one fixed point: if $a<b$ were both fixed points, then $a=G(a)\ge G(b)=b$, a contradiction. Existence follows from Corollary~\ref{cor:ggc-existence}. \Halmos
\endproof

\subsection{Mean-Variance Feedback M/M/1 Queue}
Define the active region
$A_0:=\{a\in K:\lambda(a)>0\}$.
On $A_0$, direct differentiation of \eqref{eq:mm1-response} gives
\begin{align}
    \partial_{a_1}G_1(a)
    &=-\frac{\alpha_1+\gamma\rho(a)(2-\rho(a))}{\Delta(a)^2},
    &
    \partial_{a_2}G_1(a)
    &=-\frac{\alpha_2}{\Delta(a)^2},                                      \label{eq:mm1-deriv-mean}\\
    \partial_{a_1}G_2(a)
    &=-\frac{2\alpha_1+2\gamma\rho(a)(\rho(a)^2-3\rho(a)+3)}{\Delta(a)^3},
    &
    \partial_{a_2}G_2(a)
    &=-\frac{2\alpha_2}{\Delta(a)^3}.                                    \label{eq:mm1-deriv-var}
\end{align}
Let
$S(a):=(\nabla G(a)+\nabla G(a)^\top)/2$
be the symmetric part of the Jacobian on $A_0$. Its entries are
\begin{align*}
    S_{11}(a)
    &=-\frac{\alpha_1+\gamma\rho(a)(2-\rho(a))}{\Delta(a)^2}, \quad
    S_{22}(a)
    =-\frac{2\alpha_2}{\Delta(a)^3},\\
    S_{12}(a)
    &=-\frac{\alpha_2\Delta(a)+2\alpha_1+2\gamma\rho(a)(\rho(a)^2-3\rho(a)+3)}{2\Delta(a)^3}.
\end{align*}

\begin{lemma}[A sufficient one-sided condition]
\label{lem:mm1-osl}
Suppose that, on the active-side closure,
\begin{equation}
\label{eq:mm1-exact-osl}
\left(
    \frac{\alpha_2\Delta(a)+2\alpha_1+2\gamma\rho(a)(\rho(a)^2-3\rho(a)+3)}{2\Delta(a)^3}
\right)^2
<
\left(
    1+\frac{\alpha_1+\gamma\rho(a)(2-\rho(a))}{\Delta(a)^2}
\right)
\left(
    1+\frac{2\alpha_2}{\Delta(a)^3}
\right).
\end{equation}
Then there exists $c_{\rm os}<1$ such that
$S(a)\preceq c_{\rm os}I$ throughout the active region.
Condition \eqref{eq:mm1-exact-osl} is implied by \eqref{eq:mm1-simple-sufficient}.
\end{lemma}

\proof{Proof.}
Condition \eqref{eq:mm1-exact-osl} is exactly
$\det(I-S(a))>0$. The diagonal entries of $I-S(a)$ are positive, so $I-S(a)\succ0$. By continuity on the compact active-side closure, the smallest eigenvalue is uniformly bounded away from zero. Hence $S(a)\preceq c_{\rm os}I$ for some $c_{\rm os}<1$.

For the simpler sufficient condition, note that $\Delta(a)\ge\Delta_0$ and $0\le\rho(a)\le1$. Moreover,
$0\le\rho(2-\rho)\le1$ and $0\le\rho(\rho^2-3\rho+3)\le1$.
Therefore,
\[
    |S_{12}(a)|
    \le
    \frac{\alpha_2}{2\Delta_0^2}
    +\frac{\alpha_1+\gamma}{\Delta_0^3}
    <1.
\]
The right-hand side of \eqref{eq:mm1-exact-osl} is at least one, so \eqref{eq:mm1-exact-osl} follows. \Halmos
\endproof

\proof{Proof of Proposition~\ref{prop:mm1-mean-var}.}
For every $a\in K$, $0\le\lambda(a)\le\tau_0$, $\nu(a)\ge\nu_0$, and $\Delta(a)\ge\Delta_0$. Hence $0\le\rho(a)\le\rho_0$, and
\[
    G_1(a)\le\frac{\rho_0}{\Delta_0}=\bar m,
    \qquad
    G_2(a)\le\frac{\rho_0(2-\rho_0)}{\Delta_0^2}=\bar v.
\]
Thus $G(K)\subseteq K$.

The map $G$ is continuous across the boundary $\lambda(a)=0$, because both coordinates in \eqref{eq:mm1-response} tend to zero as $\lambda(a)\downarrow0$. It is continuously differentiable on the active and inactive interiors, and the derivatives on each side are uniformly bounded. Along every line segment in $K$, the boundary is crossed at most once. It follows by integrating the a.e. derivative along the segment that $G$ is Lipschitz on $K$.

Under \eqref{eq:mm1-simple-sufficient}, Lemma~\ref{lem:mm1-osl} gives $S(a)\preceq c_{\rm os}I$ on the active region. On the inactive interior, $G$ is constant and its Jacobian is zero, so the same upper bound holds after increasing $c_{\rm os}$ if necessary. For arbitrary $x,y\in K$, the path $a(t)=y+t(x-y)$ crosses the boundary at most once and $G(a(t))$ is absolutely continuous. Hence
\begin{align*}
    \langle G(x)-G(y),x-y\rangle
    &=\int_0^1(x-y)^\top \nabla G(a(t))(x-y)\,dt\\
    &=\int_0^1(x-y)^\top S(a(t))(x-y)\,dt
    \le c_{\rm os}\|x-y\|^2.
\end{align*}
The conclusion follows from Proposition~\ref{prop:Lip}. Since the left-hand side of \eqref{eq:mm1-simple-sufficient} tends to zero as $\Delta_0\to\infty$, the claimed finite threshold $\bar\Delta$ exists. \Halmos
\endproof

\subsection{IPA for Strategic G/G/c queue}\label{app:ggc-ipa-proof}
We first establish finite-horizon pathwise differentiability. For $a\in K^+$, define the pre-sorting vector
\[
    X_n(a):=W_n(a)+S_ne_1-\tau(a,U_n){\bf 1},
    \quad \mbox{ i.e., }
    W_{n+1}(a)=r(X_n(a)^+).
\]

\begin{lemma}[Finite-horizon pathwise derivative]
\label{lem:ggc-pathwise-derivative}
Under Assumption~\ref{ass:GGc-ipa}, for every fixed $a\in K^+$ and finite $n$, the map $a'\mapsto W_n(a')$ is differentiable at $a$ almost surely. The derivative satisfies
\[
    J_{n+1}(a)
    =R_n(a)\left(J_n(a)-{\bf 1}\nabla_a\tau(a,U_n)^\top\right),
\]
where $R_n(a)$ is a partial permutation matrix. In particular,
$\|R_n(a)\|\le1$, and $R_n(a)=0$ whenever $X_n(a)<0$ coordinatewise.
\end{lemma}

\proof{Proof.}
Fix a finite horizon $N$. Work on the probability-one event on which $a'\mapsto\tau(a',U_m)$ is continuously differentiable for $m=0,\ldots,N-1$. We argue by induction.

At time zero, $W_0$ is independent of $a$, so $J_0=0$. Suppose $W_n(a)$ has no positive ties and is differentiable at $a$. Conditional on the past and on $S_n$, each coordinate of $X_n(a)$ is a fixed number minus $\tau(a,U_n)$. Since $\tau(a,U_n)$ has a density and is independent of the past and $S_n$,
\[
    \mathbb P\{X_{n,i}(a)=0\}=0.
\]
Positive ties between coordinates $i,j\ge2$ would imply a positive tie in $W_n(a)$, which has probability zero by induction. A positive tie between coordinate one and coordinate $j\ge2$ would require
$S_n=W_{n,j}(a)-W_{n,1}(a)$; this event has probability zero because $S_n$ is non-atomic and independent of the past. Thus $X_n(a)$ has no zero coordinate and $X_n(a)^+$ has no positive ties almost surely.

At such a point, the map $x\mapsto r(x^+)$ is differentiable. Its derivative is the partial permutation matrix $R_n(a)$ that discards negative coordinates and reorders the positive coordinates. The chain rule gives the stated recursion. A partial permutation matrix has Euclidean operator norm at most one. If every coordinate of $X_n(a)$ is strictly negative, the derivative of the positive-part map is zero, so $R_n(a)=0$. A countable intersection over finite $N$ completes the proof. \Halmos
\endproof

We next state the regenerative stability lemma used to control the augmented process. It is useful beyond the present queueing model.

\begin{lemma}[Regenerative stability of a nonexpansive affine recursion]
\label{lem:reset-affine-L2}
Consider a family of Markov chains $(Y_n,Z_n)$ indexed by $a\in K^+$, with
$Z_{n+1}=\Gamma_nZ_n+\xi_n$.
Suppose there is a common accessible atom $x^\circ=(y^\circ,0)$ and, with
$\sigma_a:=\inf\{n\ge1:(Y_n,Z_n)=x^\circ\}$,
there exist $r_0>1$, $C<\infty$, and $V\ge1$, all independent of $a$, such that
\begin{equation}
\label{eq:reset-exp-return}
    \mathbb E_{a,y,z}[r_0^{\sigma_a}]
    \le CV(y),
    \qquad
    \inf_{a\in K^+}\mathbb P_{a,x^\circ}\{\sigma_a=1\}>0.
\end{equation}
Assume also that, for some $q>2$,
$\|\Gamma_n\|\le1$, $\|\xi_n\|\le B_n$, and $\sup_{a,n}\mathbb E[B_n^q\mid\mathcal F_n]\le C$ a.s.
Then, each augmented chain has a unique invariant distribution $\Pi_a$, and there are constants $C<\infty$ and $\beta\in(0,1)$, uniform in $a$, such that
\begin{align}
    \sup_{n\ge0}\mathbb E_{a,y,z}\|Z_n\|^2
    &\le CV(y)(1+\|z\|^q),                                      \label{eq:reset-L2}\\
    \left\|
        \mathbb E_{a,y,z}g(Y_n,Z_n)-\Pi_ag
    \right\|
    &\le CL_gV(y)(1+\|z\|^q)\beta^n                              \label{eq:reset-fconv}
\end{align}
for every measurable vector-valued $g$ satisfying
$\|g(y,z)\|\le L_g\|z\|$.
\end{lemma}

\proof{Proof.}
Let $p:=q/(q-2)>1$. Choose $s>1$ and $r_1,r_2$ such that
$s^p<r_1<r_2<r_0$.
Before the return time $\sigma_a$, nonexpansiveness gives
\[
    \|Z_k\|
    \le\|z\|+\sum_{j=0}^{k-1}B_j.
\]
Also, $\sum_{k=0}^{m}s^k\le C_ss^m$. Hence
\begin{align}\label{eq:F1}
    \sum_{k=0}^{\sigma_a}s^k\|Z_k\|^2
    &\le
    C s^{\sigma_a}\|z\|^2
    +C s^{\sigma_a}\sigma_a
      \sum_{j=0}^{\sigma_a-1}B_j^2.                    
\end{align}
The expectation of the first term is bounded by
$CV(y)(1+\|z\|^q)$ because $s<r_0$ and $q>2$.

For the second term, Tonelli's theorem and H\"{o}lder's inequality with exponents $q/2$ and $p=q/(q-2)$ give
\begin{align}
\mathbb E_{a,y,z}\left[
    s^{\sigma_a}\sigma_a
    \sum_{j=0}^{\sigma_a-1}B_j^2
\right]
=\sum_{j\ge0}
\mathbb E_{a,y,z}\left[
    B_j^2s^{\sigma_a}\sigma_a 1_{\{j<\sigma_a\}}
\right] 
&\le C\sum_{j\ge0}
\left(
    \mathbb E_{a,y,z}\left[
        s^{p\sigma_a}\sigma_a^p1_{\{j<\sigma_a\}}
    \right]
\right)^{1/p}.                                                \label{eq:F2}
\end{align}
Because $s^p<r_1$, the polynomial factor can be absorbed into a slightly larger exponential: there is $C<\infty$ such that
$\sigma_a^ps^{p\sigma_a}\le Cr_1^{\sigma_a}$. On $\{j<\sigma_a\}$,
\[
    r_1^{\sigma_a}
    =r_2^{\sigma_a}\left(\frac{r_1}{r_2}\right)^{\sigma_a}
    \le r_2^{\sigma_a}\left(\frac{r_1}{r_2}\right)^j.
\]
Since $r_2<r_0$, \eqref{eq:reset-exp-return} implies
\[
    \mathbb E_{a,y,z}\left[
        s^{p\sigma_a}\sigma_a^p1_{\{j<\sigma_a\}}
    \right]
    \le CV(y)\left(\frac{r_1}{r_2}\right)^j.
\]
Substitution into \eqref{eq:F2}, together with $V\ge1$, shows that the series over $j$ is geometric. Combining this estimate with \eqref{eq:F1} proves
\begin{equation}
\label{eq:cycle-weighted-moment}
    \mathbb E_{a,y,z}\left[
        \sum_{k=0}^{\sigma_a}s^k(1+\|Z_k\|^2)
    \right]
    \le CV(y)(1+\|z\|^q).
\end{equation}

The atom occupation formula now defines an invariant probability measure, and \eqref{eq:cycle-weighted-moment} gives it a finite second $Z$-moment. A renewal decomposition at visits to the atom yields \eqref{eq:reset-L2}. Finally, \eqref{eq:cycle-weighted-moment}, accessibility of the atom, and the one-step aperiodicity bound in \eqref{eq:reset-exp-return} are the standard Kendall-set conditions for geometric convergence in the weighted norm generated by $1+\|z\|$. Applying the corresponding regenerative $f$-ergodic theorem gives \eqref{eq:reset-fconv}. Uniformity follows because all cycle bounds and the aperiodicity constant are uniform in $a$.\Halmos
\endproof

\proof{Proof of Theorem~\ref{thm:GGc_IPA_valid}.}
We divide the proof into four steps.

\paragraph{Step 1: regenerative stability of the augmented process.}
By Lemma~\ref{lem:ggc-pathwise-derivative},
\[
    J_{n+1}(a)=R_n(a)J_n(a)+\xi_n(a),
    \qquad
    \xi_n(a):=-R_n(a){\bf 1}\nabla_a\tau(a,U_n)^\top.
\]
Hence, $\|R_n(a)\|\le1$ and $\|\xi_n(a)\|_F\le\sqrt c\,B_T(U_n)$.
The moment condition in Assumption~\ref{ass:GGc-ipa} therefore verifies the additive-noise condition in Lemma~\ref{lem:reset-affine-L2}.

It remains to verify the exponential return condition. Lemma~\ref{lem:ggc-drift} gives a geometric drift for the workload process. Choose $\lambda\in(\eta,1)$ and enlarge the bounded set $C_0$ so that
\[
    P_aV(w)\le\lambda V(w),
    \qquad w\notin C_0,
\]
uniformly in $a\in K^+$. From every $w\in C_0$, the reset event in Lemma~\ref{lem:ggc-reset-coupling} has probability at least $p_\star>0$ and sends the workload to zero. On that event, $R_n(a)=0$, so it simultaneously sends the sensitivity to zero, regardless of its current value. Thus $(0,0)$ is an accessible atom of the augmented chain. The geometric drift plus the uniform minorization to this atom yields constants $r_0>1$ and $C<\infty$ such that
$\mathbb E_{a,w,j}[r_0^{\sigma_a}]\le CV(w)$,
uniformly in $a\in K^+$. Starting from $(0,0)$, the same reset event has probability at least $p_\star$, giving a uniform one-step return probability.

Lemma~\ref{lem:reset-affine-L2} therefore applies. Let $\Pi_a$ be the invariant distribution of $(W_n(a),J_n(a))$. Its workload marginal is invariant for $P_a$ and hence equals $\pi_a$ by uniqueness. Moreover,
\begin{align}
    \sup_{a\in K^+}\sup_{n\ge0}
    \mathbb E_{w,j}\|J_n(a)\|_F^2
    &\le CV(w)(1+\|j\|_F^q),                                    \label{eq:queue-J-L2}\\
    \left\|
        \mathbb E_{w,j}[\nabla h(W_n(a))J_n(a)]
        -\mathbb E_{\Pi_a}[\nabla h(W)J]
    \right\|_F
    &\le CL_hV(w)(1+\|j\|_F^q)\beta^n.                          \label{eq:queue-summand-bias}
\end{align}

\paragraph{Step 2: finite-horizon differentiation.}
Define
$G_n(a):=\mathbb E_0[h(W_n(a))]$.
For every fixed $n$, Lemma~\ref{lem:ggc-pathwise-derivative} and the chain rule give the pathwise derivative
$\nabla h(W_n(a))J_n(a)$. Starting from $J_0=0$, recursion \eqref{eq:ggc-sensitivity-recursion} gives
\begin{equation}
\label{eq:finite-J-envelope}
    \sup_{a\in K^+}\|J_n(a)\|_F
    \le\sqrt c\sum_{m=0}^{n-1}B_T(U_m).
\end{equation}
The same envelope, together with the Lipschitz property of $h$, dominates local difference quotients. Since $q>2$, the right-hand side is integrable. Dominated differentiation therefore gives
\begin{equation}
\label{eq:finite-horizon-gradient}
    \nabla G_n(a)=\mathbb E_0[\nabla h(W_n(a))J_n(a)].
\end{equation}
The density and non-atomicity assumptions also imply that, for almost every primitive realization, the active set and sorting order are locally constant near a fixed $a$. Together with \eqref{eq:finite-J-envelope}, dominated convergence shows that $\nabla G_n$ is continuous. Thus $G_n$ is continuously differentiable on compact subsets of $K^+$.

\paragraph{Step 3: identification of the stationary derivative.}
Fix $a_0\in K$ and choose a compact convex neighborhood $K_0$ of $a_0$ contained in $K^+$. The same drift and reset argument used above gives uniform $V$-geometric convergence of the frozen workload chain. Since $\|h\|\le C_hV$,
\begin{equation}
\label{eq:Gn-uniform-G}
    \sup_{a\in K_0}\|G_n(a)-G(a)\|
    \le C\rho^n.
\end{equation}
On the other hand, \eqref{eq:queue-summand-bias} with $(w,j)=(0,0)$ gives
\begin{equation}
\label{eq:DGn-uniform-limit}
    \sup_{a\in K_0}
    \left\|
        \nabla G_n(a)-L(a)
    \right\|_F
    \le CL_h\beta^n,
    \qquad
    L(a):=\mathbb E_{\Pi_a}[\nabla h(W)J].
\end{equation}
Thus $G_n\to G$ and $\nabla G_n\to L$ uniformly on $K_0$. The standard uniform-limit-of-derivatives theorem, applied coordinatewise on the convex set $K_0$, implies that $G$ is continuously differentiable and $\nabla G=L$. Since $a_0$ was arbitrary,
\[
    \nabla G(a)=\mathbb E_{\Pi_a}[\nabla h(W)J],
    \qquad a\in K.
\]
This proves the steady-state IPA identity rather than assuming it.

\paragraph{Step 4: mean-squared error.}
Let
$Y_n(a):=\nabla h(W_n(a))J_n(a)$.
For independent replications and initialization $(W_0,J_0)=(0,0)$, \eqref{eq:queue-J-L2} gives
\[
    \mathbb E\left\|
        \widehat{\nabla G}_{N,n}(a)-\mathbb EY_n(a)
    \right\|_F^2
    =\frac1N\mathbb E\left\|Y_n(a)-\mathbb EY_n(a)\right\|_F^2
    \le\frac1N\mathbb E\|Y_n(a)\|_F^2
    \le\frac{CL_h^2}{N}.
\]
Equation \eqref{eq:queue-summand-bias} and the derivative identity give
\[
    \|\mathbb EY_n(a)-\nabla G(a)\|_F^2
    \le CL_h^2\beta^{2n}.
\]
The centered Monte Carlo term has zero mean, so the cross term vanishes. Adding variance and squared bias proves
\[
    \mathbb E\left[
        \|\widehat{\nabla G}_{N,n}(a)-\nabla G(a)\|_F^2
    \right]
    \le\frac{CL_h^2}{N}+CL_h^2\beta^{2n}.
\]
Choosing $n$ so that $\beta^{2n}\lesssim N^{-1}$ completes the proof.\Halmos
\endproof

\section{Opinion Dynamics Appendix}
Define
\[
\mathcal P_{\lambda}(S)
:=
\left\{
\mu\in\mathcal P(S):
\mu(\mathbb R\times\{\tau\})=\lambda_\tau
\text{ for all }\tau\in\Theta
\right\},
\]
and
\[
\mathcal P_{\lambda,1}(S)
:=
\left\{
\mu\in\mathcal P_{\lambda}(S):
\int_S |y|\,\mu(dy,d\theta)<\infty
\right\}.
\]
Thus, the type marginal remains fixed, while the belief coordinate evolves.

Let \(Z:=\mathbb R^m\) denote the space of type-level mean-belief vectors,
endowed with the coordinatewise order. 
We endow
\(S\) with the partial order
$(y,\tau)\preceq (y',\tau')$ iff 
$\tau=\tau'$ and $\ y\le y'$.
Thus, states are compared only within type. Let \(\mathcal G\) denote the set
of bounded continuous functions \(g:S\to\mathbb R\) that are increasing with
respect to \(\preceq\); equivalently, for every \(\tau\in\Theta\), the map
\(y\mapsto g(y,\tau)\) is nondecreasing. For
\(\mu,\nu\in\mathcal P_{\lambda}(S)\), write
$\mu\preceq_{\mathrm{SD}_y}\nu$ if and only if
$\int_S g\,d\mu
    \le
    \int_S g\,d\nu$ 
for all $g\in\mathcal G$.
Because \(\Theta\) is finite and discrete, this order is equivalent to
first-order stochastic dominance of the conditional belief distributions
within every type: \(\mu\preceq_{\mathrm{SD}_y}\nu\) if and only if, for
each \(\tau\in\Theta\), the type-\(\tau\) belief distribution under \(\nu\)
first-order stochastically dominates that under \(\mu\).

Before we prove Theorem~\ref{thm:opinion-dynamics}, we first prove an auxiliary lemma.

\begin{lemma}[Continuity of the frozen stationary distribution]
\label{lem:mu_continuous}
Let \(S\) be a Polish space, let \(Z\) be a metric space, and let
\(\mathcal C\subseteq\mathcal P(S)\) be weakly closed. For each \(a\in Z\),
let \(P_a\) be a Markov kernel on \(S\). Assume that the family
\(\{P_a:a\in Z\}\) is weakly continuous in the following sense: whenever
\(x_n\to x\) in \(S\) and \(a_n\to a\) in \(Z\),
$P_{a_n}(x_n,\cdot)\Rightarrow P_a(x,\cdot)$.
Suppose that, for every \(a\in Z\), \(P_a\) admits a unique invariant
probability measure within \(\mathcal C\), denoted by \(\pi_a\). Finally,
assume that whenever \(a_n\to a\), the sequence
\(\{\pi_{a_n}\}_{n\ge1}\) is tight. Then
\[
    a_n\to a
    \quad\Longrightarrow\quad
    \pi_{a_n}\Rightarrow\pi_a .
\]
\end{lemma}


\proof{Proof.}
Fix a sequence \(a_n\to a\) in \(Z\). 
By assumption, the sequence \(\{\pi_{a_n}\}_{n\ge1}\) is tight. Hence, by
Prokhorov's theorem, every subsequence has a further weakly convergent
subsequence. Let \(\{\pi_{a_{n_k}}\}_{k\ge1}\) be any weakly convergent
subsequence, and write
$\pi_{a_{n_k}}\Rightarrow \vartheta$
for its weak limit. 

Because \(\mathcal C\) is weakly closed, \(\vartheta\in\mathcal C\).
The argument below shows that \(\vartheta\) is invariant for \(P_a\).
By uniqueness of the invariant probability measure within \(\mathcal C\),
it follows that \(\vartheta=\pi_a\).

Let \(g\in C_b(\mathcal S)\). For each \(k\), define
\[
    h_k(x):=(P_{a_{n_k}}g)(x)
    =
    \int_{\mathcal S}g(y)P_{a_{n_k}}(x,dy)
\quad \mbox{ and }\quad
    h(x):=(P_ag)(x)
    =
    \int_{\mathcal S}g(y)P_a(x,dy).
\]
The weak continuity of the kernels implies that
$h_k(x_k)\to h(x)$ whenever $x_k\to x$.
Moreover, since $|h_k(x)|\le \|g\|_\infty$ for all $x\in S$, $k\geq 1$, the functions \(h_k\) are uniformly bounded.
Therefore, by the convergence theorem for varying measures
\citep[Theorem 3.5]{serfozo1982convergence}, the weak convergence
\(\pi_{a_{n_k}}\Rightarrow \vartheta\) implies
\begin{equation}
\label{eq:varying-measures-proof}
    \int_{\mathcal S} h_k(x)\,\pi_{a_{n_k}}(dx)
    \longrightarrow
    \int_{\mathcal S} h(x)\,\vartheta(dx).
\end{equation}

Using the invariance of \(\pi_{a_{n_k}}\) for \(P_{a_{n_k}}\), we have
\[
    \int_{\mathcal S} g(x)\,\pi_{a_{n_k}}(dx)
    =
    \int_{\mathcal S} (P_{a_{n_k}}g)(x)\,\pi_{a_{n_k}}(dx)
    =
    \int_{\mathcal S} h_k(x)\,\pi_{a_{n_k}}(dx).
\]
Taking limits and using \eqref{eq:varying-measures-proof}, we obtain
\[
    \lim_{k\to\infty}
    \int_{\mathcal S} g(x)\,\pi_{a_{n_k}}(dx)
    =
    \int_{\mathcal S} (P_ag)(x)\,\vartheta(dx).
\]
On the other hand, since \(\pi_{a_{n_k}}\Rightarrow \vartheta\) and
\(g\in C_b(\mathcal S)\),
\[
    \lim_{k\to\infty}
    \int_{\mathcal S} g(x)\,\pi_{a_{n_k}}(dx)
    =
    \int_{\mathcal S} g(x)\,\vartheta(dx).
\]
Combining the last two displays gives
\[
    \int_{\mathcal S} g(x)\,\vartheta(dx)
    =
    \int_{\mathcal S} (P_ag)(x)\,\vartheta(dx)
    =
    \int_{\mathcal S}\int_{\mathcal S}g(y)P_a(x,dy)\,\vartheta(dx).
\]
Since this identity holds for every \(g\in C_b(\mathcal S)\), we have
$\vartheta=\vartheta P_a$.
Thus, \(\vartheta\) is an invariant probability measure for \(P_a\). By the assumed
uniqueness of the invariant probability measure of \(P_a\), it follows that
$\vartheta=\pi_a$.

We have shown that every weakly convergent subsequence of
\(\{\pi_{a_n}\}_{n\ge1}\) has limit \(\pi_a\). Since the full sequence is tight,
this implies
$\pi_{a_n}\Rightarrow \pi_a$.
Therefore the map \(a\mapsto \pi_a\) is weakly continuous.
\Halmos
\endproof

\proof{Proof of Theorem~\ref{thm:opinion-dynamics}}
We verify that the self-consistency map $G$
is continuous and coordinatewise nondecreasing on \(K\), and then apply Proposition~\ref{prop:monotone}.

Fix \(u,v\in K\) with \(u\le v\) coordinatewise. Since \(W\) has nonnegative
entries,
\[
    S_\tau(u)=(Wu)_\tau \le (Wv)_\tau=S_\tau(v),
    \quad \tau\in\Theta .
\]
Let \(g:S\to\mathbb R\) be bounded, continuous, and increasing with respect to
the within-type order \(\preceq\). That is, for each \(\tau\), the map
\(y\mapsto g(y,\tau)\) is nondecreasing. If
\((y,\tau)\preceq (y',\tau')\), then \(\tau=\tau'\) and \(y\le y'\). By the
monotonicity of \(w_\tau\) in its first two arguments,
\[
    w_\tau(y,S_\tau(u),\epsilon)
    \le
    w_\tau(y',S_\tau(v),\epsilon),
    \quad \epsilon\in\mathbb R .
\]
Therefore,
\begin{align*}
    (P_u g)(y,\tau)
    &=
    \int
    g\!\left(
        w_\tau(y,S_\tau(u),\epsilon),\tau
    \right)
    F_\tau(d\epsilon) \\
    &\le
    \int
    g\!\left(
        w_\tau(y',S_\tau(v),\epsilon),\tau
    \right)
    F_\tau(d\epsilon) 
    =
    (P_v g)(y',\tau).
\end{align*}
In particular, for each fixed \(a\in K\), the function \(P_ag\) is bounded
and increasing with respect to \(\preceq\). Moreover, the weak continuity of
\(P_a\) implies that \(P_ag\) is continuous. Hence \(P_ag\in\mathcal G\).
If \(u\le v\), then \(P_ug\le P_vg\) pointwise on \(S\).

We next show that this one-step monotonicity propagates to the stationary
distributions. Let \(\nu_0\in\mathcal P_{\lambda,1}(S)\) be any initial
distribution with type marginal \(\lambda\). We prove by induction that for every \(n\ge0\),
$\nu_0 P_u^n \preceq_{\mathrm{SD}_y} \nu_0 P_v^n$.
For \(n=0\), the two measures are equal. Suppose the claim holds at time \(n\).
For any bounded continuous increasing \(g\),
\[
    \nu_0 P_u^{n+1}(g)
    =
    \nu_0 P_u^n(P_u g) 
    \le
    \nu_0 P_v^n(P_u g) 
    \le
    \nu_0 P_v^n(P_v g)
    =
    \nu_0 P_v^{n+1}(g).
\]
The first inequality uses the induction hypothesis and the fact that \(P_u g\)
is increasing; the second inequality uses \(P_u g\le P_v g\). 

By assumption, for each fixed \(a\in K\), the frozen chain with kernel \(P_a\)
converges weakly to \(\pi_a\). Hence
$\nu_0 P_u^n \Rightarrow \pi_u$ and
$\nu_0 P_v^n \Rightarrow \pi_v$.
Passing to the limit in the preceding inequality for bounded continuous
increasing \(g\) gives
$\int_S g\,d\pi_u \le \int_S g\,d\pi_v$.
Thus, \(\pi_u\preceq_{\mathrm{SD}_y}\pi_v\). 

We now show that \(H\) preserves this order. Fix a type \(\tau\in\Theta\). For
\(M>0\), define the truncation
$T_M(y):=(-M)\vee (y\wedge M)$,
and let
$g_{M,\tau}(y,\theta)
    :=
    T_M(y)\mathbf 1\{\theta=\tau\}$.
The function \(g_{M,\tau}\) is bounded, continuous, and increasing with respect
to the within-type order. Therefore,
$\int_S g_{M,\tau}\,d\pi_u
    \le
    \int_S g_{M,\tau}\,d\pi_v$.
Letting \(M\to\infty\) and using the finite first moments of \(\pi_u\) and
\(\pi_v\), we obtain
\[
    \int_S y\mathbf 1\{\theta=\tau\}\,\pi_u(dy,d\theta)
    \le
    \int_S y\mathbf 1\{\theta=\tau\}\,\pi_v(dy,d\theta).
\]
Dividing by \(\lambda_\tau>0\) gives
$H(\pi_u)_\tau \le H(\pi_v)_\tau$.
Since \(\tau\) was arbitrary,
\[
    G(u)=H(\pi_u)\le H(\pi_v)=G(v),
\]
so \(G\) is coordinatewise nondecreasing on \(K\).

It remains to prove the continuity of \(G\). Let \(a_n\to a\) in \(K\). The uniform-integrability condition implies
that \(\{\pi_{a_n}\}_{n\ge1}\) is tight. Indeed,
\[
    \sup_n
    \pi_{a_n}\bigl(\{|y|>R\}\times\Theta\bigr)
    \le
    \frac1R
    \sup_{u\in K}
    \int_S
        |y|\mathbf 1\{|y|>R\}\,
        \pi_u(dy,d\theta)
    \rightarrow 0.
\]
Since \(\Theta\) is finite, \([-R,R]\times\Theta\) is compact.
Applying Lemma~\ref{lem:mu_continuous} with
\(\mathcal C=\mathcal P_\lambda(S)\) therefore gives
$\pi_{a_n}\Rightarrow\pi_a$.
We next show that this weak convergence implies convergence of the type-level means.
Fix \(\tau\in\Theta\), and define
\[
    \phi_\tau(y,\theta)
    :=
    \frac{1}{\lambda_\tau}y\mathbf 1\{\theta=\tau\}
    \quad \mbox{ and } \quad
    \phi_{\tau,M}(y,\theta)
    :=
    \frac{1}{\lambda_\tau}T_M(y)\mathbf 1\{\theta=\tau\}
\]
for \(M>0\).
Weak convergence gives
$\int_S \phi_{\tau,M}\,d\pi_{a_n}
    \longrightarrow
    \int_S \phi_{\tau,M}\,d\pi_a$
for every fixed  M.
Moreover,
\[
    |\phi_\tau(y,\theta)-\phi_{\tau,M}(y,\theta)|
    \le
    \frac{1}{\lambda_\tau}|y|\mathbf 1\{|y|>M\}.
\]
Therefore,
\begin{align*}
    \limsup_{n\to\infty}
    \left|
        G(a_n)_\tau-G(a)_\tau
    \right|
    &=
    \limsup_{n\to\infty}
    \left|
        \int_S \phi_\tau\,d\pi_{a_n}
        -
        \int_S \phi_\tau\,d\pi_a
    \right| \\
    &\le
    \frac{2}{\lambda_\tau}
    \sup_{\tilde a\in K}
    \int_S |y|\mathbf 1\{|y|>M\}\,\pi_{\tilde a}(dy,d\theta).
\end{align*}
Taking \(M\to\infty\) and using the uniform integrability assumption yields
$G(a_n)_\tau\to G(a)_\tau$.
Since \(\Theta\) is finite, this holds for every coordinate, and hence
\(G(a_n)\to G(a)\). Thus \(G\) is continuous on \(K\).

We have shown that \(G\) is continuous and coordinatewise nondecreasing on
\(K\). The boundary conditions imply that \(G\) maps \(K\) into itself. Indeed,
for any \(a\in K\),
\[
    \underline a
    \le
    G(\underline a)
    \le
    G(a)
    \le
    G(\overline a)
    \le
    \overline a,
\]
where the middle inequalities follow from the monotonicity of \(G\). Therefore
\(G(K)\subseteq K\).

All assumptions of Proposition~\ref{prop:monotone} are now satisfied. \Halmos
\endproof

\end{document}